\documentclass[letterpaper,11pt]{article}

\usepackage{setspace}

\usepackage{tabularx}

\usepackage{setspace}
\usepackage{multirow}

\newif\ifecta
\ectafalse        

\ifecta
  \usepackage[margin=1.25in, includefoot, footskip=0.5in]{geometry} 
\else
  \usepackage[margin=1in, includefoot, footskip=0.5in]{geometry}    
\fi
\usepackage{xcolor}
\usepackage{comment}
\usepackage{environ}
\usepackage{booktabs}
\usepackage{array} 
\usepackage{float} 

\newif\ifshowpublic
\newif\ifshowprivate
\newif\ifshowredundant

\showpublictrue
\showprivatetrue
\showredundantfalse

\ifshowpublic\includecomment{public}\else\excludecomment{public}\fi
\ifshowprivate\includecomment{private}\else\excludecomment{private}\fi

\definecolor{redundantbrown}{RGB}{150,75,0}
\NewEnviron{redu}{%
  \ifshowredundant
    \par\smallskip
    {\color{redundantbrown}\textbf{[Redundant]} \BODY}%
    \par\smallskip
  \fi
}

\usepackage{amsmath,amssymb,amsfonts,amsthm}
\usepackage{mathtools}      
\usepackage{bm}             
\usepackage{bbm}            
\usepackage{booktabs,tabularx,longtable}
\usepackage[inline]{enumitem}
\usepackage{float}
\usepackage{natbib}
\usepackage{graphicx}
\graphicspath{{figures/}{figures/}{../}{./}}
\makeatletter\def\input@path{{./}{../}}\makeatother
\usepackage{tikz}
\usetikzlibrary{arrows.meta}   
\usetikzlibrary{positioning}   
\usetikzlibrary{calc}          
\usepackage{xcolor}            
\usepackage[normalem]{ulem}    
\newcommand{\rev}[1]{#1}                
\long\def\bluerev#1{{\color{black}#1}}  
\newcommand{\redrev}[1]{{\color{black}#1}}  
\newcommand{\greenrev}[1]{{\color{black}#1}}  
\newcommand{\ridgerev}[1]{#1}  

\newif\ifshowgreenmarkup
\showgreenmarkupfalse    

\ifshowgreenmarkup
  \long\def\greenrev#1{{\color{green!55!black}#1}}
  \newcommand{\greendel}[1]{{\color{green!55!black}\sout{#1}}}
\else
  \long\def\greenrev#1{#1}
  \long\def\greendel#1{}
\fi

\setlist[description]{style=nextline,font=\normalfont\bfseries,labelsep=0.5em,leftmargin=0pt}
\setlist[enumerate]{leftmargin=*,labelsep=0.5em,itemsep=0.3em}
\setlist[itemize]{leftmargin=*,labelsep=0.5em,itemsep=0.3em}

\allowdisplaybreaks
\numberwithin{equation}{section}

\newcommand{\E}{\mathbb{E}}

\newcommand{\Pp}{\mathbb{P}}

\newcommand{\R}{\mathbb{R}}
\newcommand{\1}{\mathbbm{1}}

\newcommand{\indep}{\perp\!\!\!\perp}
\newcommand{\Var}{\operatorname{Var}}
\newcommand{\Cov}{\operatorname{Cov}}

\newcommand{\argmin}{\mathop{\mathrm{argmin}}}

\newcommand{\pto}{\stackrel{p}{\to}}

\newcommand{\Ucal}{\mathcal{U}}
\newcommand{\Dcal}{\mathcal{D}}
\newcommand{\Ncal}{\mathcal{N}}

\newcommand{\op}{\mathrm{op}}

\providecommand{\Acal}{\mathcal{A}}
\providecommand{\Hcal}{\mathcal{H}}
\providecommand{\Vcal}{\mathcal{V}}
\providecommand{\Wcal}{\mathcal{W}}

\providecommand{\Ical}{\mathcal{I}}

\makeatletter
\@ifundefined{theorem}{\newtheorem{theorem}{Theorem}[section]}{}
\@ifundefined{lemma}{\newtheorem{lemma}[theorem]{Lemma}}{}
\@ifundefined{proposition}{}{}
\@ifundefined{assumption}{\newtheorem{assumption}[theorem]{Assumption}}{}
\@ifundefined{hypothesis}{\newtheorem{hypothesis}[theorem]{Hypothesis}}{}

\theoremstyle{remark}
\@ifundefined{remark}{\newtheorem{remark}[theorem]{Remark}}{}

\theoremstyle{plain}
\makeatother
\makeatletter
\let\oldremark\remark
\let\oldendremark\endremark

\renewenvironment{remark}[1][]
  {%
    \pushQED{\qed}%
    \oldremark[#1]%
  }
  {%
    \popQED\oldendremark
  }
\makeatother
\makeatletter
\@ifundefined{definition}{\newtheorem{definition}[theorem]{Definition}}{}
\@ifundefined{corollary}{\newtheorem{corollary}[theorem]{Corollary}}{}
\makeatother

\makeatletter
\theoremstyle{definition}
\@ifundefined{definition}{}{}
\@ifundefined{example}{\newtheorem{example}[theorem]{Example}}{}
\@ifundefined{algorithm}{\newtheorem{algorithm}[theorem]{Algorithm}}{}
\@ifundefined{corollary}{\newtheorem{corollary}[theorem]{Corollary}}{}
\theoremstyle{plain}
\makeatother

\newcommand{\examplecont}[2]{%
  \par\smallskip
  \noindent\textbf{Example~\ref{#1} (#2), continued.}\ \ignorespaces}

\title{A Design-Based Approach to Testing and Inference in (Quasi-)Experiments with Spillovers%
  \thanks{\setlength{\baselineskip}{4mm}I thank Isaiah Andrews, Jesse Shapiro, Elie Tamer, and Davide Viviano for guidance and generous support. I also thank Raj Chetty for early feedback and encouragement. I am grateful to the authors of \citet{egger2022general} and \citet{muralidharan2023generaleq}, especially Ted Miguel and Paul Niehaus, for providing additional data and guidance for replication. I also benefited from helpful comments by Aristotelis Epanomeritakis, Tilman Graff, Gabriel Kreindler, Andreas Petrou-Zeniou, Yuya Sasaki, Neil Shephard, Rahul Singh, and Shengbin Wei, as well as from discussions with participants at the Harvard Econometrics Workshop and NEEPC~2026. I acknowledge support from the Harvard Griffin Fund in Economics. All remaining errors are mine.\smallskip}}

\author{
Yechan Park\thanks{\setlength{\baselineskip}{4mm}Department of Economics, Harvard University. Littauer Center, 1805 Cambridge St, Cambridge, MA~02138. Email: \texttt{yechanpark@fas.harvard.edu}}
}
\date{First version: July~9, 2026\\[2pt]This version: August~18, 2026}

\usepackage[hidelinks,bookmarksnumbered]{hyperref}
\usepackage{bookmark}
\begin{document}

\maketitle

\begin{abstract}

Economic policies rarely affect only their direct targets. To study these spillovers, researchers summarize who else was treated with a simple exposure measure, such as the share of treated neighbors within a radius. But for many settings, economic theory provides little guidance on choosing the functional form (e.g., ring) of that measure or its parameters (e.g., radius). We show that the data can inform both choices. Correctly specified exposure measures imply orthogonality conditions that can be used for both estimation and testing. We establish consistency and asymptotic normality of the resulting estimator under spatial and network dependence in a design-based framework, with all randomness arising from treatment assignment. We then characterize the efficient moment conditions. Applied to two large-scale anti-poverty programs, the framework supports some prior radius estimates but rejects others. In the latter case, the revised radius yields substantively different policy-effect estimates.


\end{abstract}

\newpage

\section{Introduction}
\label{sec:intro}

Economic policies can generate effects far beyond their direct targets. Because economic agents are embedded in spatial, market, and social networks, shocks may propagate through prices, local labor markets, commuting and shopping patterns, and social interactions. Classical analyses of social multipliers in economics emphasize that such spillovers can amplify or attenuate policy impacts in important ways (e.g., \citealp{glaeser2003social,mani2013poverty}). More recently, a growing body of experimental and quasi-experimental work in development and spatial economics has made spillovers the primary object of interest, including studies of large-scale cash transfers, public works, and transport infrastructure \citep[e.g.][]{miguel2004worms,donaldson2016railroads,egger2022general,muralidharan2023generaleq,franklin2024urban,walker2024slack}.


Despite this interest, spillovers pose econometric challenges. Unlike the standard \redrev{no-interference} setting (where each unit's outcome depends only on its own treatment), spillovers mean that unit
$i$'s outcome can depend on many, or even all units' treatments. Whereas under the standard assumption each unit has only a small number of potential outcomes (e.g., two in the binary-treatment case), with spillovers the number of potential outcomes can be very large (e.g., $2^n$ in the binary case), since each of the possible assignments can in principle produce a different outcome for that unit.

Applied work typically makes the problem tractable by replacing the full assignment vector \(W\) with a low-dimensional exposure mapping \(g(W;X_i,\theta)\) \citep{Manski2013TreatmentResponse,AronowSamii2017}\footnote{An alternative is to restrict interference to small, known groups (partial interference; e.g., households or dorm rooms), the special case in which the exposure map reads only own-group assignments. We focus on the general case because our applications feature market- and space-mediated spillovers that may cross such groupings.}. The vector \(W\) collects all units' treatment statuses, \(X_i\) contains baseline covariates such as networks or commuting flows, and \(\theta\) parameterizes features of the mapping, such as a distance radius or elasticity parameter. Examples include ring or neighborhood averages \citep{miguel2004worms,egger2022general} and smooth spatial-decay kernels such as gravity or market-access measures \citep{redding2004economic,donaldson2016railroads,franklin2024urban}.

This approach raises two key challenges. The first is the choice of the map structure \(g\left(W ; X_i, \theta\right)\). Often economic theory does not provide an exact functional form for this map beyond broad guidance, such as the idea that spillovers should decay with distance (in space, in a network, or along trade and commuting links). Reflecting such ambiguity, even in closely related empirical settings, different researchers make very different choices among these constructions. For example, in large-scale anti-poverty programs in low-income countries, \redrev{researchers studying settings with very similar institutional features choose different exposure map functional forms} \citep[e.g.][]{egger2022general,franklin2024urban,walker2024slack}.\footnote{\citet{egger2022general} study large cash transfers in rural Kenya and proxy local spillovers by per-capita transfers within distance bands around each market (``distance buffers''). \citet{franklin2024urban} evaluate an urban public-works program in Addis Ababa and, to capture spillovers across neighborhoods, combine a spatial-equilibrium model with commuting data so that exposure is mediated by the commuting network rather than by geographic distance alone. In related Kenyan data, \citet{walker2024slack} construct gravity-style ``shopping network'' exposure indices that weight village-level activity by baseline shopping flows from villages to markets, and compare the performance of these network-based measures to buffer-style measures of spatial exposure in explaining cross-market inflation \citep[][Figure~A.17]{walker2024slack}.} Even within a single paper, researchers often report multiple exposure constructions and informally compare them \citep{donaldson2016railroads,muralidharan2023generaleq}.\footnote{For instance, \citet{muralidharan2023generaleq} take as their main specification a ring-based exposure measure, given by the treated share of units within a 20 km radius of each location, but also show in their Appendix~B that similar results obtain when exposure is measured using a market-access-style kernel. Likewise, \citet{donaldson2016railroads} summarize the impact of the nineteenth-century U.S.\ railroad network through a market-access index, but Appendix Table~3 reports estimates under a range of \redrev{alternative exposure maps}: they redefine market access using different measures of ``market size'' (population versus county wealth) and restrict the set of trading partners (for example, only counties beyond certain distance thresholds, only urban counties, only major cities, or only New York City).}

Second, even given a particular functional form for the exposure mapping, empirical researchers must still select \(\theta\) and account for the uncertainty this choice introduces into downstream policy estimates. For example, \citet{egger2022general} note that they ``had no a priori knowledge of the relevant distances over which general equilibrium effects might operate\ldots''. In practice, applied papers often report results for several plausible values of \(\theta\), or select a preferred specification using model-selection criteria or informal diagnostics \citep[e.g.][]{currie2015toxicplants,donaldson2016railroads,egger2022general}. Yet standard errors and confidence intervals typically condition on this choice, treating \(\theta\) as fixed rather than as a source of estimation uncertainty.\footnote{In \citet{egger2022general}, local price changes are regressed on exposure measures defined over multiple distance buffers, with the outer radius chosen by a Schwarz/Bayesian information criterion; the selected 0--2 km band is then used as the main exposure definition in the reduced-form analysis \citep[][Section~4.2]{egger2022general}. \citet{currie2015toxicplants} estimate effects within several distance bands around polluting manufacturing plants and report alternative distances, but focus on a preferred near/far comparison, 0--1 versus 1--2 miles, guided by atmospheric dispersion evidence \citep[][pp.~278, 299--312]{currie2015toxicplants}. \citet{donaldson2016railroads} recompute market access under alternative trade elasticities as robustness checks \citep[][pp.~843--844, Appendix Table~3]{donaldson2016railroads}. In each case, variation across \(\theta\) is treated as specification analysis rather than propagated into sampling uncertainty.}

To address these challenges in choosing and calibrating exposure mappings, this paper develops a design-based framework for testing and estimating exposure mappings in settings with spillovers. A key difficulty in such settings is that correctly specifying the outcome model is especially challenging given the complex dependence induced by spillovers. We adopt a design-based approach that sidesteps this by treating potential outcomes as fixed for the experimental sample and letting all randomness come from the known assignment mechanism; this randomization serves as the key assumption needed to test and conduct inference on the exposure mapping. Starting from a proposed exposure map $g(W; X_i, \theta)$, we develop tools to (i) test whether the implications of this map are consistent with observed outcomes under the known assignment mechanism and (ii) estimate the tuning parameters $\theta$. To quantify uncertainty, we develop a design-based law of large numbers and central limit theorem for GMM estimators under spatial and network dependence, characterize an efficiency bound and optimal moments within the class of design-based estimators, and propagate the resulting exposure-map uncertainty into downstream policy estimands.

Formally, we work with a known randomized (or quasi-randomized) assignment design \(\Dcal_n\) and a proposed exposure map \(g(W;X_i,\theta)\). A key observation is that a correctly specified exposure mapping implies an \emph{exposure sufficiency} property: there exists a true parameter \(\theta_0\) such that potential outcomes depend on the full treatment assignment vector only through the exposure mapping \(g(W;X_i,\theta_0)\), so that conditional on \(g(W;X_i,\theta_0)\) the remaining details of \(W\) carry no additional outcome-relevant information.
This property implies a large set of orthogonality conditions between outcomes and functions of \(W\), which we construct via what we call \emph{design-side residuals}. Given any integrable outcome transformation \(\phi\) and design function \(\psi\), we orthogonalize \(\psi(W)\) with respect to the exposure mapping by subtracting its conditional expectation given \(g(W;X_i,\theta)\) \citep[in the spirit of][]{BorusyakHull2023}, and for each unit \(i\) define the design-side residual
\(
  R_{i,\theta}{\psi}(W)
  :=
  \psi(W)
  -
  \E\big[\psi(W)\mid g(W;X_i,\theta)\big],
\)
where the conditional expectation is taken with respect to \(\Dcal_n\), and thus can be computed by repeatedly sampling from the known design. Under exposure sufficiency and (quasi-)random assignment, we show that at the true parameter these design-side residuals are orthogonal to any function of unit $i$'s outcome:
\(
  \E\big[\phi(Y_i)\, R_{i,\theta_0}{\psi}(W)\big] = 0
\)
for every \(\phi\) and \(\psi\). Varying \((\phi,\psi)\) generates a family of moment conditions that form the basis for GMM estimation; when the system is overidentified, the associated \(J\)-test assesses whether the proposed exposure mapping is consistent with the design.
In this way, a single design-side residualization strategy supports both of our key objectives: evaluating the fit of a proposed exposure map and learning about the tuning parameter \(\theta_0\) from the data.

For inference, we establish consistency and asymptotic normality of the design-based GMM estimator under spatial and network dependence. We work in the affinity-set dependence framework of \citet{chandrasekhar2023general}, which nests classical structures such as mixing random fields, m-dependence, and dependency graphs as special cases. Existing results in this framework are pointwise; to our knowledge we are the first to establish uniform consistency and asymptotic normality in it, and we give primitive conditions that cover both smooth exposure maps (gravity and market-access kernels) and non-smooth ones (the ring radius).
We also characterize efficiency within this moment class: there is an asymptotic-variance lower bound for any moment the framework constructs, and a feasible sieve attains it.

Finally, we carry the resulting uncertainty about $g(W; X_i, \theta)$ into downstream policy estimands (e.g., average effects under counterfactual assignment rules) so that inference reflects both the experimental variation and uncertainty about the exposure specification itself.


We illustrate the framework in two applications to large-scale anti-poverty
programs in development economics: the large-scale public-works reform in rural India analyzed
by \citet{muralidharan2023generaleq}, and the GiveDirectly cash-transfer
experiment in rural Kenya studied by \citet{egger2022general} together with the
follow-up structural analysis of \citet{walker2024slack}. Both study large interventions in
low-income economies, are explicitly motivated by general-equilibrium
spillovers, involve spatially linked local markets, and rely primarily on
ring-based exposure mappings.

The two applications yield contrasting conclusions. In the
\citet{muralidharan2023generaleq} application, the design-based
overidentification tests do not reject their ring specification, and the
estimated radii broadly support the original 20 km labor-market scale for the
main income outcomes.
In the \citet{egger2022general} application,
by contrast, the ring specification is rejected for several core outcomes, and our support diagnostics reject the original 2 km radius for all outcomes,
\bluerev{with the smallest non-rejected supports often being 4--6 km or larger.} This pattern is consistent with the follow-up analysis of
\citet{walker2024slack}, which argues that very local rings can miss
market-level ambient spillovers that are common to households within local
markets. Consistent with this narrative, we obtain a local fiscal transfer multiplier of \bluerev{\ridgerev{1.60}}, which is smaller than \citet{egger2022general}'s original 2.5, and much closer to the model-predicted multiplier of 1.54 by \citet{walker2024slack}.

Together, these applications illustrate the value of treating exposure mappings as objects of design-based inference rather than fixed researcher choices: the framework can support an economically motivated map, detect and guide revisions when it does not, and propagate the resulting uncertainty into policy-relevant estimates.


%
\subsection{Related literature}


\paragraph{Design-based uncertainty.}
We adopt a finite-population, design-based perspective in which $\{Y_i(\cdot),X_i\}_{i=1}^n$ are fixed and randomness arises only from the assignment mechanism. This tradition traces to \citet{Neyman1923} and motivates modern distinctions between design-based and superpopulation uncertainty in regression and experimental analyses \greendel{\citep{AbadieAtheyImbensWooldridge2020,AbadieAtheyImbensWooldridge2023}}\greenrev{\citep{AbadieAtheyImbensWooldridge2020}}. Beyond fully randomized designs, \citet{RambachanRoth2025} develop a design-based theory of uncertainty for canonical quasi-experimental estimators when treatment propensities can vary across units, and \citet{LiDing2017} provide finite-population central limit theorems yielding $\sqrt{n}$-Gaussian approximations under randomization-based dependence. Our contribution operates within this design-based framework but targets moment conditions implied by exposure-sufficiency restrictions for parameterized exposure mappings\redrev{, under the cross-unit dependence induced by interference}.

\paragraph{Interference and exposure mappings.}

Foundational potential-outcomes treatments of interference include \citet{Sobel2006} and \citet{HudgensHalloran2008}. Building on this tradition, \citet{AronowSamii2017} formalize \emph{exposure mappings} and develop randomization-based estimators for causal effects under general (but specified) interference structures, while \citet{SavjeAronowHudgens2021} clarify large-sample behavior of conventional estimators under unknown interference and show that standard variance formulas can fail. \greenrev{\citet{Leung2022ANI} allows treatments assigned to increasingly distant units to have smaller but nonzero effects under approximate neighborhood interference, a class that no finite-dimensional exposure map characterizes and so falls outside our framework.} Recent work also studies what can and cannot be learned about interference structure from the design itself: \redrev{\citet{GaoHarshawSavjeWang2026} show that no specification test can have uniform power against any richer exposure-mapping alternative to an exposure-mapping model}\footnote{See Remark \ref{rem:gaocomp} for detailed discussion.}, while \redrev{\citet{Zhong2025Unconditional} develops finite-sample randomization tests for the existence and extent of interference (for example, whether spillovers vanish beyond a given distance) under minimal assumptions on the network.} Complementary work emphasizes robustness to unknown spillover mechanisms within broad classes and studies estimators (and sometimes designs) with minimax or neighborhood-adaptive guarantees; see \citet{belloni2022neighborhood} and \citet{faridani2024linear}. Our approach is complementary to these robust and testing-based methods: rather than remaining fully agnostic about structure or testing a partially sharp no-spillover null, we take the parametric exposure maps used in applied work (e.g.,\ ring, gravity, market-access) \redrev{as maintained} and use the known design to (i) estimate the tuning parameter $\theta_0$ and (ii) test the associated exposure-sufficiency restrictions via overidentifying, design-implied orthogonality conditions.\greenrev{\footnote{A further complementary literature asks how exposure-based estimands should be interpreted when the researcher's chosen exposure mapping is misspecified. \citet{Savje2024Misspecified} separates the use of an exposure mapping to define an estimand from the assumption that the mapping captures the complete causal structure, and establishes conditions under which exposure effects remain estimable under misspecification. \citet{ParkYang2026} instead take the marginal policy effect as primitive and show that a researcher-chosen exposure mapping induces a pseudo-true outcome model and a corresponding decomposition into direct and spillover components.}}

\paragraph{Testing models with trusted shocks.}
A line of work going back to \citet{Lucas1980} views structural models as objects to be disciplined by their responses to shocks with well-understood sources and exogeneity properties. \greendel{In development, \citet{todd2006assessing}, \citet{AttanasioMeghirSantiago2012}, and \citet{DufloHannaRyan2012} combine dynamic structural models with randomized policy variation to estimate deep parameters and test whether the model reproduces experimentally identified impacts.} In quantitative trade, \citet{AdaoCostinotDonaldson2023} propose IV-based goodness-of-fit statistics that compare a model's predicted response to quasi-experimental tariff shocks with the observed response. We bring this ``trusted shocks'' logic to spillovers and exposure mappings: the parametric exposure family $g(W;X_i,\theta)$ plays the role of the structured model\redrev{, a structural restriction on how treatments propagate even in otherwise reduced-form designs}, the (quasi-)experimental assignment $W\sim\Dcal_n$ (or $W\sim G_n(\cdot\mid X)$) provides the trusted shock, and the design-implied orthogonality conditions deliver both specification tests for the maintained exposure class and design-based estimators of the exposure parameter $\theta_0$.

\paragraph{Orthogonalization and recentering.}
Our design-side residualization builds on a long tradition of orthogonalization in semiparametric and GMM settings, including the partially linear model of \citet{robinson1988root} and the construction of ``orthogonal instruments'' in IO \citep{AckerbergCrawford2009,AckerbergCrawfordHahn2011,AndrewsBarahonaGentzkowRambachanShapiro2025}. Closest in spirit is \redrev{the recentering approach of} \citet{BorusyakHull2023}, who study settings where a known shock design generates a constructed regressor (a formula instrument) and recenter it by subtracting its conditional mean given covariates to obtain orthogonal moments \redrev{for a downstream causal or structural parameter}. \greenrev{Extensions to optimality are provided by \citet{BorusyakHull2026Optimal}, and to nonlinear settings by \citet{borusyak2025estimating}, whose application is demand estimation rather than spillovers.} \redrev{Whereas that work uses recentering to identify and estimate a downstream parameter given a fixed formula or exposure, we apply the same orthogonalization logic upstream to discipline and test the exposure map itself}: we treat $g(W;X_i,\theta)$ as the object of interest and use the known (quasi-)experimental design to identify, estimate, and test $\theta_0$. Relatedly, \citet{Ritzwoller2025Spillovers} develops reweighting procedures for proximity-exposure regressions using residualized proximity measures to isolate variation orthogonal to alternative mediating channels.

\rev{The rest of the paper is organized as follows.
Sections~\ref{sec:design-orthogonality}--\ref{sec:db-optimal-instruments} develop
the framework: identification via design-based orthogonality
(Section~\ref{sec:design-orthogonality}), a large-sample theory under spatial and
network dependence (Section~\ref{sec:DB-consistency-AN-affinity}), and the
efficient choice of moments (Section~\ref{sec:db-optimal-instruments}). Section~\ref{sec:stage2} propagates
the estimated map into downstream estimands, and
Section~\ref{sec:applications-main} presents the two applications.}

\section{Setup and Moment Conditions}
\label{sec:design-orthogonality}

\subsection{Design, outcomes, and exposure maps}
We observe units $i=1,\dots,n$ (regions, households, or network nodes) and a treatment
assignment vector $W = (W_1,\dots,W_n)^\top$ drawn from a known experimental (or
quasi-experimental) design $\Dcal_n$ (e.g., complete, Bernoulli, stratified,
cluster-randomized, shock-based, etc.).

\rev{The assignment takes values in a known support $\mathcal{W}$. Binary
treatment, $\mathcal{W}=\{0,1\}^n$, is the leading case, but nothing in what
follows requires it: the individual assignment $W_i$ may be multivalued or continuous, as with the
per-village transfer amounts in our second application.}
For each assignment $w \in \mathcal{W}$, unit $i$ has a potential outcome $Y_i(w)$\redrev{, a
function of the \emph{entire} assignment vector $w$ rather than unit $i$'s own
assignment $w_i$ alone; this is what allows unit $i$'s outcome to depend on other
units' treatments}. The realized outcome is $Y_i = Y_i(W)$.

Let $X_i$ denote observed unit-level information, taking values in a set $\mathcal{X}$, that may be relevant for how
assignment affects unit $i$, such as location, network links, strata, or baseline
covariates. An \emph{exposure map} is a function
\[
  g : \mathcal{W} \times \mathcal{X} \times \Theta \to \mathbb{R}^k,
\]
with parameter $\theta \in \Theta \subset \mathbb{R}^p$ and $k \ll n$, intended to
summarize the aspects of the assignment that are relevant for $Y_i$ through
\[
  g(W;X_i,\theta).
\]
When there is no risk of confusion, we use the shorthand notation
\[
  g_i(W;\theta) := g(W;X_i,\theta),
\]
and write $g_i(w;\theta)$ for the exposure induced by a nonrandom assignment
$w \in \mathcal{W}$.


Throughout, we adopt a finite-population, design-based perspective: the potential-outcome
schedule $\{Y_i(\cdot)\}_{i=1}^n$ and covariates $\{X_i\}_{i=1}^n$ for the experimental (or quasi-experimental) sample are treated as fixed (or, equivalently, conditioned upon), and all randomness
comes from $W \sim \Dcal_n$.

We write $\E$ and $\Pp$ for expectation and probability with respect to $\Dcal_n$,
conditioning implicitly on this fixed schedule of potential outcomes and covariates.
This perspective, as in
\citet{Neyman1923,AbadieAtheyImbensWooldridge2020,LiDing2017}, anchors inference to the specific network, market, or spatial
environment actually exposed to the policy, rather than positing a hypothetical
superpopulation experiment in which entire economies, including their equilibrium
prices, networks, and cross-unit dependence, are repeatedly resampled and subjected
to new assignments $W$. Such an experiment is least credible exactly when spillovers
are system-wide, since it would then require strong assumptions about how the joint
distribution of $\{Y_i(w)\}_{i,w}$ and the broader equilibrium environment vary
coherently across draws. The design-based framework is correspondingly most appealing
here: taking the realized economy as fixed and the known randomization or shock design
as the sole source of uncertainty, it delivers a transparent basis for inference
without additional assumptions on the population-generating process.

\subsection{Examples of exposure maps}
\label{subsec:three-exposure-maps}

As a concrete illustration, we document three families of exposure maps that recur in applied work. Let
\(d(i,j)\) denote a fixed, nonstochastic distance between units \(i\) and \(j\), where
distance may be geographic, travel-time, or network distance.

\begin{example}[Ring exposure]
\label{ex:ring-egger}
A ring exposure map imposes a hard spatial cutoff: only treated units within distance
\(\theta\) of unit \(i\) contribute to exposure. Define the row-normalized weights
\begin{equation}
  a_{ij}(\theta)
  :=
  \frac{\mathbf 1\{d(i,j)\le \theta\}\,\mathbf 1\{j\neq i\}}
       {\sum_{k=1}^n \mathbf 1\{d(i,k)\le \theta\}\,\mathbf 1\{k\neq i\}},
  \qquad
  a_{ii}(\theta):=0,\footnote{
We use the convention $a_{ij}(\theta):=0$ for all $j$ when no other unit lies within distance $\theta$ of $i$ (empty denominator).}
  \label{eq:ring-exposure-perunit}
\end{equation}

and let
\[
  g^{\mathrm{ring}}(W;X_i,\theta)
  :=
  \sum_{j=1}^n a_{ij}(\theta) W_j .
\]
Thus \(g^{\mathrm{ring}}(W;X_i,\theta)\) is the average treatment status among units
lying within radius \(\theta\) of \(i\). This is the logic behind the distance-buffer
specifications in \citet{egger2022general}. It is also the main exposure design in
\citet{muralidharan2023generaleq}, where the baseline specification uses the treated share
of locations within a fixed \(20\) km radius.
\end{example}

\begin{example}[Smooth spatial decay]
\label{ex:smooth-spatial-decay}
Smooth spatial-decay maps replace the hard cutoff in a ring design with weights that decline
continuously with distance. Let \(L_j>0\) denote a pre-treatment measure of the economic size
or attractiveness of location \(j\), such as population, employment, or market size. Define
\begin{equation}
  a_{ij}(\theta)
  :=
  \frac{K(d(i,j),L_j;\theta)\,\mathbf 1\{j\neq i\}}
       {\sum_{k=1}^n K(d(i,k),L_k;\theta)\,\mathbf 1\{k\neq i\}},
  \qquad
  a_{ii}(\theta):=0,
  \label{eq:smooth-spatial-decay-weights}
\end{equation}
and
\begin{equation}
  g^{\mathrm{ssd}}(W;X_i,\theta)
  :=
  \sum_{j=1}^n a_{ij}(\theta) W_j .
  \label{eq:smooth-spatial-decay-exposure}
\end{equation}

This formulation nests several familiar kernels. One example is a gravity-style exponential
kernel,
\[
  K_{\mathrm{grav}}(d,L;\theta) := L\,\exp(-\theta d),
\]
under which exposure decays smoothly with distance at rate \(\theta\), used in  \citet{franklin2024urban,walker2024slack}. Another important
example is a market-access kernel with power decay,
\[
  K_{\mathrm{ma}}(d,L;\theta) := L(1+\alpha d)^{-\theta},
\]
which is closely related to market-access measures in \citet{redding2004economic} and
\citet{donaldson2016railroads}. This is the alternative exposure class considered in
Appendix~B of \citet{muralidharan2023generaleq}, where the authors use the form
\((1+\alpha d)^{-\theta}\) with \(\alpha=1/100\) and fix \(\theta=8\) based on
\citet{donaldson2016railroads}.
\end{example}

\begin{example}[Network exposure]
\label{ex:network-distance}

When spillovers propagate through a known baseline network, distance may be measured in graph
steps rather than kilometers. The ring exposure in Example~\ref{ex:ring-egger} also covers network spillovers as a
special case if distance is interpreted as graph distance rather than geographic distance. Let \(G=(V,E)\) be a fixed graph on nodes
\(V=\{1,\dots,n\}\), and let \(\mathrm{dist}_G(i,j)\) denote shortest-path distance on
\(G\). For \(\theta\in\{1,2,\dots\}\), define the neighborhood
\[
  \mathcal N_{\le \theta}(i)
  :=
  \{\,j\neq i : \mathrm{dist}_G(i,j)\le \theta\,\}.
\]
A natural scalar exposure map is then
\begin{equation}
  g^{\mathrm{net}}(W;X_i,\theta)
  :=
  \frac{1}{|\mathcal N_{\le \theta}(i)|}
  \sum_{j\in\mathcal N_{\le \theta}(i)} W_j,
  \label{eq:network-hop-exposure}
\end{equation}
with the convention \(g^{\mathrm{net}}(W;X_i,\theta)=0\) if
\(|\mathcal N_{\le \theta}(i)|=0\).

In social-network applications, researchers often further decompose this into the share of
treated friends, the share of treated friends of friends, and so on up to distance
\(\theta\).\footnote{In the network-interference literature this is often described as a
\(K\)-hop restriction. We write the truncation parameter as \(\theta\) here to keep notation
consistent across the three exposure families.} When \(\theta=1\),
\eqref{eq:network-hop-exposure} reduces to the familiar share of treated friends. This type
of restriction is common in empirical work on peer effects and network spillovers; see, for
example, \citet{cai2015social} and \citet{BramoulleDjebbariFortin2009}.
\end{example}

Across these examples, the role of \(\theta\) is always the same: it indexes how far
spillovers reach or how quickly they decay. Table~\ref{tab:exposure-cartoon} illustrates the three families.

\newcommand{\RingCartoon}{%
\begin{tikzpicture}[
    scale=0.7,
    every node/.style={circle, draw, inner sep=1pt, minimum size=5pt},
    treated/.style={circle, draw, fill=black, inner sep=1pt, minimum size=5pt},
    untreated/.style={circle, draw, fill=white, inner sep=1pt, minimum size=5pt}
]
  \node[treated,label=above:{\scriptsize $i$}] (i) at (0,0) {};
  \draw[dashed] (0,0) circle [radius=1.2];

  \foreach \angle/\style in {90/treated,150/treated,210/untreated,330/treated}{
    \node[\style] at (1.2*cos{\angle},1.2*sin{\angle}) {};
  }
  \foreach \angle in {30,270}{
    \node[untreated] at (2.0*cos{\angle},2.0*sin{\angle}) {};
  }
\end{tikzpicture}%
}
\newcommand{\GravCartoon}{%
\begin{tikzpicture}[
    scale=0.8,
    every node/.style={circle, draw, inner sep=1pt, minimum size=4pt},
    center/.style={circle, draw, fill=black, inner sep=1pt, minimum size=5pt},
    ring/.style={line width=0.25pt, draw=black!35},
    arrowstyle/.style={-stealth}
]
  \node[center,label=above:{\scriptsize $i$}] (i) at (-0.1,0) {};

  \foreach \r in {0.7,1.2,1.7} {
    \draw[ring] (i) circle[radius=\r];
  }

  \node (n1) at ($(i)+(0.7,0.2)$) {};      
  \node (n2) at ($(i)+(-1.0,0.7)$) {};     
  \node (n3) at ($(i)+(1.1,-0.7)$) {};     
  \node[minimum size=7pt] (n4) at ($(i)+(-1.6,-0.2)$) {}; 

  \draw[arrowstyle, ultra thick] (n1) -- (i);
  \draw[arrowstyle, thick]       (n2) -- (i);
  \draw[arrowstyle, thin]        (n3) -- (i);
  \draw[arrowstyle, very thin]   (n4) -- (i);
\end{tikzpicture}%
}
\newcommand{\SmoothDecayCartoon}{\GravCartoon}

\newcommand{\KHopCartoon}{%
\begin{tikzpicture}[
    scale=0.6,
    every node/.style={circle, draw, inner sep=1pt, minimum size=5pt},
    center/.style={circle, draw, fill=black, inner sep=1pt, minimum size=6pt},
    hopone/.style={circle, draw, fill=gray!60, inner sep=1pt, minimum size=5pt},
    hoptwo/.style={circle, draw, fill=white, very thick, inner sep=1pt, minimum size=5pt},
    far/.style={circle, draw, fill=white, inner sep=1pt, minimum size=5pt}
]
  \node[center] (i) at (0,0) {};
  \node[hopone] (a) at (1,0) {};
  \node[hopone] (b) at (-1,0) {};
  \node[hopone] (c) at (0,1) {};
  \node[hopone] (d) at (0,-1) {};
  \node[hoptwo] (e) at (2,0) {};
  \node[hoptwo] (f) at (-2,0) {};
  \node[hoptwo] (g) at (0,2) {};
  \node[hoptwo] (h) at (0,-2) {};
  \node[far] (p) at (2.5,-1.5) {};
  \node[far] (q) at (-2.5,-1.5) {};

  \foreach \nbr in {a,b,c,d}{ \draw (i) -- (\nbr); }
  \draw (a) -- (e); \draw (b) -- (f); \draw (c) -- (g); \draw (d) -- (h);
  \draw (e) -- (p); \draw (h) -- (q);

  \draw[dashed, rounded corners=4pt] (-1.5,-1.5) rectangle (1.5,1.5);
  \draw[dotted, rounded corners=4pt] (-2.3,-2.3) rectangle (2.3,2.3);
\end{tikzpicture}%
}
\newcommand{\NetworkCartoon}{\KHopCartoon}
\begin{table}[t]
\centering
\caption{Illustration of canonical exposure maps}
\label{tab:exposure-cartoon}
\small
\begin{tabularx}{\textwidth}{@{} l
                                m{3.0cm}
                                X
                               @{}}
\toprule
& Illustration & Informal description \\
\midrule
\multicolumn{3}{l}{} \\[-0.75ex]

Ring
  & \raisebox{-0.5\height}{\RingCartoon}
  & Exposure is the average treatment status among units lying within distance
    $\theta$ of unit $i$. This is the main specification of
    \citet{muralidharan2023generaleq} and the distance-buffer logic used in
    \citet{egger2022general}. \\[0.9ex] 

Smooth spatial decay
  & \raisebox{-0.5\height}{\SmoothDecayCartoon}
  & All locations can affect $i$, but their influence declines smoothly with
    distance. Gravity-style exponential kernels and market-access kernels of the
    form $(1+\alpha d)^{-\theta}$ are leading examples; the latter is the
    alternative exposure class considered in Appendix~B of
    \citet{muralidharan2023generaleq}. \\[0.9ex]

Network
  & \raisebox{-0.5\height}{\NetworkCartoon}
  & Exposure is averaged over nodes within graph distance $\theta$ of $i$. In
    social-network applications this often corresponds to the share of treated
    friends, or more generally treated friends, friends of friends, and so on,
    up to distance $\theta$. \\
\bottomrule
\end{tabularx}
\end{table}

\subsection{Exposure sufficiency}
\label{subsec:randomization-exposure-sufficiency}

We now state the two primitive conditions that generate the design-based moment
restrictions. The first condition is the usual finite-population randomization
condition: the potential-outcome schedule is fixed, and all randomness comes from
the known assignment law.

\begin{assumption}[Randomized assignment]
\label{ass:rand-indep}
The assignment $W$ is drawn from a known law $\Dcal_n$ that does
not depend on the potential outcomes. Formally, for every collection of potential
outcomes $\{Y_i(\cdot)\}_{i=1}^n$,
\[
  W \,\big|\, \{Y_i(\cdot),X_i\}_{i=1}^n \sim \Dcal_n .
\]
\end{assumption}

The second condition is the exposure-map hypothesis. It states that, at the
correct value $\theta_0$, the candidate exposure map contains all information in the
assignment vector that is relevant for unit $i$'s outcome.

\begin{hypothesis}[Exposure map]
\label{hyp:exposure-suff}
\label{ass:exposure-suff}
\label{ass:exposure-sufficiency}
The exposure mapping $g(W;X_i,\theta)$ is \emph{well specified} if there exist
$\theta_0\in\Theta$ and measurable functions
$\widetilde Y_i:\mathbb R^k\to\mathbb R$ such that, for every assignment $w$,
\[
  Y_i(w)
  =
  \widetilde Y_i\bigl(g(w;X_i,\theta_0)\bigr),
  \qquad i=1,\dots,n .
\]
\end{hypothesis}

Thus Hypothesis~\ref{ass:exposure-suff} makes the exposure \(g(W;X_i,\theta_0)\) a summary of the assignment for unit $i$:
the map is \emph{well specified} when the exposure captures everything in $W$ relevant to
$Y_i$. The goal is to test this restriction and learn $\theta_0$ using the known design
\(\Dcal_n\).

\redrev{Our main result is that, combined with the known design, well specification has an
observable consequence. Because the assignment law is known and does not depend on the
potential outcomes, conditioning on the exposure removes all dependence between $Y_i$ and
the residual variation in $W$. We call this consequence exposure sufficiency. It is what
makes the map testable: the design leaves assignment variation that a well-specified
exposure must render irrelevant to $Y_i$, and any leftover dependence is evidence against
the map.}

\begin{theorem}[Exposure sufficiency]
\label{thm:design-orthogonality}
\label{thm:design-orthogonality-unit}
\label{thm:exposure-sufficiency}
Suppose Assumption~\ref{ass:rand-indep} and Hypothesis~\ref{hyp:exposure-suff} hold.
Then, for every unit \(i\),
\[
  Y_i \;\perp\!\!\!\perp\; W
  \;\big|\;
  g(W;X_i,\theta_0)
\]
under the design distribution.
\end{theorem}

\redrev{To see what the theorem requires, and why it can be tested, consider the ring map
of Example~\ref{ex:ring-egger}, where $g(W;X_i,\theta)$ is the treated share of units within
radius $\theta$ of unit $i$. Applied researchers already probe this choice informally:
one draws rings of half a mile, a mile, two miles, and looks for the radius at which
estimated spillovers level off. The design-based restriction makes that intuition precise.
Well specification at radius $\theta_0$ says unit $i$'s outcome depends on the assignment
only through the within-$\theta_0$ treated share: two assignments with the same share
produce the same outcome for $i$, no matter which of those neighbors are treated and
regardless of the treatment of any unit beyond $\theta_0$. This is a substantive economic
restriction, and it can fail in two ways: a treated unit just outside the radius still
affects $i$, or the identity rather than the count of treated neighbors matters. The known
randomization turns each failure into something observable. If the ring is correct, then
under the design the treatments of units outside radius $\theta_0$ are uncorrelated with
$Y_i$ once we condition on the within-ring share; leftover correlation is evidence that
spillovers reach past $\theta_0$. This is exactly the equal-to-zero condition that the
moments of the following subsection formalize, both to test a candidate radius and to
select $\theta_0$.}
\subsection{Orthogonal moments}
\label{subsec:orthogonal-moments}

\redrev{The ring discussion tested one exposure map with one natural statistic, the
treatments of units outside the radius. The conditional independence in
Theorem~\ref{thm:exposure-sufficiency} says much more: \emph{any} function of the
assignment, once purged of what the candidate exposure explains, must be unrelated to
\emph{any} function of the outcome. We encode this with a residual that strips from a
design function the part predictable from the exposure.}
To construct the moments, let
\(\psi:\mathcal{W}\times\mathcal{X}\to\mathbb R\) be any integrable design function, possibly depending on the unit's covariates \(X_i\).\footnote{The design residual \(R_{i,\theta}\) below evaluates the design function at the unit's covariates \(X_i\), so the unit-centering enters through the same index \(i\) that \(R_{i,\theta}\) already carries. This covers the leading applied case, in which \(\psi\) aggregates the assignment in a neighborhood of \(X_i\) (for example, population-weighted treatment averages over distance bands around unit \(i\)). To avoid clutter we continue to write \(\psi(W)\), and \(\psi_m(W)\) for dictionary elements, each understood to be evaluated at the relevant unit's covariates inside \(R_{i,\theta}\).} For each unit \(i\) and candidate parameter
\(\theta\), define the design residual
\begin{equation}
  R_{i,\theta}{\psi}(W)
  :=
  \psi(W;X_i)
  -
  \E\!\left[\psi(W;X_i)\mid g(W;X_i,\theta)\right],
  \label{eq:Ri-theta-def}
\end{equation}
where the conditional expectation is taken under the known design \(\Dcal_n\).
This residual removes from \(\psi(W)\) the component explained by the candidate
exposure \(g(W;X_i,\theta)\).

\begin{corollary}[Orthogonal moments]
\label{cor:design-moment-unit}
Under the conditions of Theorem~\ref{thm:exposure-sufficiency}, for any outcome
transformation \(\phi:\mathbb R\to\mathbb R\) and design function \(\psi\)
\redrev{with \(\E[\phi(Y_i)^2]<\infty\) and \(\E[\psi(W)^2]<\infty\)},
\begin{equation}
\label{eq:design-moment-unit}
  \E\!\left[
    \phi(Y_i)\,
    R_{i,\theta_0}{\psi}(W)
  \right]
  =
  0 .
\end{equation}
Equivalently,
\[
  \E\!\left[
    \phi(Y_i)
    \left\{
      \psi(W)
      -
      \E[\psi(W)\mid g(W;X_i,\theta_0)]
    \right\}
  \right]
  =
  0 .
\]
\end{corollary}

\redrev{Two features of the corollary matter for what follows. First, it is agnostic about
the map: nothing in \eqref{eq:design-moment-unit} is special to the ring, so the same
construction disciplines gravity, market-access, or network exposures, with only
$g(W;X_i,\theta)$ changing. Second, the equalities hold at $\theta_0$ for \emph{every}
admissible $(\phi,\psi)$, which is what lets a single strategy serve two ends: matching
them identifies and estimates $\theta_0$, while a candidate map that violates them is
detectable through the resulting overidentification.}

For estimation, fix a finite dictionary of moment-generating pairs
\[
  \{(\phi_m,\psi_m):m=1,\dots,M\},
\]
where each \(\phi_m:\mathbb R\to\mathbb R\) is an outcome transformation and each
\(\psi_m:\mathcal{W}\times\mathcal{X}\to\mathbb R\) is a design function. Define the $i$-specific moment
contribution
\begin{equation}
  \eta_{m,i,n}(\theta)
  :=
  \phi_m(Y_i)\,R_{i,\theta}{\psi_m}(W),
  \label{eq:single-unit-score}
\end{equation}
and the corresponding sample moment
\[
  \eta_{m,n}(\theta)
  :=
  \frac{1}{n}\sum_{i=1}^n\eta_{m,i,n}(\theta).
\]
Stacking the \(M\) moments gives
\[
  \eta_n(\theta)
  :=
  \bigl(
    \eta_{1,n}(\theta),\dots,\eta_{M,n}(\theta)
  \bigr)^\top .
\]
A design-based GMM estimator is then
\begin{equation}
  \hat\theta_n
  \in
  \arg\min_{\theta\in\Theta}
  Q_n(\theta),
  \qquad
  Q_n(\theta)
  :=
  \eta_n(\theta)^\top \Lambda_n \eta_n(\theta),
  \label{eq:design-gmm-criterion}
\end{equation}
where \(\Lambda_n\) is a positive semidefinite weight matrix.
\begin{remark}[Practical implementation]
\label{rem:computing-design-projection}
The conditional expectations
\[
  \E[\psi_m(W)\mid g(W;X_i,\theta)]
\]
are design-side objects: because the assignment law $\Dcal_n$ is known
(Assumption~\ref{ass:rand-indep}), each is a functional of the design alone and
involves neither the potential outcomes nor any outcome model, so it can be computed
without additional assumptions. When enumeration of the support of $\Dcal_n$ is
feasible, the projection is evaluated directly. Otherwise, one draws assignments
$W^{(b)}\sim\Dcal_n$, $b=1,\dots,B$, and approximates
$\E[\psi_m(W)\mid g(W;X_i,\theta)=s]$ by regressing $\psi_m(W^{(b)})$ on a flexible
function of the simulated exposure $g(W^{(b)};X_i,\theta)$: for a discrete exposure,
by averaging $\psi_m$ within its realized values; for a continuous exposure, by a
low-order polynomial or kernel smoother in $s$, with $B$ and the smoother chosen so
that the residual approximation error is first-order negligible. In the applications
we use a quadratic in the exposure index across a few hundred placebo draws\ridgerev{, fit by
cross-validated ridge regression in the \citet{egger2022general} application}.
%
\end{remark}

\begin{remark}[Overidentification and testing exposure maps]
\label{rem:overidentification-exposure-testing}
When \(M>\dim(\theta)\), the system is overidentified. The associated design-based
\(J\)-statistic can therefore be used to test the maintained exposure specification or
to compare alternative exposure maps. Section~\ref{sec:DB-consistency-AN-affinity}
derives the large-sample behavior of \(\hat\theta_n\), and
Section~\ref{sec:db-optimal-instruments} studies the efficient choice of
moments.
\end{remark}

\begin{remark}[Multi-unit restrictions]
\label{rem:joint-exposure-future-work}
Hypothesis~\ref{ass:exposure-suff} also implies orthogonality restrictions
for any finite set of units. For any finite
$S\subset\{1,\dots,n\}$, let
\[
  Y_S := (Y_i)_{i\in S},
  \qquad
  g_S(W;\theta) := \bigl(g(W;X_i,\theta)\bigr)_{i\in S}.
\]
If the exposure map is well specified, then
\[
  Y_S \;\perp\!\!\!\perp\; W
  \;\big|\; g_S(W;\theta_0),
\]
and hence, for any square-integrable $\phi_S$ and $\psi$,
\[
  \E\!\left[
    \phi_S(Y_S)
    \left\{
      \psi(W)
      -
      \E[\psi(W)\mid g_S(W;\theta_0)]
    \right\}
  \right]
  =0.
\]

These restrictions could generate moments based on pairs or larger groups of
units, such as moments involving $(Y_i,Y_j)$ and the joint exposure vector
$(g(W;X_i,\theta_0),g(W;X_j,\theta_0))$. We do not pursue the full multi-unit
moment system here: it greatly expands the class of possible moments, and it is
less clear how to choose among them. Instead, the paper focuses on the unit-level
moments, which nest the moments used in applied work and already deliver tractable
estimators, specification tests, and efficiency analysis. Developing practical
procedures that exploit the broader cross-unit restrictions is left for future work.
\end{remark}

\begin{remark}[Quasi-experimental shock designs and relation to recentering]
\label{rem:quasi-and-BH}
Many applications of exposure mappings are based on \emph{quasi-experimental} variation
rather than literal randomized assignment.
In such settings, the researcher observes a realized shock vector $W$ (e.g.,\ line
openings in a transport network, sectoral demand shocks, or other environmental shocks),
treats the baseline covariates $X=(X_1,\dots,X_n)$ as fixed, and specifies a
\emph{shock design} $\Dcal_n^{\text{shock}}$ that captures the as-good-as-random
component of $W$.
Provided this design is known and satisfies the same structural property as
Assumption~\ref{ass:rand-indep},
\[
  W \,\big|\, \{Y_i(\cdot)\}_{i=1}^n \sim \Dcal_n^{\text{shock}}
  \quad\text{for all potential-outcome schedules},
\]
all of the constructions above go through after replacing $\Dcal_n$ by
$\Dcal_n^{\text{shock}}$.
In practice, $\Dcal_n^{\text{shock}}$ is implemented via permutations, placebo
networks, or other simulation schemes, and the re-randomization step used to approximate
$\E[\psi_m(W)\mid g(W;X_i,\theta)]$ is carried out by simulating $W$ from this shock
design.

This design-side residualization is analogous in spirit to the recentering approach of
\citet{BorusyakHull2023}: both start from a known shock or assignment distribution and
construct transformed instruments or moments that are orthogonal, by design, to certain
components of the assignment.
However, the goals and maintained structures are different.
\citet{BorusyakHull2023} assume a particular linear homogeneous treatment-effect model
and a given formula instrument, and their objective is to identify and estimate the
resulting coefficient $\beta$.\footnote{The recentering logic is not intrinsically tied to a linear model: recently, \citet{borusyak2025estimating} apply it to nonlinear (nested and mixed logit)
demand estimation.}
By contrast, in the present framework the exposure mapping itself is the primary object
of interest: we use the randomized or quasi-random shock design to learn about
$\theta_0$ in $g(W;X_i,\theta_0)$ and to test whether a proposed exposure class is
consistent with the design-implied orthogonality conditions. That said, we provide a general procedure to conduct formal inference on policy functionals in a downstream stage, discussed in Section~\ref{sec:stage2}.
\end{remark}

\section{Design-Based Consistency and Asymptotic Normality}
\label{sec:DB-consistency-AN-affinity}

Section~\ref{sec:design-orthogonality} derived design-based moment restrictions implied
by exposure sufficiency. We now study the large-sample behavior of GMM estimators and
tests constructed from those moments. Our goal is to establish design-based consistency
and asymptotic normality of the GMM estimator, and to develop the associated
overidentification tests of the exposure-map specification. See
Remark~\ref{rem:exact-randomization-tests} for discussion of exact randomization tests.


To separate the asymptotic argument from any particular choice of moments, we work with a
generic finite-dimensional moment contribution
\[
  \Psi_i(\theta)\in\mathbb R^q .
\]
In the exposure-mapping application, this vector is obtained by stacking finitely many
residualized design moments,
\[
  \Psi_i(\theta)
  =
  \bigl(
    \eta_{1,i,n}(\theta),\dots,\eta_{q,i,n}(\theta)
  \bigr)^\top,
  \qquad
  \eta_{m,i,n}(\theta)
  =
  \phi_m(Y_i)R_{i,\theta}{\psi_m}(W).
\]
The generic notation lets us state the LLN, consistency, CLT, and asymptotic-normality
results once, in terms of the dependence structure induced by the assignment design. These
results are then used directly for exposure-map estimation here and for the optimal-moment
analysis in Section~\ref{sec:db-optimal-instruments}.


\subsection{Finite-population setup and GMM criterion}

Following the design-based literature (e.g.,\
\citealp{freedman2008regression,AronowSamii2017}), we work with a triangular array
of experiments $\{(\Ucal_n,\Dcal_n)\}_{n\geq 1}$, where both the population
size $N_n:=|\Ucal_n|$ and the assignment design $\Dcal_n$ are allowed to
change with $n$.

For each $n$, let $\Ucal_n=\{1,\dots,N_n\}$ denote the finite experimental
population. For each unit $i\in\Ucal_n$ and parameter
$\theta\in\Theta\subset\R^p$, let $\Psi_i(\theta)\in\R^q$ be a
$q$-dimensional moment vector, with true parameter value $\theta_0$.
The array $\{\Psi_i(\theta)\}_{i\in\Ucal_n}$ is generated by the known
assignment design $\Dcal_n$, while the potential outcomes and baseline
attributes are treated as fixed. The design fixes the distribution of these moments in any given population; the large-sample results below require, in addition, regularity conditions on the dependence the design induces across units, which we impose through the affinity-set structure introduced next.
%

The sample moment is the
normalized finite-population average
\[
  \bar\Psi_n(\theta)
  :=
  \frac{1}{N_n}\sum_{i=1}^{N_n}\Psi_i(\theta),
  \qquad
  \mu_n(\theta)
  :=
  \frac{1}{N_n}\sum_{i=1}^{N_n}\E[\Psi_i(\theta)].
\]
Let $\Lambda_n$ be a symmetric positive semidefinite $q\times q$ weight matrix
with $\Lambda_n\pto \Lambda\succeq 0$, and define the quadratic GMM criterion
\[
  Q_n(\theta)
  :=
  \bar\Psi_n(\theta)^\top \Lambda_n\,\bar\Psi_n(\theta),
  \qquad
  \hat\theta_n \in \arg\min_{\theta\in\Theta} Q_n(\theta).
\]

In the limit, write
\begin{equation}
\label{eq:Q-theta-unit}
\begin{aligned}
  \mu(\theta)
  &:=
  \lim_{n\to\infty}\mu_n(\theta),\\
  Q(\theta)
  &:=
  \mu(\theta)^\top \Lambda\,\mu(\theta), 
\end{aligned}
\end{equation}
whenever the limit exists. In particular, we assume that the limit is centered
at the true parameter,
$  \mu(\theta_0)=0$.\footnote{
This centering follows directly from the
design-based orthogonality result in
Corollary~\ref{cor:design-moment-unit}.}


\subsection{Affinity sets and consistency}
\label{sec:affset}
To accommodate design-induced spatial or network dependence, for each $i\in\Ucal_n$
we fix an \emph{affinity set} $A_i\subseteq\Ucal_n$ with $i\in A_i$ collecting units
whose assignments may have non-negligible covariance with that of unit $i$.
We write $|A_i|$ for its cardinality and allow arbitrary dependence within $A_i$.
In the spatial examples, $A_i$ can be read as a growing geographic or travel-time
ball around unit $i$; in the market-access and gravity examples, it collects
locations whose assignment shocks receive non-negligible kernel weight for $i$;
and in the network examples, it corresponds to a local graph neighborhood whose
radius may grow slowly with $N_n$.
Outside $A_i$ we do not impose conditional independence; instead, we only assume
that the aggregate covariance contribution from $\{j\notin A_i\}$ is
asymptotically negligible relative to the contribution from within $A_i$.

Define centered variables
\[
  Z_{i,n}(\theta)
  := \Psi_i(\theta) - \E[\Psi_i(\theta)],
  \qquad
  \bar Z_n(\theta)
  := \frac{1}{N_n} \sum_{i=1}^{N_n} Z_{i,n}(\theta)
  = \bar\Psi_n(\theta)-\mu_n(\theta),
\]
and the \emph{within-affinity covariance matrix}\footnote{We take the affinity
relation to be symmetric ($j\in A_i\Leftrightarrow i\in A_j$), so that
$\Omega_n(\theta)$ is symmetric; this holds in all our examples (metric or
travel-time balls, symmetric kernels, undirected graph neighborhoods), and any
directed relation may be replaced by its symmetrization $\{j:j\in A_i\ \text{or}\
i\in A_j\}$.}
\[
  \Omega_n(\theta)
  := \sum_{i=1}^{N_n} \ \sum_{j \in A_i} \
  \Cov\!\big(Z_{i,n}(\theta),\, Z_{j,n}(\theta)\big)
  \in \R^{q \times q}.
\]

We first state the high-level ULLN and identification requirements.  The
following subsection then explains how the ULLN is verified under primitive
conditions in the two leading cases.

\begin{assumption}[Design-based ULLN and identification]
\label{ass:as-lln}
\mbox{}
\begin{enumerate}[label=\textup{AS--LLN\arabic*}, ref=\textup{AS--LLN\arabic*}, leftmargin=1.6cm]

\item\label{ass:as-lln1}
\textbf{Deterministic stabilization.}
There exists a deterministic $\mu:\Theta\to\R^q$ such that
\[
  \sup_{\theta\in\Theta}\,\|\mu_n(\theta)-\mu(\theta)\|\ \to\ 0.
\]

\item\label{ass:as-lln2}
\textbf{Uniform LLN for the design-based moments.}
The sample moments converge uniformly in probability to their limits:
\[
  \sup_{\theta\in\Theta}\big\|\bar\Psi_n(\theta)-\mu_n(\theta)\big\|
  \ \pto\ 0.
\]

\item\label{ass:as-lln3}
\textbf{Parameter space and identification of the population criterion.}
The parameter space $\Theta\subset\R^p$ is compact and $\theta_0\in\Theta$.
\redrev{The limit map $\mu:\Theta\to\R^q$ from \ref{ass:as-lln1} is continuous on $\Theta$.}
Given \ref{ass:as-lln1}, define
$Q(\theta)=\mu(\theta)^\top \Lambda\,\mu(\theta)$ on $\Theta$, and assume
\[
 Q(\theta)=0
  \quad\Longleftrightarrow\quad
  \theta=\theta_0.
\]

\item\label{ass:as-lln4}
\textbf{Weight convergence.}
$\Lambda_n\pto \Lambda\succeq0$.
\end{enumerate}
\end{assumption}

Assumption~\ref{ass:as-lln} is stated at a high level in terms of the
population map $\mu(\theta)$ and the sample map $\bar\Psi_n(\theta)$.
The next result shows that these high-level conditions are sufficient for
consistency.\footnote{Genuinely discrete exposure parameters (for example, an
integer hop-count) are handled by a separate finite-grid selection argument in
Appendix~\ref{app:finite-grid-ulln-proof}.}

\begin{remark}[Interpretable necessary conditions for identification]
\label{rem:identification-necessary}
\bluerev{We impose identification at the criterion level in \ref{ass:as-lln3}. Two necessary conditions specific to the moments
\eqref{eq:single-unit-score} are worth mentioning. First, if $w\mapsto g(W;X_i,\theta)$ is
injective at some $\theta$, then conditioning on $g(W;X_i,\theta)$ is equivalent
to conditioning on $W$, so $R_{i,\theta}\psi\equiv0$ for every $\psi$ by
\eqref{eq:Ri-theta-def} and the moments are uninformative about that $\theta$;
identification therefore requires strictly positive residual variation,
$\E[(R_{i,\theta}\psi(W))^2]>0$ for some admissible $\psi$. This is the
degeneracy of a continuous exposure parameter paired with a finely resolved
assignment, such as a smooth market-access decay in (near-)continuous
distances.
Second, the outcomes must respond to the assignment: under the sharp null of no
effect on any unit each $\phi(Y_i)$ is design-nonrandom, so
$\E[\phi(Y_i)R_{i,\theta}\psi(W)]=\phi(Y_i)\,\E[R_{i,\theta}\psi(W)]=0$ for
\emph{every} $\theta$ by the law of iterated expectations, and $\theta_0$ is not
separated.}
\end{remark}

\begin{theorem}[Design-based GMM consistency]
\label{thm:as-consistency}
Suppose Assumption~\ref{ass:as-lln} holds. Then
\[
  \sup_{\theta\in\Theta}\big|Q_n(\theta)-Q(\theta)\big| \;\pto\; 0,
\]
and any sequence of minimizers
$\hat\theta_n\in\arg\min_{\theta\in\Theta} Q_n(\theta)$ satisfies
$\hat\theta_n\pto\theta_0$.
\end{theorem}

\subsection{Primitive conditions for the uniform law of large numbers}
\label{subsec:primitive-ulln-main}
Consistency was established above under the high-level
uniform law of large numbers assumed in \ref{ass:as-lln2},
\[
  \sup_{\theta\in\Theta}\bigl\|\bar\Psi_n(\theta)-\mu_n(\theta)\bigr\|\pto 0 .
\]
This subsection gives primitive conditions under which that uniform law holds.  The smooth case follows from a standard argument; the ring needs care, because its
sample path is a step function of the radius $\theta$ and the smooth-GMM argument does
not apply. We verify it for the two exposure maps used in the applications, continuing
Examples~\ref{ex:smooth-spatial-decay} and~\ref{ex:ring-egger};
Appendix~\ref{subsec:primitive-LLN-proof} states the primitive conditions more formally
and gives the proofs.
\examplecont{ex:smooth-spatial-decay}{smooth spatial decay}
When the weights $a_{ij}(\theta)$ vary smoothly with $\theta$, as in the gravity and
market-access kernels, the moment vector is Lipschitz in $\theta$ and a standard
covering argument applies. The condition that is not automatic concerns the design
projection $m_{i,m,\theta}(s):=\E[\psi_m(W)\mid g_i(W;\theta)=s]$, which must itself vary
regularly in both arguments,
\[
  \bigl|m_{i,m,\theta}(s)-m_{i,m,\theta'}(s')\bigr|
  \le M_{i,n}\bigl(\|\theta-\theta'\|+\|s-s'\|\bigr),
  \qquad
  \frac1{N_n}\sum_{i=1}^{N_n}\E\,M_{i,n}=O(1).
\]
Smoothness of the exposure map does not by itself deliver this, so we impose the displayed
regularity condition directly as a primitive assumption in Appendix~\ref{subsec:primitive-LLN-proof}.
\examplecont{ex:ring-egger}{ring exposure}
Raising the radius $\theta$ changes exposure only through the units whose distance to
$i$ crosses $\theta$, so the sample path $\theta\mapsto\bar\Psi_n(\theta)$ moves in steps
and is not differentiable; in place of a smoothness argument, the uniform law rests on a
condition on the distances. That condition is the spatial counterpart of the bounded-density requirement familiar from regression-discontinuity
and threshold-regression designs, with pairwise distance $d_{ij}$ playing the role of the
running variable and the radius $\theta$ that of the cutoff: as a band of radii shrinks,
the weighted mass of distances inside it must vanish. Formally, for an interval
$I\subset\R_{+}$ of length $|I|$,
\[
  \lim_{\delta\downarrow0}\ \limsup_{n\to\infty}\
  \sup_{|I|\le\delta}\ \frac1{N_n}\sum_{i=1}^{N_n}\sum_{j\ne i}
  a_{ij}\,\1\{d(i,j)\in I\}=0,
\]
with $a_{ij}$ the exposure weights of Example~\ref{ex:ring-egger}. Moving the radius then
shifts only a vanishing share of the moment, so the empirical criterion converges
uniformly and its population limit is continuous in $\theta$ even though every sample path
jumps.
\begin{remark}[Relation to exact randomization tests]
\label{rem:exact-randomization-tests}
A common inferential approach in experiments with spillovers is the finite-sample exact
randomization test. These tests are cleanest for a \emph{sharp} null: one that pins down
each unit's potential outcome under every assignment, so realized outcomes can be
recomputed for any counterfactual assignment and the statistic has a known permutation
distribution. Exposure sufficiency is not sharp: it restricts how outcomes depend on the
assignment through the exposure map but leaves the potential outcomes otherwise
unspecified. For such non-sharp nulls, exact tests remain available by choosing focal
units and conditioning on the assignment cell within which the focal units' outcomes are
invariant; see, for example, \citet{athey2018exact}.
Our target, however, is not a single fixed null but estimation of the continuous tuning
parameter $\theta$, together with inference on the downstream regression coefficients
through the GMM criterion $Q(\theta)$. Exact randomization inference is already awkward
for the coefficient-based analyses in the applied work cited above, and more so for estimating
$\theta$ or evaluating $Q$. We therefore take the asymptotic conditional-moment route
developed in this section, which draws on the implications of exposure sufficiency
directly.
\end{remark}

\subsection{Asymptotic normality}

We now strengthen the LLN conditions above to obtain a CLT and a GMM asymptotic
normality result.
We retain the finite-population, design-based setup and notation introduced
above, and assume Assumption~\ref{ass:as-lln} holds so that
$\hat\theta_n\pto\theta_0$ by Theorem~\ref{thm:as-consistency}.

\bluerev{Asymptotic normality rests on two further ingredients, both imposed at
the population level. We state them formally as Assumption~\ref{ass:AN-AFF} in
Appendix~\ref{app:DB-AN-affinity} and describe their roles here. The first delivers
a \emph{pointwise} central limit theorem for the moment vector at the truth. Because
the design-based moments are dependent across units through the affinity sets of
Section~\ref{sec:affset}, we invoke the affinity-set central limit theorem of
\citet{chandrasekhar2023general}: under bounded fourth moments, decay of the
within-affinity covariances, and a stabilizing aggregate covariance
$\Omega_n/N_n\to\Omega$ with $\Omega$ positive definite, the centered moments satisfy
$\sqrt{N_n}\,\bar\Psi_n(\theta_0)\Rightarrow\Ncal(0,\Omega)$.

The second ingredient carries this pointwise statement to asymptotic normality of the
GMM estimator $\hat\theta_n$. It has two parts: \emph{mean differentiability} of the
population moment map $\mu(\theta)$ at $\theta_0$, with full-rank Jacobian $G$, and
\emph{stochastic equicontinuity} of the centered empirical process
$\mathbb G_n(\theta)=\sqrt{N_n}\big(\bar\Psi_n(\theta)-\mu(\theta)\big)$ near
$\theta_0$. Both are restrictions on the population map and on the empirical process,
not on the individual sample paths $\Psi_i(\theta)$. This is what lets the argument, a
Z-estimator application of \citet[Theorem~3.3.1]{vaart1996weak} to the projected map
$\theta\mapsto G^\top\Lambda_n\bar\Psi_n(\theta)$, cover the non-differentiable ring
path alongside the smooth gravity and market-access kernels; the affinity-set
conditions enter only through the pointwise CLT and the equicontinuity of
$\mathbb G_n$.}

\begin{theorem}[Design-based GMM asymptotic normality]
\label{thm:DB-AN-aff}
Under Assumption~\ref{ass:as-lln} and Assumption~\ref{ass:AN-AFF},
\[
  \sqrt{N_n}\,(\hat\theta_n-\theta_0)
  \ \Rightarrow\
  \Ncal\!\Big(
    0,\ (G^\top \Lambda\, G)^{-1} G^\top \Lambda\, \Omega\, \Lambda\, G\, (G^\top \Lambda\, G)^{-1}
  \Big),
\]
where $G=G(\theta_0)$ and $\Omega$ is the design-based asymptotic covariance in
Assumption~\ref{ass:AN-AFF}\,\textup{\ref{ass:an-aff3}}.
With the optimal weight $\Lambda=\Omega^{-1}$,
the asymptotic variance simplifies to $(G^\top \Omega^{-1}G)^{-1}$.
Moreover, if $\Lambda_n\pto\Omega^{-1}$, then the
efficient-weight case is obtained.
\end{theorem}

%
\begin{corollary}[Design-based overidentification test]
\label{cor:db-J-test}
Suppose Assumption~\ref{ass:as-lln} and Assumption~\ref{ass:AN-AFF} hold with
$q>p$, and let $\Lambda_n\pto\Omega^{-1}$ as in Theorem~\ref{thm:DB-AN-aff}. Under
the maintained exposure specification (Hypothesis~\ref{hyp:exposure-suff}),
\[
  N_n\,Q_n(\hat\theta_n)
  \ \Rightarrow\
  \chi^2_{\,q-p}.
\]
\end{corollary}

The statistic $N_n Q_n(\hat\theta_n)$ therefore provides a formal test of the
maintained exposure map, as anticipated in
Remark~\ref{rem:overidentification-exposure-testing}: rejection is evidence that
residualized assignment variation left over after conditioning on
$g(W;X_i,\theta)$ remains predictive of outcomes through the chosen moments.
Different finite moment collections yield different $J$-statistics;
Section~\ref{sec:db-optimal-instruments} studies how to choose moments efficiently.
\begin{remark}[Scope of the specification test]
\label{rem:gaocomp}
Contemporaneously, \citet{GaoHarshawSavjeWang2026} prove that specification
testing of exposure-mapping models is impossible against any richer
exposure-mapping alternative: every testing procedure, design included, has worst-case Type I and Type II errors summing to one, at every sample size. The alternative is simply
too large in that it places no structure across units and so admits
adversarial outcome schedules that are maximally separated from the null yet
generate the same observed-data distribution under every assignment. The implication is that informative tests exist only
against alternatives restricted beyond what exposure mappings alone impose,
as in their consistent test against a linear-in-means model.

Our test naturally embodies the restriction their result requires, in the form applied work already adopts. We maintain a
parametric exposure class $\{g(\cdot;\theta):\theta\in\Theta\}$ and a finite
set of residualized design moments; the null states that some
$\theta_0\in\Theta$ satisfies $\mu(\theta_0)=0$. The $J$-test then has power
against alternatives that keep the moment criterion bounded away from zero
uniformly over $\Theta$, under the convergence conditions of
Assumption~\ref{ass:as-lln} maintained along the alternative sequence.
Rejection therefore signals that assignment variation left over after
conditioning on the proposed exposure map remains predictive of outcomes.
Non-rejection means the maintained class is consistent with the design through
these moments, and nothing more: it does not certify exposure sufficiency
against the unrestricted alternative.
\end{remark}

\begin{remark}[Nonconservative design-based variance]
\label{rem:nonconservative-variance}
The variance in Theorem~\ref{thm:DB-AN-aff} is design-based but not a
Neyman-style conservative bound. Neyman-type bounds arise because the exact
variance of a treatment-effect estimator involves co-moments of the same unit's
potential outcomes across assignments, which are never jointly observed
\citep{Neyman1923, AbadieAtheyImbensWooldridge2020}. No such term arises
here: by Corollary~\ref{cor:design-moment-unit}, the score
\[
  \phi(Y_i)\left\{\psi(W)-\E[\psi(W)\mid g(W;X_i,\theta_0)]\right\}
\]
is mean zero unit by unit, whatever the outcome functions $\widetilde{Y}_i$ may
be, so the centered score is observed for the realized assignment and $\Omega$
is identified as a second-moment functional of observed scores under the known
design. Given the graph-HAC condition of Assumption~\ref{ass:graph-hac}, the
sandwich therefore estimates the exact large-sample variance for the maintained
exposure model.\end{remark}

It remains to estimate the asymptotic design covariance matrix \(\Omega\). 
%
%
%
%
%
%
%
%
%
%
%
%
%
%
%
In our setting $\Omega$ captures spatial and network dependence through the
affinity sets, so a natural object is a spatial/graph HAC estimator built from
$\{Z_{i,n}\}_{i\le N_n}$.
Rather than spell out primitive conditions for a particular estimator, we impose
a high-level consistency requirement in the spirit of spatial and network HAC
methods, and refer to the existing literature for sufficient conditions.
For spatial dependence, see \citet{Conley1999,KimSun2011}.
For general network dependence and $\psi$-dependent processes on graphs,
\citet{KojevnikovMarmerSong2021,sasaki2025gmm} provide LLN, CLT, and consistency results for a
closely related network HAC estimator.
We formalize our requirement as follows.

\begin{assumption}[Graph-HAC estimation of the asymptotic covariance]
\label{ass:graph-hac}
Let $\Omega$ be the design-based asymptotic covariance in
Assumption~\ref{ass:AN-AFF}\,\textup{\ref{ass:an-aff3}}.
There exists a sequence of graph-HAC (Conley-type) estimators $\hat\Omega_n$
constructed from $\{Z_{i,n}\}_{i\le N_n}$ and the affinity sets
$\{A_i\}_{i\le N_n}$ such that
        
 $\hat\Omega_n$ is consistent in operator norm for the normalized
        asymptotic covariance:
        \[
          \big\|\hat\Omega_n - \Omega_n/N_n\big\|_{\mathrm{op}} \ \pto\ 0.
        \]
\redrev{Evaluated at a preliminary consistent estimator $\hat\theta_n$ (Algorithm~\ref{alg:twostep-practical-scores}), the resulting plug-in error is $o_p(1)$: the design centering is computed from the known design rather than estimated, and $\hat\theta_n\pto\theta_0$ under Assumption~\ref{ass:as-lln} (Theorem~\ref{thm:as-consistency}).}
\end{assumption}



\section{Efficient Moments}
\label{sec:db-optimal-instruments}

Sections~\ref{sec:design-orthogonality}--\ref{sec:DB-consistency-AN-affinity}
introduced a class of design-based moments parameterized by an outcome
transformation $\phi$ and a design function $\psi$,
\[
  \eta_{\phi,\psi,n}(\theta)
  :=
  \frac1{N_n}\sum_{i\in\Ucal_n}
  \phi(Y_i)\,R_{i,\theta}\psi(W).
\] Two facts delimit what optimality within this class means.
First, the class is exhaustive at the unit level:
Appendix~\ref{app:conditional-poirier} shows that the collection of these
moments over all bounded measurable $(\phi,\psi)$ is equivalent to the
exposure sufficiency $Y_i \indep W \mid g(W;X_i,\theta_0)$. Second,
a correctly-specified exposure map also implies cross-unit restrictions, involving pairs or
larger sets of units (Remark~\ref{rem:joint-exposure-future-work}), which lie
outside this class and are not exploited. The bound below is therefore an
efficiency bound within the unit-level moment class, the class containing the
moments used in applied work and in our applications, not a semiparametric
bound over all implications of an exposure map hypothesis.

\begin{theorem}[Efficiency bound and sieve approximation; informal]
\label{thm:efficiency-informal}
Under the regularity conditions of Appendix~\ref{app:db-efficiency}:
\begin{enumerate}[label=\textup{(\roman*)}, leftmargin=1.2cm]
\item There is a finite $V^\star>0$ such that every GMM estimator of
$\theta_0$ built from finitely many moments in the unit-level class has
design-based asymptotic variance at least $V^\star$


\item For a nested dictionary
$\{(\phi_m,\psi_m):m\geq1\}$ whose span is dense in the unit-level
moment class, the finite-dictionary oracle variances satisfy
\(
  V_M^\star\to V^\star.
\)
If the fixed-dimensional feasible-GMM conditions hold for every fixed
$M$ and $b_n/N_n\to0$, there exists a deterministic, nondecreasing
sequence $M_n\to\infty$ such that
\(
  \frac{M_n^2b_n}{N_n}\to0
\)
and
\(
  \sqrt{N_n}
  (\hat\theta_{M_n,n}-\theta_0)
  \Rightarrow
  \Ncal(0,V^\star)
\)

\end{enumerate}
\end{theorem}

For empirical work, fixing a finite dictionary and running two-step GMM
with an estimated covariance matrix yields the optimal GMM variance within
that dictionary. As the population dictionary spaces become dense, their
oracle variances approach the within-class bound. Theorem~\ref{thm:sieve_oracle}
further establishes feasible attainment along some sufficiently slowly
growing deterministic dictionary sequence. 

\begin{algorithm}[Two-step GMM]
\label{alg:twostep-practical-scores}
Given a dictionary $\{(\phi_m,\psi_m):m=1,\ldots,M\}$, stack
$\eta_n^{(M)}(\theta):=\bigl(\eta_{1,n}(\theta),\ldots,\eta_{M,n}(\theta)\bigr)^\top$,
$\eta_{m,n}(\theta):=\frac1{N_n}\sum_{i\in\Ucal_n} \phi_m(Y_i)R_{i,\theta}\psi_m(W)$.
\begin{enumerate}[label=\textup{(\alph*)}, leftmargin=1.2cm]
\item Compute $\hat\theta_n^{(1)}\in\arg\min_{\theta\in\Theta}
  \eta_n^{(M)}(\theta)^\top\eta_n^{(M)}(\theta)$.
\item Estimate the covariance of $\sqrt{N_n}\,\eta_n^{(M)}(\theta_0)$ by a
  graph-HAC estimator $\hat\Omega_{M,n}$ (Assumption~\ref{ass:graph-hac})
  evaluated at $\hat\theta_n^{(1)}$.
\item Compute $\hat\theta_n^{(2)}\in\arg\min_{\theta\in\Theta}
  \eta_n^{(M)}(\theta)^\top\hat\Omega_{M,n}^{-1}\eta_n^{(M)}(\theta)$; then
  \[
    \sqrt{N_n}(\hat\theta_n^{(2)}-\theta_0)\Rightarrow
    \Ncal\bigl(0,(G^{(M)\top}\Omega^{(M)-1}G^{(M)})^{-1}\bigr)
  \]
  as in Theorem~\ref{thm:DB-AN-aff}, where $G^{(M)}$ and $\Omega^{(M)}$ are the
  limiting Jacobian and covariance of the stacked moments
  (Appendix~\ref{app:sieve_gmm}).
\end{enumerate}
\end{algorithm}

\section{Stage 2: Outcome regression given the exposure map}
\label{sec:stage2}

Stage~1 learns the exposure map $g(W;X_i,\hat\theta_n)$.  In most applications,
the object of ultimate interest is a downstream exposure--response or policy parameter, such as
the coefficient from an outcome regression on the learned exposure index.  This
section states the main implication for such second-step inference.

The main case is a regular exposure-map parameter $\theta\in\Theta\subset\R^p$
estimated in Stage~1.  Since $\hat\theta_n$ is an input into the downstream
exposure--response equation, the sampling uncertainty in $\hat\theta_n$ must be
propagated into second-step inference.

For concreteness, consider the linear exposure--response projection.  Define
\[
  Z_i(\theta):=
  \begin{pmatrix}
    1\\
    g(W;X_i,\theta)
  \end{pmatrix}.
\]
The target $(\alpha_0,\beta_0)$ is the design-based projection coefficient
satisfying
\begin{equation}
\label{eq:linear-main}
  \lim_{n\to\infty}
  \frac{1}{N_n}\sum_{i\in\Ucal_n}
  \E\!\left[
    Z_i(\theta_0)
    \{Y_i-\alpha_0-\beta_0^\top g(W;X_i,\theta_0)\}
  \right]
  =0.
\end{equation}
A correctly specified linear conditional mean is sufficient for
\eqref{eq:linear-main}, but the projection interpretation does not require it.

Collect the regular parameters as
\[
  \zeta=(\theta^\top,\alpha,\beta^\top)^\top .
\]
The two-step estimator is characterized by stacking the projected Stage~1 GMM
equation with the Stage~2 exposure--response equation:
\[
  G_1^\top\Lambda_{1,n}\bar\Psi_{1,n}(\hat\theta_n)=o_p(N_n^{-1/2}),
  \qquad
  \bar\Psi_{2,n}(\hat\theta_n,\hat\alpha_n,\hat\beta_n)=o_p(N_n^{-1/2}).
\]
Under a joint affinity-set CLT, mean differentiability of the stacked population map,
and stochastic equicontinuity, Appendix~\ref{app:stage2} shows that
\[
  \sqrt{N_n}(\hat\zeta_n-\zeta_0)
  \Rightarrow
  \Ncal(0,V_\zeta),
\]
where $V_\zeta$ is the corresponding design-based sandwich covariance matrix.
Inference for the exposure--response coefficient $\beta_0$ uses the
$(\beta,\beta)$ block of $V_\zeta$.  More general low-dimensional Stage~2
moments, including nonlinear regressions or policy-functional estimating
equations, are handled by the same stacked-system argument.


\begin{remark}[Design-based interpretation of the Stage-2 standard errors]
\label{rem:stage2-conservative}
\bluerev{The Stage-2 projection moments, unlike the Stage-1 moments
(Remark~\ref{rem:nonconservative-variance}), are mean zero only in aggregate:
\eqref{eq:linear-main} fixes the average moment at zero but leaves the
unit-level means nonzero, and those means are not identified from the realized
assignment (each involves $\widetilde{Y}_i$ at unrealized exposures). The
feasible sandwich must therefore center at the sample mean, so it need not
reproduce the exact design variance of $\hat\beta$. When the affinity sets are
singletons or block-diagonal the discrepancy is positive semidefinite and the
reported standard errors are conservative, with no assumption of correct
specification; it vanishes when the linear exposure--response holds unit by
unit; and it is of ambiguous sign under general overlapping dependence absent a
local-alignment condition. Appendix~\ref{subsec:beta-inference-app} makes this
precise and states the condition.}
\end{remark}

\section{Applications}
\label{sec:applications-main}

This section applies the framework to two large-scale anti-poverty programs in
development economics that explicitly study general-equilibrium (GE) effects:
(i) the Smartcards reform of India's National Rural Employment Guarantee Scheme
analyzed by \citet{muralidharan2023generaleq}; and (ii) the GiveDirectly
cash-transfer experiment in rural Kenya studied by \citet{egger2022general},
together with the follow-up structural analysis of \citet{walker2024slack}.

As alluded to in the introduction, the two settings are chosen to be similar ex ante, but we find that they yield contrasting conclusions.

\subsection{\texorpdfstring{Revisiting \citet{muralidharan2023generaleq}}{Revisiting Muralidharan et al. (2023)}}
\label{subsec:mns-main}

We begin with the general-equilibrium effects of India's
National Rural Employment Guarantee Scheme (NREGS) studied by
\citet{muralidharan2023generaleq}. The program guarantees up to 100 days of
public employment per year to rural households. \citet{muralidharan2023generaleq}
exploit a large randomized rollout of biometric ``Smartcards'' that improved
the implementation of NREGS at the mandal level, and combine direct effects on
treated mandals with spillovers to nearby untreated areas.\footnote{Administrative units nest as district $\supset$ mandal $\supset$
Gram Panchayat (GP): a mandal is a sub-district (average population roughly
62,500 in the study sample), and a GP is a village-cluster government
comprising one or more census villages. Treatment was randomized at the
mandal level, outcomes are measured at the GP (or household) level, and
spatial exposure is computed from geocoded 2001 Census village locations
\citep{muralidharan2023generaleq}. Our Stage~1 below implements the exposure
measure at the census-village level; see Appendix~\ref{app:mns-stage1}.}

\paragraph{Ring-based exposure measure.}
In their main specification (their equation~(1)),
\citet{muralidharan2023generaleq} estimate regressions of the form
\begin{equation}
  Y_{pmd}
  \;=\;
  \alpha
  \;+\;\beta_0 T_m
  \;+\;\beta_1 \mathrm{NR}_{pmd}^{20}
  \;+\; X_{pmd}'\gamma
  \;+\;\varepsilon_{pmd},
  \label{eq:mns-main-eq}
\end{equation}
where $Y_{pmd}$ is an outcome such as NREGS earnings, wage-labor income, or
total income for Gram Panchayat (GP) $p$ in mandal $m$ and district $d$,
$T_m$ is the Smartcard treatment indicator for mandal $m$, and $X_{pmd}$ are
controls. The key spillover variable $\mathrm{NR}_{pmd}^{20}$ is the share of
GPs in other mandals within 20 km of $p$ that were assigned to treatment; it is a
ring-based neighborhood-treatment measure constructed at the GP level using
geographic distance, an instance of the ring exposure map of
Example~\ref{ex:ring-egger} with radius $\theta = 20$ km.

\greenrev{The headline finding of \citet{muralidharan2023generaleq} is that 86\% of beneficiary
income gains came from non-program earnings: from higher private-sector wages and employment
rather than from NREGS payments themselves, with statistically significant spatial spillovers
of wages and employment onto nearby areas. This finding is a decomposition of the total effect
on the treated into a direct component (captured by $\beta_0$) and a neighborhood component
operating through $\mathrm{NR}^{20}_{pmd}$. The authors fix the radius entering that
decomposition a priori at 20 km, motivated by commuting distances and the geographic scale of
local labor markets, rather than estimate it.}
\greendel{The decomposition therefore rests on that choice.}
\greenrev{The decomposition therefore rests on a choice that is neither estimated nor tested
against the design.}

\paragraph{Stage~1: letting the data choose the radius.}
We apply the Stage~1 design-based GMM procedure of
Sections~\ref{sec:design-orthogonality} and~\ref{sec:DB-consistency-AN-affinity},
taking the ring radius as the tuning parameter $\theta$. The exposure measure
follows \citet{muralidharan2023generaleq}'s construction: for each GP we compute the
population-weighted fraction of non-same-mandal treated census villages within
radius $\theta$, using the same 6,662-village geometry as the original paper.
\greendel{The
design functions $\psi_i(W)$ are population-weighted annulus averages of the mandal
treatment indicator in distance bands around GP $i$, built from the same village-level distance structure as the
exposure; residualizing them as in~\eqref{eq:Ri-theta-def} gives the design-side
residuals $R_{i,\theta}\psi(W)$, and the averaged product moments
$Y_iR_{i,\theta}\psi(W)$ form the GMM criterion~\eqref{eq:design-gmm-criterion}.}
\greenrev{The design functions are population-weighted annulus averages of the mandal
treatment indicator in distance bands around GP $i$, built from the same
village-level distance structure as the exposure.}
Because physical interaction distances vary continuously in space, we treat the
radius as a continuous parameter and search over a fine 0.1 km grid. We report
Wald confidence intervals justified by the asymptotic normality result and the design-based overidentification
($J$) test of Section~\ref{sec:DB-consistency-AN-affinity}. Full technical details (the
two-step criterion, Wald interval construction, and overidentification test) are in
Appendix~\ref{app:mns-stage1}.

\greendel{Table~\ref{tab:mns-stage1-main} reports the Stage-1 radius estimates for the three
outcomes that enter \citet{muralidharan2023generaleq}'s general-equilibrium decomposition. For
each outcome it gives the selected radius $\hat\theta$, its Wald standard error and 95\%
confidence interval, and the overidentification statistic $J(\hat\theta)$ with its $p$-value.
The selected radii are 23.7 km for total income, 14.1 km for NREGS earnings, and 25.8 km for
wage-labor income. The $J$ statistics are small, and none is significant at conventional levels
($p = 0.42$, $0.54$, and $0.40$), so at the selected radii the design-based moments give no
evidence against the ring specification.}
\greenrev{Table~\ref{tab:mns-stage1-main} reports the Stage-1 estimates: the selected radii are
23.7 km for total income, 14.1 km for NREGS earnings, and 25.8 km for wage-labor income. For
each outcome the table also gives the Wald standard error and 95\% confidence interval for
$\hat\theta$, and the overidentification statistic $J(\hat\theta)$ with its $p$-value. The $J$
statistics are small and none is significant at conventional levels ($p = 0.42$, $0.54$, and
$0.40$), so at the selected radii the design-based moments give no evidence against the ring
specification.}

\greendel{Figure~\ref{fig:mns-objective-main} plots the Stage-1 criterion against the candidate
radius for each outcome, with the dashed vertical line marking the selected radius
$\hat\theta$. The criterion is what selects the radius: it is large where the ring's implied
moments are far from zero and small where they are close, so the design favors the radius that
minimizes it. In all three panels the criterion is high at short radii and falls as the radius
widens.}
\greenrev{Figure~\ref{fig:mns-objective-main} plots the Stage-1 criterion against the candidate
radius for each outcome, with the dashed vertical line marking the selected radius
$\hat\theta$; the criterion is small where the ring's implied design-based moments are close to
zero. In all three panels the criterion is high at short radii, falls to a clear interior
minimum at the selected radius, and rises beyond it, so the design identifies a finite
interaction radius for every outcome. The total-income criterion also has local minima near the
two component radii, consistent with total income aggregating a more local and a
broader-reaching margin.}

\begin{table}[htbp]
  \centering
  \caption{Stage-1 radius estimates for \citet{muralidharan2023generaleq} outcomes.}
  \label{tab:mns-stage1-main}
  \begin{tabular}{lcccc}
    \toprule
    Outcome $Y$ & $\hat\theta$ (km) & Wald SE (km) & 95\% Wald CI & \rev{$J(\hat\theta)$ ($p$-value)} \\
    \midrule
        Total income   & 23.7 & 1.22 & [21.30,\;26.10] & \rev{6.01 (0.422)} \\
    NREGS earnings  & 14.1 & 0.64 & [12.84,\;15.36] & \rev{5.05 (0.538)} \\
    Wage-labor income  & 25.8 & 1.05 & [23.75,\;27.85] & \rev{6.26 (0.395)} \\
    \bottomrule
  \end{tabular}
  \begin{flushleft}
  \footnotesize Notes: Population-weighted village-level exposure and
  instruments (3 km bands, 0.1 km radius grid).
  Wald SE and CI treat the radius as a continuous parameter; see
  Appendix~\ref{app:mns-stage1} for details.
  The $J$-statistic is the overidentification test\rev{, with $6$ degrees of freedom; $p$-values in parentheses}.
  \end{flushleft}
\end{table}

\begin{figure}[htbp]
  \centering
    \includegraphics[width=0.48\textwidth]{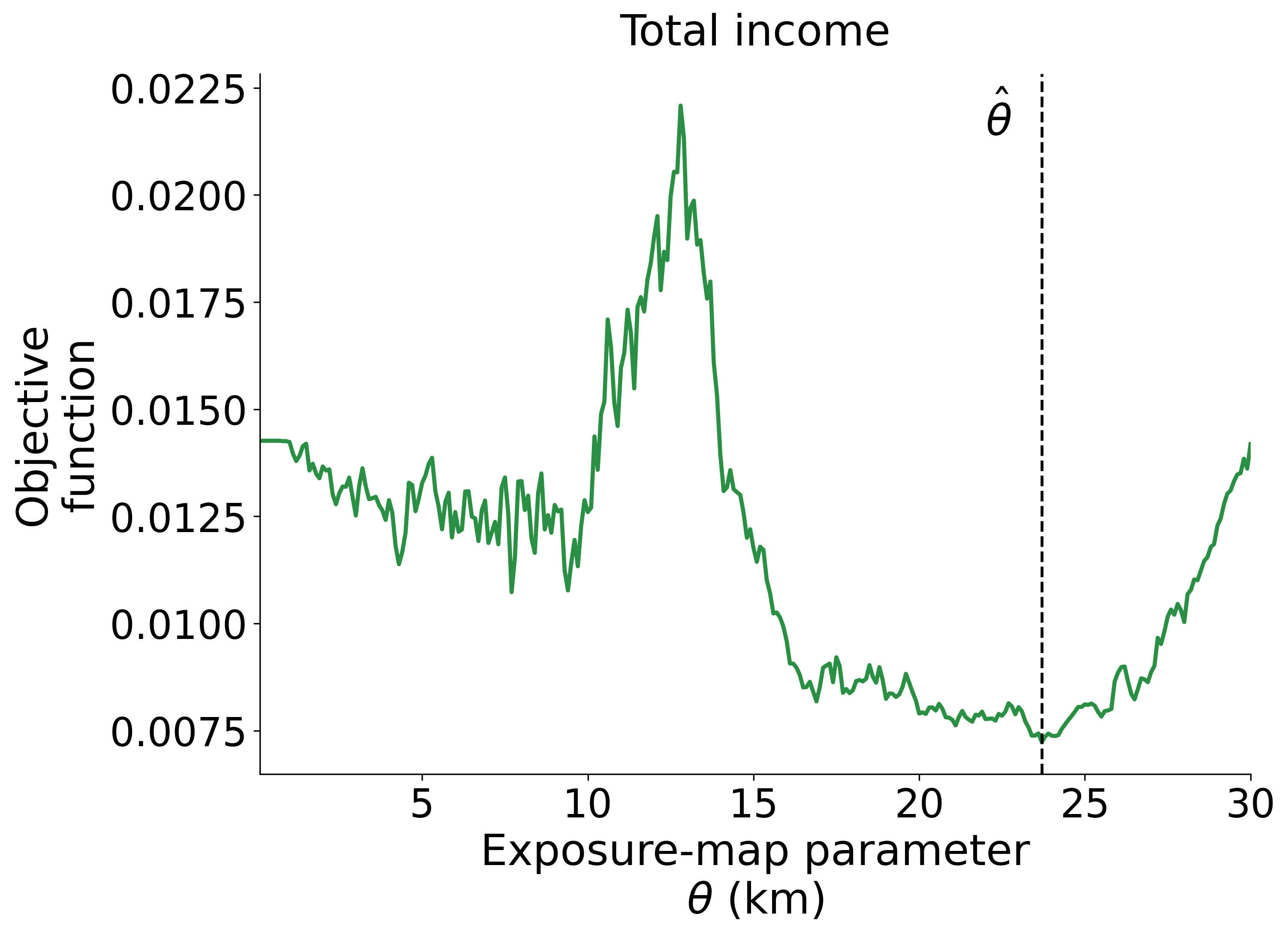}

  \vspace{0.5em}
  \includegraphics[width=0.48\textwidth]{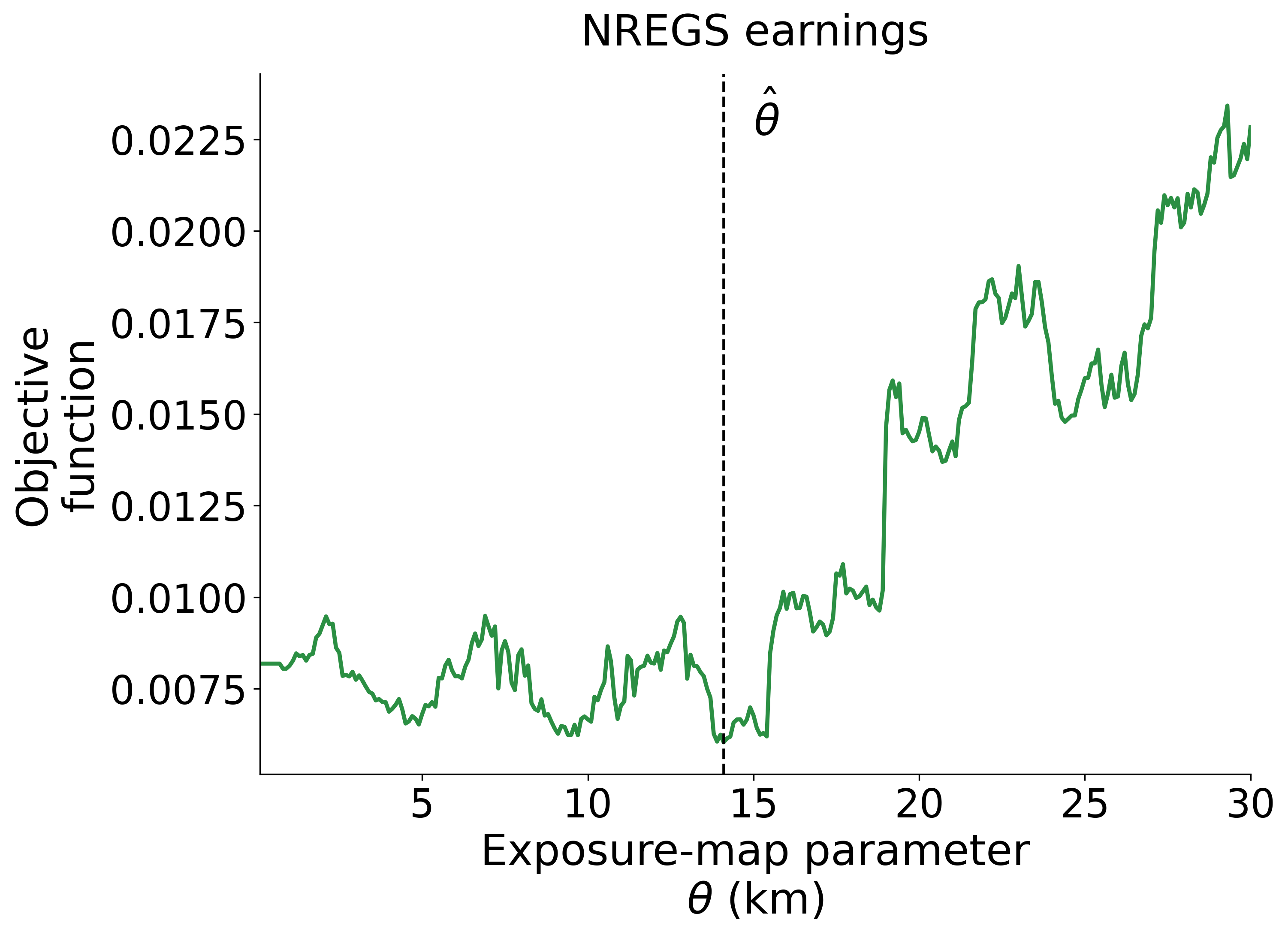}
  \includegraphics[width=0.48\textwidth]{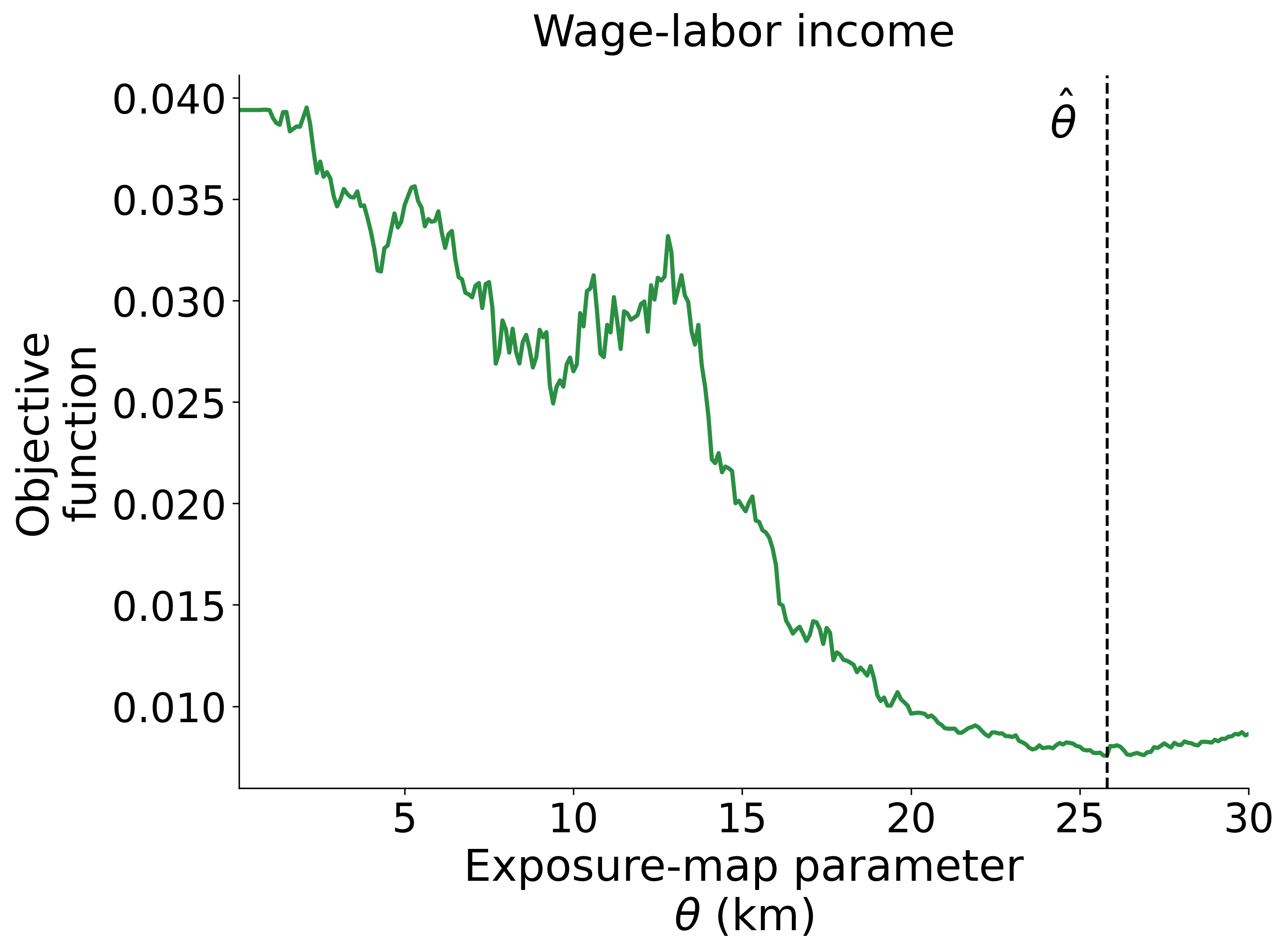}
  \caption{\rev{Stage-1 objective functions for total income (top),
  NREGS earnings (bottom left), and wage-labor income (bottom right).}
  Dashed line: selected radius $\hat\theta$. Population-weighted village-level
  exposure and instruments (3 km bands), strict placebo design.}
  \label{fig:mns-objective-main}
\end{figure}

\greendel{In this application, the design-based moments broadly support the 20 km exposure
choice that \citet{muralidharan2023generaleq} make on institutional grounds, in two ways.
First, all three objective functions have clear interior minima
(Figure~\ref{fig:mns-objective-main}), and the overidentification statistics are small at the
selected radii (last column of Table~\ref{tab:mns-stage1-main}), so the data give no strong
evidence against the ring specification. Second, and most directly relevant to the original
paper, the criterion function for total income, the headline outcome for the income-source
decomposition, is minimized at 23.7 km (first row of Table~\ref{tab:mns-stage1-main}; top panel
of Figure~\ref{fig:mns-objective-main}), close to the 20 km labor-market scale they adopt a
priori.}
\greenrev{Taken together, the exhibits broadly support the exposure choice that
\citet{muralidharan2023generaleq} make on institutional grounds: the ring specification is not
rejected for any outcome, and the estimated radius for total income, the outcome underlying the
decomposition, is 23.7 km, close to the 20 km they adopt a priori (first row of
Table~\ref{tab:mns-stage1-main}).}

\greenrev{NREGS earnings and wage-labor income decompose total income into a more local and a
broader-reaching component (rows two and three of Table~\ref{tab:mns-stage1-main}), and they
separate in the direction the authors' own mechanism predicts: NREGS earnings, the direct
public-employment margin, localize to 14.1 km, whereas wage-labor income, which operates
through private labor markets, reaches farther, to 25.8 km.\footnote{\citet{muralidharan2023generaleq}
emphasize that jobcard holders ``could only do NREGS work in their own villages,'' so the
public-employment margin should track local program access; the 14.1 km estimate has a tight
95\% confidence interval of [12.8, 15.4] km. In their mechanism for wage-labor income, better
NREGS implementation raises reservation wages and propagates through markets they describe as
``spatially integrated beyond individual GPs or even mandals.''}
The total-income radius, 23.7 km, falls between them, and the 20 km choice lies inside this
range.}

\paragraph{Stage~2: the income-source decomposition.}
Re-estimating \citet{muralidharan2023generaleq}'s outcome equations at the outcome-specific radii leaves their conclusions essentially intact (Table~\ref{tab:mns-stage2-main}). Total income
($9344.2$) and wage-labor income ($7627.9$) are approximately unchanged from the 20 km
specification; only NREGS earnings move materially, falling from $1294.6$ to
$758.7$. \citet{muralidharan2023generaleq} attribute 14\% of the income gain to
program earnings and 86\% to non-program earnings. At the design-based radii the
program share falls to 8.1\% ($758.7$ of $9344$), so non-program
earnings account for roughly 91.9\% of the gain.\footnote{Estimating sources at different radii need not yield shares that sum to one. As an alternative accounting exercise, restricting to the two sources with outcome-specific radii (NREGS and wage labor) gives a program share of 9.0\%, of the same small magnitude as the 8.1\% in the text.} Their central finding, that the
income gains come predominantly from non-program (general-equilibrium) earnings
rather than direct program payments, is thus robust to estimating the radius.

The design-based standard errors are close to their fixed-radius
counterparts despite additionally propagating the sampling uncertainty in
$\hat\theta$, so treating the radius as estimated rather than known does
not materially inflate the uncertainty around the decomposition.

\begin{table}[htbp]
\centering
\caption{\rev{Treatment effects for
\citet{muralidharan2023generaleq} outcomes: original 20\,km specification vs.\
design-consistent radii.}}
\label{tab:mns-stage2-main}
\rev{%
\resizebox{0.98\linewidth}{!}{%
\begin{tabular}{lcccc}
\toprule
Outcome & 20\,km spec.\ (MNS) & Design-based ($\hat\theta$) & $\hat\theta$ (km) & $\Delta$ vs.\ 20\,km \\
\midrule
Total income      & 9579.5 (4538.7) & 9344.2 (4969.1) & 23.7 & $-235.3$ \\
NREGS earnings    & 1294.6 (1061.0) & \phantom{0}758.7 \phantom{0}(624.6) & 14.1 & $-535.9$ \\
Wage-labor income & 7607.2 (2720.7) & 7627.9 (2996.3) & 25.8 & $+20.7$ \\
\bottomrule
\end{tabular}%
}}
\begin{flushleft}
\footnotesize
\emph{Notes:} \rev{Entries are adjusted treatment effects with standard errors in
parentheses. The ``20\,km spec.'' column is \citet{muralidharan2023generaleq}'s
original fixed-radius specification; the ``Design-based'' column re-estimates the
same Equation~\eqref{eq:mns-main-eq}-style regression (treatment, neighborhood
exposure, district fixed effects, strata controls; baseline total income
included only for total income) at the outcome-specific radius $\hat\theta$
selected in Stage~1. Standard errors are clustered at the mandal level; the
design-based standard error additionally propagates the Stage-1 Wald variance of
$\hat\theta$ via the delta method.}
\end{flushleft}
\end{table}

\subsection{Revisiting Egger et al.\ (2022)}
\label{subsec:egger-main}

\greendel{\citet{egger2022general} study the GiveDirectly cash-transfer experiment in rural Kenya,
whose two-tier randomization is what makes the exposure map estimable here: treatment was assigned
both at the village level and through a second-tier saturation design across sublocations,
generating experimental variation in indirect exposure to transfers at the village and
local-market level.}
\greenrev{\citet{egger2022general} study the GiveDirectly cash-transfer experiment in rural
Kenya, in which treatment was randomized both across villages and, through a second-tier
saturation design, across sublocations. The two-tier design generates experimental variation in
indirect exposure to transfers at the village and local-market level; this is the variation our
Stage~1 exploits below.}\footnote{See
\citet{egger2022general} for details on the experimental design and data
collection.}

\greenrev{Their main specification relates household outcomes to the amount transferred to the
household's own village and to the amount transferred to other villages within 2 km. The 2 km
outer radius, the distance beyond which transfers are assumed to have no effect, is selected
by a Bayesian information criterion over concentric 2 km distance bands. We refer to the single
included band as their scalar $0$--$2$ km spillover measure; in the language of
Section~\ref{sec:design-orthogonality}, it is a one-band member of the annular exposure-index
class we estimate below.}

\greenrev{The headline finding of \citet{egger2022general} is a local transfer multiplier of
roughly 2.5 (their Table~V), together with large positive spillovers onto non-recipient
households. The spillover component of that accounting is estimated from the $0$--$2$ km
exposure, so the multiplier inherits the support choice; and the support is selected by
in-sample predictive fit rather than estimated from, or tested against, the experimental
design.}

\paragraph{Stage~1.}
\greendel{We follow the same two-stage template as in the
\citet{muralidharan2023generaleq} application: Stage~1 estimates the exposure
map from the design-based moments, and Stage~2 runs the authors' original
outcome regression with the estimated map in place of their fixed one. The
applications differ in the exposure map being estimated.}
\greenrev{We follow the same two-stage template as in the
\citet{muralidharan2023generaleq} application.} Let
\[
A_v(W)
=
\left(
T_{v,0-2}^{\neg v}(W),
T_{v,2-4}^{\neg v}(W),
\ldots,
T_{v,18-20}^{\neg v}(W)
\right)'
\]
denote neighboring-village transfer amounts in the ten annuli up to 20 km.
\greendel{The first of these, $T_{v,0-2}^{\neg v}(W)$, is \citet{egger2022general}'s own spillover
variable, the amount transferred to other villages within $2$ km of the household's village; we
call it their scalar $0$--$2$ km spillover measure.}
\greenrev{The first of these, $T_{v,0-2}^{\neg v}(W)$, is their scalar $0$--$2$ km spillover
measure defined above.} For a
candidate \emph{support} $R_m=2m$, \rev{the radius beyond which transfers are
assumed to have no causal effect on a household,} Stage~1 forms the linear
spillover exposure index
\[
g_v^{(m)}(W;\theta_m)
=
\sum_{j=1}^{m}
\theta_{m,j}T_{v,2(j-1)-2j}^{\neg v}(W),
\qquad
\mathbf 1'\theta_m=1,
\]
and estimates $\theta_m$ from the design-based exposure-sufficiency moments.\footnote{The scale of $\theta_m$ is immaterial for
exposure sufficiency: replacing $\theta_m$ by $c\theta_m$, with $c\neq0$, only
rescales the downstream coefficient on the index. We impose $\mathbf 1'\theta_m=1$
and let the Stage~2 coefficient absorb scale. When $m=1$, this normalization
implies $\theta_{1,1}=1$. Thus the 2 km case coincides with the original
Egger et al.\ one-ring exposure measure, and Stage~1 has no nontrivial shape
parameter to estimate.} In
this notation, the exposure map used in Stage~1 is
\[
\left(
T_v^{\mathrm{own}}(W),
g_v^{(m)}(W;\theta_m)
\right).
\]
\greendel{Stage~1 estimates the annulus weights $\theta_m$: it asks which weighted
annular exposure index makes the remaining assignment variation orthogonal to
outcomes under the known randomization design, the same role the radius played in
the ring application.}

\paragraph{Stage~2.}
\greendel{Stage~2 returns to the original \citet{egger2022general} IV specification, exactly as in
the \citet{muralidharan2023generaleq} application.}
\greenrev{Stage~2 returns to \citet{egger2022general}'s original IV specification, as in
the \citet{muralidharan2023generaleq} application, replacing their scalar $0$--$2$ km spillover
measure with the estimated index; standard errors account for the generated weights
$\widehat\theta_m$ by the delta method.} We keep the same outcome equation
and instrumenting logic:
\[
Y_{iv}
=
\alpha
+
\beta_{\mathrm{own},m}T_v^{\mathrm{own}}(W)
+
\gamma_m g_v^{(m)}(W;\widehat\theta_m)
+
X_{iv}'\delta
+
\varepsilon_{iv}.
\]
\greendel{Standard errors use the delta method to account for the generated exposure
weights $\widehat\theta_m$.}

\greendel{How much the support choice matters can be read directly off the Stage-2 estimates as the
support varies. Figure~\ref{fig:egger-radius-path-core} traces the recipient and non-recipient
total effects for the three aggregate outcomes, from $R=2$ to $20$ km.}
\greendel{Two endpoints anchor the path. The leftmost point, $R=2$ km, reproduces the
original \citet{egger2022general} specification: with only the $0$--$2$
km annulus included, the normalized index collapses to their single $0$--$2$ km
ring ($\theta_{1,1}=1$). The rightmost point, $R=20$ km, is the opposite
extreme: the fully flexible index that imposes no support restriction across the
ten annuli. Each interior point reports the estimate obtained when we impose, ex
ante, that transfers beyond $R$ km do not affect the household.}
\greendel{Two patterns stand out. The recipient total effects are stable across the
path: their point estimates move
little and they mostly remain
significant throughout. This is what one would expect if the recipient response
is driven mainly by the direct, own-village transfer, which does not depend on
the spillover radius. The non-recipient (pure spillover) effects behave very
differently: as the support widens, their point estimates attenuate and their
standard errors grow, and by the fully flexible $20$ km support they are
generally no longer distinguishable from zero.
The standard errors along both paths widen and tighten unevenly rather than rising
monotonically with the support, most visibly near $16$ km; as Appendix~\ref{app:egger-extra},
part~D shows, this reflects collinearity among the annular transfer variables at particular
supports rather than the strength of spillovers at those distances.}
\greenrev{Figure~\ref{fig:egger-radius-path-core} traces the Stage-2 recipient and non-recipient
effects as the support widens from the authors' 2 km to the fully flexible 20 km: the recipient
point estimates are stable across the entire path, while the non-recipient effects attenuate
and lose precision until, by 20 km, they are generally no longer distinguishable from
zero.\footnote{The confidence intervals along both paths do not necessarily widen monotonically with the
support; Appendix~\ref{app:egger-extra}, part~D, discusses potential explanations, including
collinearity among the annular transfer variables at particular supports.} The
two endpoints anchor the comparison. At $R=2$ km the normalized index collapses to their single
$0$--$2$ km ring ($\theta_{1,1}=1$), so the leftmost point replicates the original
\citet{egger2022general} specification; at $R=20$ km the index imposes no support
restriction across the ten annuli; each interior point imposes, ex ante, that transfers beyond
$R$ km do not affect the household. The contrast between the two paths is what one would
expect: recipients respond mainly to the direct, own-village transfer, which does not depend on
the support, whereas non-recipients respond only through spillovers, which do. The lower panel
of each plot reports the $p$-value path of the support test introduced next.}

\begin{figure}[htbp]
\centering
\includegraphics[width=0.48\textwidth]{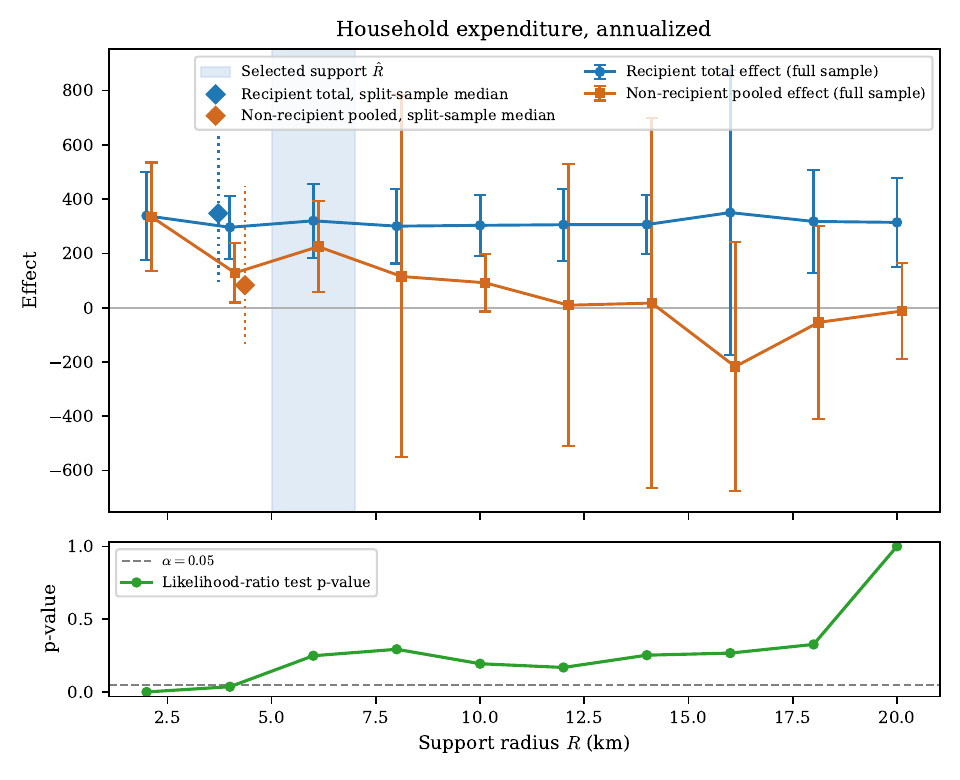}
\includegraphics[width=0.48\textwidth]{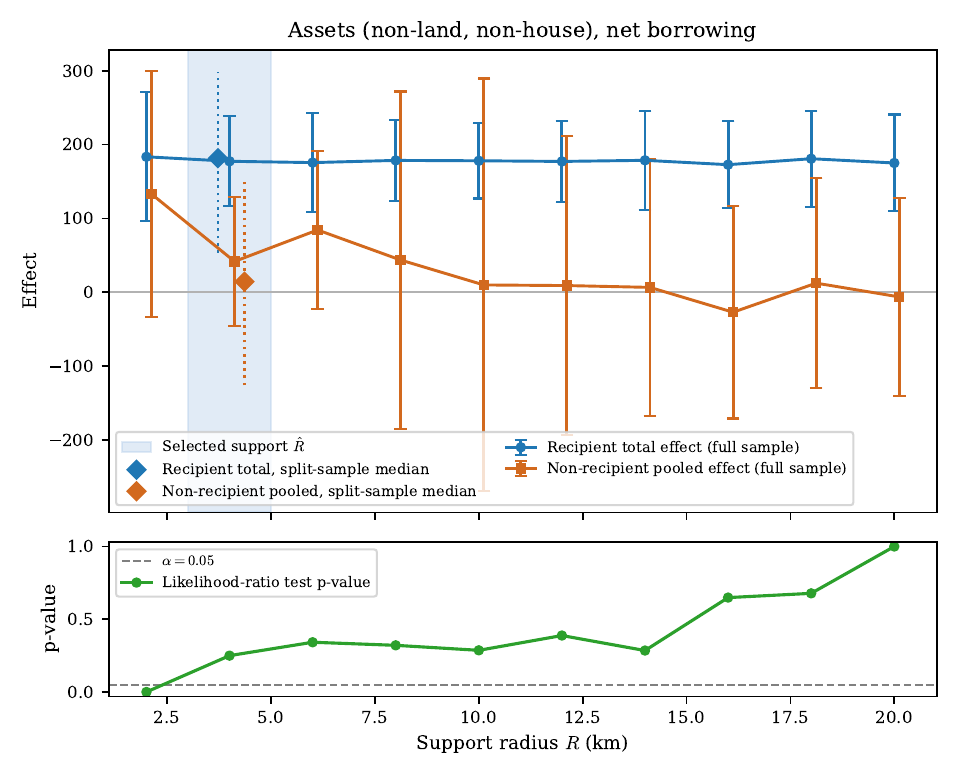}

\vspace{0.5em}

\includegraphics[width=0.52\textwidth]{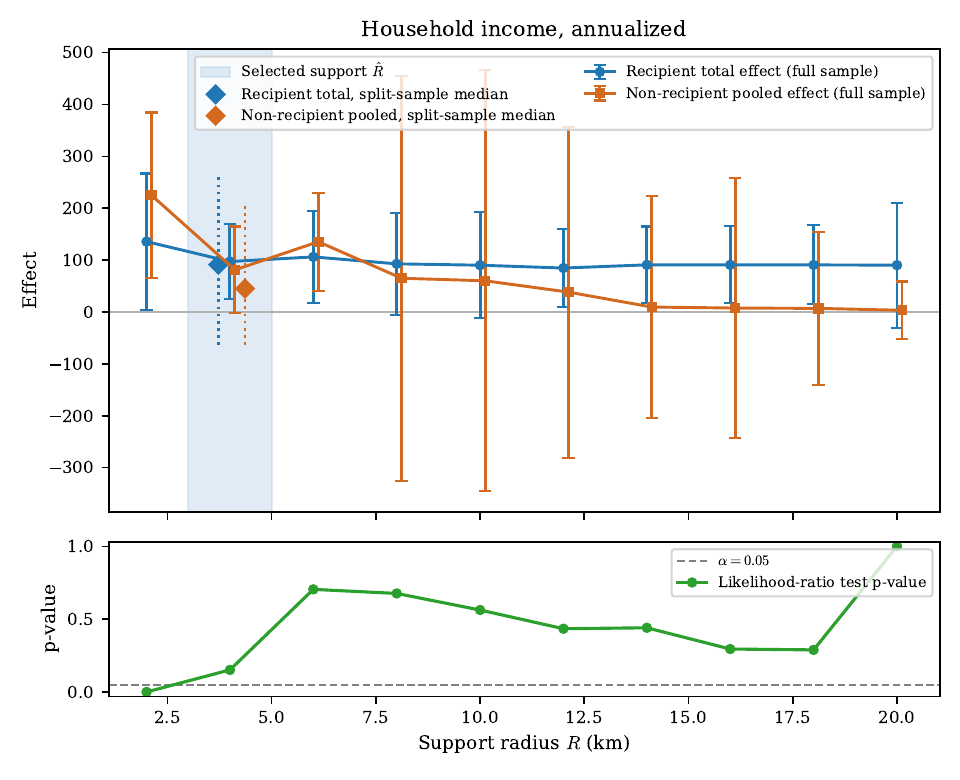}
\caption{Recipient and non-recipient effects across candidate supports,
\citet{egger2022general} application}
\label{fig:egger-radius-path-core}
\begin{flushleft}
\footnotesize
\emph{Notes:} The figure plots recipient total effects and non-recipient pooled
spillover effects across candidate supports $R=2,\ldots,20$ km. The shaded band
marks the selected support. The lower panel reports the likelihood-ratio test $p$-value
path, and the dashed horizontal line marks $\alpha=0.05$. \bluerev{Diamonds report
sublocation-level split-sample medians with their median intervals (dotted
whiskers).} The appendix reports additional outcome diagnostics.
\end{flushleft}
\end{figure}

\paragraph{\rev{Choosing the support by a design-based test.}}
\greendel{Ex ante, we would prefer to impose no a priori restriction on the support and report
the fully flexible $20$ km index. The radius path above shows, however, that once all ten annuli are left unrestricted, the pure-spillover effects are estimated too imprecisely to be informative, and at $20$ km they are indistinguishable from
zero. We therefore look for the most parsimonious support that remains adequate,
the same goal behind the Bayesian information criterion that
\citet{egger2022general} use to set their outer radius. What differs is how
adequacy is judged. Exposure sufficiency makes the support a testable
restriction: if the annuli beyond $R_m$ carry no outcome-relevant assignment
variation, zeroing their coefficients leaves the design-based moments satisfied,
which is the support form of the exposure-sufficiency restriction that the
overidentification test of Section~\ref{sec:DB-consistency-AN-affinity} (Corollary~\ref{cor:db-J-test}) evaluates. We
test the unrestricted $20$ km index against the restricted index that zeroes all
annuli beyond $R_m$ and select the smallest support the test does not reject,
implemented as a likelihood-ratio statistic along the path (see the notes to
Figure~\ref{fig:egger-radius-path-core}). Whereas the BIC rewards cross-sectional
predictive fit, the design-based test asks whether the omitted annuli are
consistent with the randomization, the identifying assumption the downstream
estimates rely on.}
\greenrev{Neither endpoint can serve as the downstream specification: 2 km is an untested
restriction, and the fully flexible 20 km index estimates the pure-spillover effects too
imprecisely to be informative. We therefore select the smallest support the design does not
reject. Exposure sufficiency makes the support a testable restriction: if the annuli beyond
$R_m$ carry no outcome-relevant assignment variation, zeroing their coefficients leaves the
design-based moments satisfied, which is the restriction the overidentification test of
Section~\ref{sec:DB-consistency-AN-affinity} (Corollary~\ref{cor:db-J-test}) evaluates. We test
the unrestricted 20 km index against the restricted index that zeroes all annuli beyond $R_m$,
implemented as a likelihood-ratio statistic along the path (lower panels of
Figure~\ref{fig:egger-radius-path-core}), and select the smallest non-rejected support. This is
the same parsimony goal behind the Bayesian information criterion \citet{egger2022general} use
to set their outer radius; the difference is that the BIC scores cross-sectional predictive
fit, whereas the design-based test scores consistency with the randomization, the identifying
assumption the downstream estimates rely on.}\footnote{As a check on post-selection and spatial-dependence concerns, the diamonds in
Figure~\ref{fig:egger-radius-path-core} report a sublocation-level sample-splitting
exercise. We split at the level of \emph{sublocations} (the higher administrative
tier at which the saturation design randomizes treatment intensity, each containing
many villages and a local market) rather than at the village or household level:
because cash-transfer spillovers operate through the local market
\citep{walker2024slack}, dependence is contained within a sublocation and is
approximately negligible across them, so assigning whole sublocations to the
selection and estimation subsamples makes the two approximately independent while
leaving the within-sublocation spillovers the exposure map captures intact. We aggregate across repeated splits
following the median procedure of \citet{chernozhukov2025fisher}: we report the median point estimate across splits and form the band
from the median lower and median upper confidence limits, yielding an approximately
$90\%$ interval.
\greendel{The split estimates are noisier but preserve the message: $2$ km is
too local and the selected supports are modestly larger.}
\greenrev{The split estimates are noisier but point the same way: $2$ km is too local and the
selected supports are modestly larger.}}

\rev{This procedure rejects the original $2$ km support for all three core outcomes\greenrev{: in
every lower panel of Figure~\ref{fig:egger-radius-path-core}, the $p$-value path lies below
$0.05$ at 2 km and \ridgerev{crosses above it by 4 km for assets and household income
and by $6$ km for household expenditure}}.
\bluerev{\ridgerev{The smallest non-rejected support is $6$ km for household expenditure
and $4$ km for assets and household income}; across the broader set of outcomes in
Appendix~\ref{app:egger-extra} the selected supports are $4$--$6$ km, occasionally larger.} The $2$ km exposure is thus too
local from the design's perspective.}

\greendel{The comparison isolates what the support choice does, because the exposure measure is the
only thing that changes. Table~\ref{tab:egger-core-IV-linear-index} sets \citet{egger2022general}'s
original estimates against the design-based estimates at the selected supports $\hat R$: their
scalar $0$--$2$ km spillover measure is replaced by the estimated index at $\hat R$, and nothing
else is altered.}
\greendel{The comparison echoes the radius path, and it separates the two margins sharply.
The recipient total effects survive the change of exposure map: all three point estimates move by
less than half of \citet{egger2022general}'s own standard error. In contrast, the
non-recipient side is where the exposure choice bites. Every pooled spillover attenuates by about one
to one and a half of \citet{egger2022general}'s standard errors: expenditure from $334.7$ to $128.9$,
assets from $135.4$ to $41.9$, and income from $225.0$ to $81.4$, the last no longer distinguishable
from zero at the $5\%$ level. Relative to the \citet{muralidharan2023generaleq} application, then, the
exposure-map choice does matter here.}
\greenrev{Table~\ref{tab:egger-core-IV-linear-index} isolates what the support choice does,
since the exposure measure is the only thing that changes between the two columns; the change
bites on the non-recipient margin only. The recipient total effects survive the change of
exposure map: all three point estimates move by less than half of \citet{egger2022general}'s
own standard error. The non-recipient side is where the exposure choice bites:
every pooled spillover attenuates by about one to one and a half of
\citet{egger2022general}'s standard errors. Expenditure falls from $334.7$ to \ridgerev{$224.9$} and
remains significant; income falls from $225.0$ to \ridgerev{$81.1$} and is no longer distinguishable from
zero at the $5\%$ level; assets, not significant in the original specification, falls from
\ridgerev{$133.1$ to $41.6$}.\footnote{The ``Egger IV'' column re-estimates the original 2 km
specification on our replication sample; the estimates closely match the published Table~1
values (see the table notes).} Relative to the \citet{muralidharan2023generaleq} application,
then, the exposure-map choice does matter here.}

\greendel{Two different tests appear in this comparison, and it is worth keeping
them apart. The supports $\hat R$ are chosen by the likelihood-ratio test
described above: $\hat R$ is the smallest support that test does not reject.
The last column of Table~\ref{tab:egger-core-IV-linear-index} instead reports a
standard GMM overidentification ($J$) test of the selected linear-index
specification, and the $J$ test does reject for expenditure and assets.
The two findings are compatible: an ambient, market-level component that is
common to households within a local market is not spanned by an annular index
at any support, which is precisely the mechanism of \citet{walker2024slack}
discussed below. The residual $J$ rejection therefore sharpens, rather than
muddies, the contrast with the \citet{muralidharan2023generaleq} application,
where the ring class is not rejected at the selected radii.}
\greenrev{The last column of Table~\ref{tab:egger-core-IV-linear-index} reports the
overidentification ($J$) test of the selected linear-index specification. The test rejects for
expenditure and assets (\ridgerev{$p = 0.007$} and $0.001$) and does not reject for income
(\ridgerev{$p = 0.134$}).}

\begin{table}[htbp]
  \centering
  \caption{Core IV components: Egger et al.\ vs. design-based estimates}
  \label{tab:egger-core-IV-linear-index}
  \footnotesize
  \setlength{\tabcolsep}{5pt}
  \begin{tabular}{lcccc}
    \toprule
     & Egger IV & Design-based IV & $\Delta\beta,\ \Delta$SE & $J$ p-val \\
    \midrule
    \multicolumn{5}{l}{\emph{Household expenditure, annualized}} \\
    \quad Recipients     & 338.6 (100.6) & 320.0 (69.4) & $-18.6,\ -31.2$ & 0.007 \\
    \quad Non-recipients & 334.7 (131.1) & 224.9 (85.2) & $-109.8,\ -45.9$ & 0.007 \\[0.3em]
    \multicolumn{5}{l}{\emph{Assets (non-land, non-house), net borrowing}} \\
    \quad Recipients     & 183.4 (55.2) & 177.3 (31.2) & $-6.1,\ -24.0$ & 0.001 \\
    \quad Non-recipients & 133.1 (101.3) & 41.6 (44.8) & $-91.5,\ -56.5$ & 0.001 \\[0.3em]
    \multicolumn{5}{l}{\emph{Household income, annualized}} \\
    \quad Recipients     & 135.7 (82.1) & 97.2 (36.9) & $-38.5,\ -45.1$ & 0.134 \\
    \quad Non-recipients & 225.0 (100.1) & 81.1 (42.6) & $-143.8,\ -57.4$ & 0.134 \\
    \bottomrule
  \end{tabular}
  \begin{flushleft}
  \footnotesize\emph{Notes:} Entries are IV point estimates with standard errors in parentheses. Egger IV values re-estimate the original 2~km Table~1 specification on our replication sample (cluster standard errors); they closely match the published Table~I values. Design-based IV values use each outcome's selected design-based support ($\hat R = 6.0$~km for household expenditure, $4.0$~km for assets and income); the design-side conditional expectations entering the index are estimated by cross-validated ridge regression on a quadratic polynomial in the exposure index, using $200$ placebo draws of the assignment. Recipient entries report total effects; non-recipient entries report pooled spillover effects, both with joint delta-method standard errors. $J$ p-val is the design-GMM overidentification p-value for the selected support; the support-selection rule is the likelihood-ratio test.
  \end{flushleft}
\end{table}


\greendel{This pattern has a natural interpretation, and it is one a fit-based radius rule is bound
to miss.}
\greenrev{The spillover attenuation and the residual $J$ rejection are both consistent with a
market-level ``ambient'' component in the spillovers, a component that a fit-based radius rule
misses.} \citet{walker2024slack} argue that
local GE responses include a market-level ``ambient'' or ``slack'' component:
\[
Y_i
\;=\;
\alpha
\;+\; \beta^{\mathrm{loc}} g^{\mathrm{loc}}(W;X_i,R_0)
\;+\; \underbrace{A_{m(i)}(W)}_{\text{market-level ambient}}
\;+\; u_i,
\]
where $g^{\mathrm{loc}}$ captures very local spillovers and $A_{m(i)}(W)$ is
common to households in local market $m(i)$.
\greendel{\citet{egger2022general} choose the radius $r$ by
minimizing $\mathrm{BIC}(r)=N\log\hat\sigma_r^2+k(r)\log N$, where $\hat\sigma_r^2$
is the residual variance from regressing $Y_i$ on the radius-$r$ ring exposures and
controls, $k(r)$ is the number of parameters in that regression, and $N$ is the
number of villages; the radius is thus scored by how much the rings reduce residual
variance, traded off against the penalty $k(r)\log N$. But a component that is
approximately common within a market is largely absorbed by the intercept and
controls, so outer rings that mainly proxy $A_{m(i)}(W)$ yield only small reductions
in $\hat\sigma_r^2$ and BIC prefers short radii,
even when the beyond-2\,km transfers generating that component move the total GE
response.}
\greenrev{A component that is approximately common within a market is largely absorbed by the
intercept and controls, so outer annuli that mainly proxy $A_{m(i)}(W)$ barely reduce residual
variance and the BIC prefers short radii, even when the beyond-2 km transfers generating that
component move the total general-equilibrium response; and no annular index spans such a
component at any support, which is why the $J$ test still rejects at $\hat
R$.\footnote{\citet{egger2022general} choose the radius $r$ by minimizing
$\mathrm{BIC}(r)=N\log\hat\sigma_r^2+k(r)\log N$, where $\hat\sigma_r^2$ is the residual
variance from regressing $Y_i$ on the radius-$r$ ring exposures and controls, $k(r)$ is the
number of parameters, and $N$ is the number of villages: the radius is scored by how much the
rings reduce residual variance, traded off against the penalty.}}
The design-based diagnostic instead scores the radius by consistency with
the randomization rather than by residual fit, so it retains those annuli and
extends the support. In addition, broader market-level forces such as inflation or
price responses can offset part of the measured income and non-recipient gains.
Together these explain why the support diagnostic moves beyond 2\,km even as some
downstream estimates become smaller or less precise.

\paragraph{Implications for multiplier accounting.}
\greendel{Finally, the support choice feeds into the local fiscal multiplier reported
in \citet{egger2022general}'s Table~V. Table~\ref{tab:egger-multiplier-results} recomputes that accounting with the
design-selected indexes, giving \bluerev{an expenditure multiplier of about $1.6$, an
income multiplier of about $1.5$, and a mean multiplier of about $1.57$, well
below the original headline value near $2.5$.}}
\greenrev{Recomputing \citet{egger2022general}'s Table~V multiplier accounting with the
design-selected indexes gives a mean local transfer multiplier of about \ridgerev{$1.60$}, well below the
original headline value near $2.5$ and close to the $1.54$ implied by
\citet{walker2024slack}'s calibrated general-equilibrium model
(Table~\ref{tab:egger-multiplier-results}). The expenditure and income components are about
\ridgerev{$1.63$ and $1.57$}.} We obtain standard errors by a
delta method that propagates both the Stage~1 uncertainty in the index weights
and the Table~V coefficient uncertainty. With only $84$ sublocation clusters
these multipliers are imprecisely estimated, with standard errors comparable to
\citet{egger2022general}'s own; the mean multiplier carries a standard error of
\bluerev{\ridgerev{$1.10$}}, so although its point estimate lies well below $2.5$, it is not
statistically distinguishable from one.
\greendel{This revised magnitude lines
up with the structural follow-up of \citet{walker2024slack}, whose calibrated
general-equilibrium model implies a real multiplier of $1.54$, on the same real
(deflated) basis as \citet{egger2022general}'s Table~V.}
\greenrev{The \citet{walker2024slack} benchmark is on the same real (deflated) basis as
\citet{egger2022general}'s Table~V.} Appendix~\ref{app:egger-extra}, part~E
details the reparameterization and the delta-method inference.

\begin{table}[htbp]
  \centering
  \caption{\rev{Local transfer multiplier: design-based estimates and benchmarks}}
  \label{tab:egger-multiplier-results}
  \rev{%
  \begin{tabular}{lccc}
    \toprule
    & Design-based & \citet{egger2022general} & \citet{walker2024slack} \\
    &  & Table~V & model  \\
    \midrule
    Expenditure & \bluerev{\ridgerev{1.63 (1.49)}} & 2.58 (1.44) & -- \\
    Income      & \bluerev{\ridgerev{1.57 (0.97)}} & 2.47 (1.71) & -- \\
    Mean        & \bluerev{\ridgerev{1.60 (1.10)}} & 2.52 (1.39) & 1.54 \\
    \bottomrule
  \end{tabular}}
  \begin{flushleft}
  \footnotesize \rev{\emph{Notes:} Entries are multiplier point estimates with
  standard errors in parentheses, all on a real basis: transfers and
  outcomes are deflated to January~2015 USD, matching \citet{egger2022general}'s
  Table~V. Our design-based estimates recompute that deflated accounting after
  replacing each ring exposure with the design-selected index; their standard
  errors are delta-method values that propagate the Stage-1 uncertainty in the
  index weights and the Table~V coefficient uncertainty (Appendix~\ref{app:egger-extra}, part~E). The \citet{egger2022general} column reproduces their Table~V
  estimates, whose standard errors come from a wild clustered bootstrap. The \citet{walker2024slack} figure is a
  model-implied real multiplier from a calibrated general-equilibrium model, with
  no sampling standard error and reported only in aggregate (their nominal
  counterpart is $1.83$).}
  \end{flushleft}
\end{table}

\greenrev{The second application thus illustrates the revision use of the framework: when the
design rejects the original exposure map, both the map and the headline policy estimate change
materially.}

\section{Conclusion}
\label{sec:conclusion}
Estimates of spillover effects depend on an exposure map that turns the realized
assignment into an index of each unit's exposure. Both the functional form of
that map (e.g., a ring) and its parameter (e.g., the radius) are usually fixed
before estimation, with little guidance for either choice. This paper shows that
the same randomization that identifies treatment effects can guide both choices. Specifically, a correctly specified exposure map implies
orthogonality conditions that can be used for both estimation and testing:
design-based GMM can estimate the mapping parameter, and the same conditions, when
overidentified, can test the map itself.

We establish consistency and asymptotic normality of the resulting estimator for
the maps used in applied work, with all randomness arising from the assignment,
and we characterize the efficient moments. Finally, we carry the resulting
uncertainty about $g(W;X_i,\theta)$ into downstream policy estimands (e.g.,
average effects under alternative assignment rules), so that inference reflects
both the experimental variation and the uncertainty about the exposure
specification itself.

We apply the method to two large-scale anti-poverty experiments. For the NREGS
reform studied by \citet{muralidharan2023generaleq}, the 20~km radius chosen on
institutional grounds is not rejected by the design, and the program's headline
decomposition is robust: even at the data-chosen radius, the income gains come
predominantly from non-program earnings. For the cash-transfer experiment of
\citet{egger2022general}, the 2~km support is rejected for every core outcome,
and replacing it with the smallest support the design does not reject lowers the
estimated local transfer multiplier (mean) from about $2.5$ to \bluerev{\ridgerev{$1.60$}}, close to the
value implied by the independent structural model of \citet{walker2024slack}.

\bibliographystyle{apalike}
\bibliography{references}
\newpage
\appendix
\section{\texorpdfstring{Proofs for Section \ref{sec:design-orthogonality}}{Proofs for the design-orthogonality results}}

\begin{proof}[Proof of Theorem \ref{thm:exposure-sufficiency} and Corollary \ref{cor:design-moment-unit}]
Fix a unit \(i\), and write
\[
  G_i(\theta_0) := g(W;X_i,\theta_0).
\]
By Hypothesis~\ref{hyp:exposure-suff}, there exists \(\theta_0\in\Theta\) and a
measurable map \(\widetilde Y_i:\R^k\to\R\) such that, for every assignment
\(w\in\mathcal{W}\),
\[
  Y_i(w)
  =
  \widetilde Y_i\bigl(g(w;X_i,\theta_0)\bigr).
\]
Under the design \(\Dcal_n\), the realized outcome is \(Y_i=Y_i(W)\). Hence
\[
  Y_i
  =
  \widetilde Y_i\bigl(g(W;X_i,\theta_0)\bigr)
  =
  \widetilde Y_i\bigl(G_i(\theta_0)\bigr)
  \quad\text{a.s.}
\]
By Assumption~\ref{ass:rand-indep}, the potential-outcome schedule is fixed
under $\Dcal_n$, so each map $\widetilde Y_i$ is nonrandom. Thus \(Y_i\) is measurable with respect to
\(\sigma(G_i(\theta_0))\).

We first verify the conditional independence statement. Let
\(f:\R\to\R\) and \(h:\mathcal{W}\to\R\) be bounded measurable functions. Since
\(Y_i=\widetilde Y_i(G_i(\theta_0))\) almost surely,
\[
  f(Y_i)
  =
  f\!\left(\widetilde Y_i(G_i(\theta_0))\right)
\]
is \(\sigma(G_i(\theta_0))\)-measurable. Therefore,
\begin{align*}
  \E\!\left[
    f(Y_i)h(W)
    \,\big|\,
    G_i(\theta_0)
  \right]
  &=
  f\!\left(\widetilde Y_i(G_i(\theta_0))\right)
  \E\!\left[
    h(W)
    \,\big|\,
    G_i(\theta_0)
  \right]  \\
  &=
  \E\!\left[
    f(Y_i)
    \,\big|\,
    G_i(\theta_0)
  \right]
  \E\!\left[
    h(W)
    \,\big|\,
    G_i(\theta_0)
  \right].
\end{align*}
This verifies
\[
  Y_i \;\perp\!\!\!\perp\; W
  \;\big|\;
  g(W;X_i,\theta_0)
\]
under the design distribution.

We next verify the orthogonality condition. Let
\(\phi:\R\to\R\) \redrev{be a square-integrable outcome transformation, \(\E[\phi(Y_i)^2]<\infty\), and let
\(\psi:\mathcal{W}\times\mathcal{X}\to\R\) be a square-integrable design function, \(\E[\psi(W)^2]<\infty\)}. By definition,
\[
  R_{i,\theta_0}{\psi}(W)
  =
  \psi(W)
  -
  \E\!\left[
    \psi(W)
    \mid
    G_i(\theta_0)
  \right].
\]
Hence
\[
  \E\!\left[
    R_{i,\theta_0}{\psi}(W)
    \mid
    G_i(\theta_0)
  \right]
  =
  0
  \quad\text{a.s.}
\]
Moreover, since \(Y_i\) is \(\sigma(G_i(\theta_0))\)-measurable,
\(\phi(Y_i)\) is also \(\sigma(G_i(\theta_0))\)-measurable.
\redrev{Because conditional expectation is an \(L^2\) contraction,
\(\psi(W)\in L^2\) implies
\(R_{i,\theta_0}\psi(W)=\psi(W)-\E[\psi(W)\mid G_i(\theta_0)]\in L^2\); together with
\(\phi(Y_i)\in L^2\), the Cauchy--Schwarz inequality gives
\(\E\big[\,|\phi(Y_i)\,R_{i,\theta_0}\psi(W)|\,\big]
\le \|\phi(Y_i)\|_{L^2}\,\|R_{i,\theta_0}\psi(W)\|_{L^2}<\infty\),
so the product is integrable and the law of iterated expectations applies below.}
The law of iterated
expectations gives
\begin{align*}
  \E\!\left[
    \phi(Y_i)R_{i,\theta_0}{\psi}(W)
  \right]
  &=
  \E\!\left[
    \E\!\left[
      \phi(Y_i)R_{i,\theta_0}{\psi}(W)
      \,\big|\,
      G_i(\theta_0)
    \right]
  \right] \\
  &=
  \E\!\left[
    \phi(Y_i)
    \E\!\left[
      R_{i,\theta_0}{\psi}(W)
      \,\big|\,
      G_i(\theta_0)
    \right]
  \right] \\
  &=
  0.
\end{align*}
Equivalently,
\[
  \E\!\left[
    \phi(Y_i)
    \left\{
      \psi(W)
      -
      \E\!\left[
        \psi(W)
        \mid
        g(W;X_i,\theta_0)
      \right]
    \right\}
  \right]
  =
  0,
\]
which is the desired unit-level design-based orthogonality condition.

Since \(i\), \(\phi\), and \(\psi\) were arbitrary, the result holds for every unit and
every admissible outcome transformation and design function. This completes the proof.
\end{proof}
\bigskip

\section{\texorpdfstring{Additional Results and Proofs for Section~\ref{sec:DB-consistency-AN-affinity}}{Additional results and proofs for design-based consistency and asymptotic normality}}
\label{app:DB-GMM-consistency-affinity}

Throughout this appendix we work under the notation and conditions of
Section~\ref{sec:DB-consistency-AN-affinity}.
In particular, for each $n$ and $i\in\Ucal_n$ we define, for $\theta\in\Theta$,
\[
  Z_{i,n}(\theta)
  := \Psi_i(\theta)-\E[\Psi_i(\theta)],
  \qquad
  \bar Z_n(\theta)
  := \frac{1}{N_n}\sum_{i=1}^{N_n} Z_{i,n}(\theta),
\]
and
\[
  \mu_n(\theta)
  :=
  \frac{1}{N_n}\sum_{i=1}^{N_n}\E[\Psi_i(\theta)],
  \qquad
  \mu(\theta) := \lim_{n\to\infty}\mu_n(\theta),
\]
whenever the limit exists.
The GMM criterion and its limit are
\[
  Q_n(\theta)
  :=
  \bar\Psi_n(\theta)^\top \Lambda_n\,\bar\Psi_n(\theta),
  \qquad
  Q(\theta)
  :=
  \mu(\theta)^\top \Lambda\,\mu(\theta),
\]
with $\Lambda_n\pto\Lambda\succeq 0$.
Assumption~\ref{ass:as-lln} in the main text encodes the high-level uniform LLN,
identification, and weight convergence conditions needed for consistency.

The main text focuses on the regular continuous exposure-map case.
This appendix gives the formal supporting material for the baseline regular
continuous case and then records the finite-grid variants used in
applications. 

\subsection{Consistency of the baseline regular case}
\paragraph{Uniform convergence of the criterion.}

The key step is to show that the high-level uniform LLN for the moments implies
uniform convergence of the quadratic criterion $Q_n(\theta)$ to $Q(\theta)$.

\begin{lemma}[Uniform convergence of the quadratic criterion]
\label{lem:as-Q-unif}
Suppose Assumption~\ref{ass:as-lln} holds.
Then
\[
  \sup_{\theta\in\Theta}\big|Q_n(\theta)-Q(\theta)\big| \;\pto\; 0.
\]
\end{lemma}

\begin{proof}
Define, for each $n$ and $\theta\in\Theta$,
\[
  \Delta_n(\theta)
  := \bar\Psi_n(\theta)-\mu_n(\theta)
  = \bar Z_n(\theta).
\]
Then
\[
  \bar\Psi_n(\theta)
  = \mu_n(\theta) + \Delta_n(\theta).
\]
Substituting into $Q_n(\theta)$ yields
\begin{align*}
  Q_n(\theta)
  &= \bar\Psi_n(\theta)^\top \Lambda_n \bar\Psi_n(\theta) \\
  &= \big(\mu_n(\theta)+\Delta_n(\theta)\big)^\top
      \Lambda_n
      \big(\mu_n(\theta)+\Delta_n(\theta)\big) \\
  &=
  \Delta_n(\theta)^\top \Lambda_n \Delta_n(\theta)
  + 2\,\mu_n(\theta)^\top \Lambda_n \Delta_n(\theta)
  + \mu_n(\theta)^\top \Lambda_n \mu_n(\theta).
\end{align*}
Subtracting $Q(\theta)=\mu(\theta)^\top \Lambda\mu(\theta)$ and adding and
subtracting $\mu_n(\theta)^\top \Lambda\mu_n(\theta)$ gives
\begin{align}
\label{eq:QnQ-decomp-app}
  Q_n(\theta) - Q(\theta)
  &=
  \underbrace{\Delta_n(\theta)^\top \Lambda_n \Delta_n(\theta)}_{T_{1,n}(\theta)}
  + \underbrace{2\,\mu_n(\theta)^\top \Lambda_n \Delta_n(\theta)}_{T_{2,n}(\theta)} \nonumber\\
  &\quad
  + \underbrace{\mu_n(\theta)^\top (\Lambda_n - \Lambda)\mu_n(\theta)}_{T_{3,n}(\theta)}
  + \underbrace{\big\{\mu_n(\theta)^\top \Lambda\mu_n(\theta)
         - \mu(\theta)^\top \Lambda\mu(\theta)\big\}}_{T_{4,n}(\theta)}.
\end{align}

We bound each term $T_{k,n}(\theta)$ uniformly in $\theta\in\Theta$.

\medskip\noindent
\textbf{First term.}
For each $\theta$,
\[
  \big|\Delta_n(\theta)^\top \Lambda_n \Delta_n(\theta)\big|
  \le \|\Lambda_n\|_{\op}\,\|\Delta_n(\theta)\|^2,
\]
hence
\[
  \sup_{\theta\in\Theta}
  \big|\Delta_n(\theta)^\top \Lambda_n \Delta_n(\theta)\big|
  \le \|\Lambda_n\|_{\op}
  \Big(\sup_{\theta\in\Theta}\|\Delta_n(\theta)\|\Big)^2.
\]
By Assumption~\ref{ass:as-lln}\,\textup{(AS--LLN2)},
\[
  \sup_{\theta\in\Theta}\|\Delta_n(\theta)\|
  =
  \sup_{\theta\in\Theta}\big\|\bar\Psi_n(\theta)-\mu_n(\theta)\big\|
  \;\pto\; 0,
\]
and by Assumption~\ref{ass:as-lln}\,\textup{(AS--LLN4)},
$\Lambda_n\pto\Lambda$ implies $\|\Lambda_n\|_{\op}=O_p(1)$.
Therefore
\[
  \sup_{\theta\in\Theta}
  \big|T_{1,n}(\theta)\big|
  = o_p(1).
\]

\medskip\noindent
\textbf{Second term.}
Similarly,
\[
  \big|\mu_n(\theta)^\top \Lambda_n \Delta_n(\theta)\big|
  \le \|\mu_n(\theta)\|\,\|\Lambda_n\|_{\op}\,\|\Delta_n(\theta)\|,
\]
so
\[
  \sup_{\theta\in\Theta}
  \big|T_{2,n}(\theta)\big|
  \le
  2\Big(\sup_{\theta\in\Theta}\|\mu_n(\theta)\|\Big)
    \|\Lambda_n\|_{\op}
    \Big(\sup_{\theta\in\Theta}\|\Delta_n(\theta)\|\Big).
\]
By Assumption~\ref{ass:as-lln}\,\textup{(AS--LLN1)},
$\mu_n(\theta)\to\mu(\theta)$ uniformly on $\Theta$.
Continuity of $\mu$ and compactness of $\Theta$ imply
\[
  \sup_{\theta\in\Theta}\|\mu(\theta)\| < \infty,
  \qquad
  \sup_{\theta\in\Theta}\|\mu_n(\theta)\| = O(1).
\]
Combining with Assumption~\ref{ass:as-lln}\,\textup{(AS--LLN2)} and
\textup{(AS--LLN4)}, we obtain
\[
  \sup_{\theta\in\Theta}
  \big|T_{2,n}(\theta)\big|
  = o_p(1).
\]

\medskip\noindent
\textbf{Third term.}
For each $\theta$,
\[
  \big|\mu_n(\theta)^\top (\Lambda_n-\Lambda)\mu_n(\theta)\big|
  \le \|\mu_n(\theta)\|^2\,\|\Lambda_n-\Lambda\|_{\op},
\]
hence
\[
  \sup_{\theta\in\Theta}
  \big|T_{3,n}(\theta)\big|
  \le
  \Big(\sup_{\theta\in\Theta}\|\mu_n(\theta)\|\Big)^2
  \|\Lambda_n-\Lambda\|_{\op}.
\]
By Assumption~\ref{ass:as-lln}\,\textup{(AS--LLN4)},
$\|\Lambda_n-\Lambda\|_{\op}\pto0$, and as above
$\sup_{\theta\in\Theta}\|\mu_n(\theta)\|=O(1)$.
Thus
\[
  \sup_{\theta\in\Theta}
  \big|T_{3,n}(\theta)\big|
  = o_p(1).
\]

\medskip\noindent
\textbf{Fourth term.}
For each $\theta$,
\begin{align*}
  &\mu_n(\theta)^\top \Lambda\mu_n(\theta)
  - \mu(\theta)^\top \Lambda\mu(\theta) \\
  &\quad=
  \big(\mu_n(\theta)-\mu(\theta)\big)^\top \Lambda\mu_n(\theta)
  + \mu(\theta)^\top \Lambda\big(\mu_n(\theta)-\mu(\theta)\big),
\end{align*}
so
\begin{align*}
  &\big|\mu_n(\theta)^\top \Lambda\mu_n(\theta)
  - \mu(\theta)^\top \Lambda\mu(\theta)\big| \\
  &\qquad\le
  \|\mu_n(\theta)-\mu(\theta)\|\,\|\Lambda\|_{\op}\,\|\mu_n(\theta)\|
  +
  \|\mu(\theta)\|\,\|\Lambda\|_{\op}\,\|\mu_n(\theta)-\mu(\theta)\|.
\end{align*}
Taking suprema over $\theta\in\Theta$ and using
Assumption~\ref{ass:as-lln}\,\textup{(AS--LLN1)}, together with boundedness of
$\sup_{\theta}\|\mu_n(\theta)\|$ and $\sup_{\theta}\|\mu(\theta)\|$, we obtain
\[
  \sup_{\theta\in\Theta}
  \big|T_{4,n}(\theta)\big|
  =
  \sup_{\theta\in\Theta}
  \big|\mu_n(\theta)^\top \Lambda\mu_n(\theta)
  - \mu(\theta)^\top \Lambda\mu(\theta)\big|
  \to 0.
\]

\medskip

Combining the four bounds with \eqref{eq:QnQ-decomp-app}, we conclude that
\[
  \sup_{\theta\in\Theta}\big|Q_n(\theta)-Q(\theta)\big|
  \le
  \sum_{k=1}^4 \sup_{\theta\in\Theta}|T_{k,n}(\theta)|
  = o_p(1).
\]
This establishes the desired uniform convergence.
\end{proof}

\paragraph{\texorpdfstring{Argmin consistency and proof of Theorem~\ref{thm:as-consistency}.}{Argmin consistency and proof of the consistency theorem.}}

We now use Lemma~\ref{lem:as-Q-unif} and the identification part of
Assumption~\ref{ass:as-lln} to show that any minimizer of $Q_n(\theta)$ is
consistent.

\begin{lemma}[Argmin consistency]
\label{lem:as-argmin}
Suppose Assumption~\ref{ass:as-lln} holds and
Lemma~\ref{lem:as-Q-unif} is in force.
Then any
$\hat\theta_n\in\arg\min_{\theta\in\Theta} Q_n(\theta)$ satisfies
$\hat\theta_n \pto \theta_0$.
\end{lemma}

\begin{proof}
By Assumption~\ref{ass:as-lln}\,\textup{(AS--LLN3)},
$\Theta$ is compact and $\mu(\theta)$ is continuous in $\theta$ on $\Theta$.
Since $\Lambda\succeq0$ is fixed,
$Q(\theta)=\mu(\theta)^\top \Lambda\mu(\theta)$ is
continuous on $\Theta$.
Assumption~\ref{ass:as-lln}\,\textup{(AS--LLN3)} further implies that $Q$ is
uniquely minimized at $\theta_0$.

Fix $\eta>0$ and define
\[
  S_\eta := \{\theta\in\Theta : \|\theta-\theta_0\|\ge\eta\}.
\]
Compactness of $\Theta$ implies $S_\eta$ is compact.
By continuity and uniqueness of the minimizer,
\[
  \delta_\eta
  :=
  \inf_{\theta\in S_\eta}\big\{Q(\theta)-Q(\theta_0)\big\}
  >0.
\]

By Lemma~\ref{lem:as-Q-unif},
\[
  \sup_{\theta\in\Theta}\big|Q_n(\theta)-Q(\theta)\big| \;\pto\;0.
\]
Hence, for $\varepsilon=\delta_\eta/4$,
\[
  \Pp\!\left(
  \sup_{\theta\in\Theta}\big|Q_n(\theta)-Q(\theta)\big|
  > \frac{\delta_\eta}{4}
  \right)
  \to 0.
\]
For all large $n$, with probability at least $1-o(1)$, the event
\[
  \Acal_{n,\eta}
  := \left\{
  \sup_{\theta\in\Theta}\big|Q_n(\theta)-Q(\theta)\big|
  \le \frac{\delta_\eta}{4}
  \right\}
\]
occurs.
On $\Acal_{n,\eta}$ we have, for every $\theta\in S_\eta$,
\[
  Q_n(\theta)
  \ge
  Q(\theta) - \frac{\delta_\eta}{4}
  \ge
  Q(\theta_0) + \delta_\eta - \frac{\delta_\eta}{4}
  = Q(\theta_0) + \frac{3\delta_\eta}{4},
\]
while
\[
  Q_n(\theta_0)
  \le
  Q(\theta_0) + \frac{\delta_\eta}{4}.
\]
Therefore, on $\Acal_{n,\eta}$,
\[
  Q_n(\theta)
  > Q_n(\theta_0)
  \qquad\text{for all }\theta\in S_\eta,
\]
so no minimizer of $Q_n$ can lie in $S_\eta$.
In particular,
\[
  \Pp\big(\hat\theta_n\in S_\eta\big)
  \le
  \Pp(\Acal_{n,\eta}^c)
  \to 0.
\]
Since $\eta>0$ is arbitrary, this implies $\hat\theta_n\pto\theta_0$.
\end{proof}

\begin{proof}[Proof of Theorem~\ref{thm:as-consistency}]
Under Assumption~\ref{ass:as-lln}\,\textup{(AS--LLN1)}–\textup{(AS--LLN4)},
Lemma~\ref{lem:as-Q-unif} gives
\[
  \sup_{\theta\in\Theta}\big|Q_n(\theta)-Q(\theta)\big| \;\pto\; 0.
\]
Assumption~\ref{ass:as-lln}\,\textup{(AS--LLN3)} ensures that $Q$ is continuous
on $\Theta$ and uniquely minimized at $\theta_0$.
Applying Lemma~\ref{lem:as-argmin} then yields $\hat\theta_n\pto\theta_0$ for
any sequence of minimizers
$\hat\theta_n\in\arg\min_{\theta\in\Theta}Q_n(\theta)$.
This proves the consistency result stated in
Theorem~\ref{thm:as-consistency}.
\end{proof}

\subsection{Primitive LLN via affinity sets}
\label{subsec:primitive-LLN-proof}

This subsection verifies Assumption~\ref{ass:as-lln}\,\textup{(AS--LLN2)} from
primitive design-based conditions.  The only substantive residual remains the
Section~\ref{sec:design-orthogonality} design residual
\[
  R_{i,\theta}{\psi}(W)
  =\psi(W)-\E[\psi(W)\mid g_i(W;\theta)].
\]
For the proofs below, use the standard centered-moment notation
\[
  Z_{i,n}(\theta)
  :=\Psi_i(\theta)-\E\Psi_i(\theta),
  \qquad
  \bar Z_n(\theta)
  :=\frac{1}{N_n}\sum_{i=1}^{N_n}Z_{i,n}(\theta).
\]
Thus the desired ULLN is simply
\[
  \sup_{\theta}\|\bar Z_n(\theta)\|\pto0.
\]
All expectations and covariances in this subsection are with respect to the known design
$\Dcal_n$, conditional on the realized finite population.

\subsubsection*{B.2.1 Smooth exposure maps}

The smooth case covers gravity and market-access exposure maps.  We state projection
regularity directly because smoothness of $g_i(W;\theta)$ alone does not imply smoothness of
$\E[\psi(W)\mid g_i(W;\theta)]$ as $\theta$ varies.

\begin{assumption}[Smooth-kernel ULLN primitives]
\label{ass:as-lln-smooth}
Let $\Theta\subset\R^p$ be compact and let the moment dictionary
$\{(\phi_m,\psi_m):m=1,\ldots,q\}$ be fixed.  Write
\[
  \Psi_{i,m}(\theta)
  :=
  \phi_m(Y_i)\{\psi_m(W)-m_{i,m,\theta}(g_i(W;\theta))\},
  \qquad
  m_{i,m,\theta}(s):=
  \E[\psi_m(W)\mid g_i(W;\theta)=s].
\]
Assume the following conditions.
\begin{enumerate}[label=\textup{(S\arabic*)}, leftmargin=1.4cm]
\item \textbf{Bounded dictionary.}
There are constants $C_\phi,C_\psi<\infty$ such that
$|\phi_m(Y_i)|\le C_\phi$ and $|\psi_m(W)|\le C_\psi$ for all $i,m,n$.

\item \textbf{Smooth exposure and regular projection.}
There exist nonnegative envelopes $G_{i,n}(W)$ and $M_{i,m,n}$ such that, for all
$\theta,\theta'\in\Theta$,
\[
  \|g_i(W;\theta)-g_i(W;\theta')\|
  \le G_{i,n}(W)\|\theta-\theta'\|,
\]
and for all support points $s,s'$,
\[
  |m_{i,m,\theta}(s)-m_{i,m,\theta'}(s')|
  \le
  M_{i,m,n}\{\|\theta-\theta'\|+\|s-s'\|\}.
\]
The envelopes satisfy
\[
  \sup_n \frac{1}{N_n}\sum_{i=1}^{N_n}
  \max_{1\le m\le q}M_{i,m,n}\{1+\E G_{i,n}(W)\}<\infty.
\]

\item \textbf{Affinity covariance bound.}
There exist affinity sets $A_i=A_{i,n}$, with $i\in A_i$, such that
$b_n:=\max_i |A_i|=o(N_n)$.  Moreover, with
\[
  \rho_n
  :=
  \sup_{\theta\in\Theta}
  \frac{1}{N_n^2}
  \sum_{i=1}^{N_n}\sum_{j\notin A_i}
  \big\|\Cov\{Z_{i,n}(\theta),Z_{j,n}(\theta)\}\big\|_{\mathrm{op}},
\]
we have $\rho_n\to0$.
\end{enumerate}
\end{assumption}

\begin{lemma}[Smooth-kernel ULLN]
\label{lem:as-ulln-S}
Under Assumption~\ref{ass:as-lln-smooth},
\[
  \sup_{\theta\in\Theta}
  \left\|
  \frac{1}{N_n}\sum_{i=1}^{N_n}\{\Psi_i(\theta)-\E\Psi_i(\theta)\}
  \right\|\pto0.
\]
\end{lemma}

\begin{proof}
First, Assumption~\ref{ass:as-lln-smooth}\,\textup{(S1)} implies
$\sup_{i,n,\theta}\E\|\Psi_i(\theta)\|^2<\infty$ because
$|m_{i,m,\theta}(s)|\le C_\psi$.  Hence, for any $i,j,\theta$,
\[
  \big\|\Cov\{Z_{i,n}(\theta),Z_{j,n}(\theta)\}\big\|_{\mathrm{op}}
  \le C
\]
for a finite constant $C$.  Therefore
\begin{align*}
  \sup_{\theta\in\Theta}
  \E\|\bar Z_n(\theta)\|^2
  &\le
  C\left\{
  \frac{b_n}{N_n}+\rho_n
  \right\}
  =o(1).                                      \tag{B.1$'$}
\end{align*}
Indeed,
\[
  \E\|\bar Z_n(\theta)\|^2
  =\operatorname{tr}\{\Var(\bar Z_n(\theta))\}
  \le q\,\|\Var(\bar Z_n(\theta))\|_{\mathrm{op}},
\]
and
\[
  \|\Var(\bar Z_n(\theta))\|_{\mathrm{op}}
  \le
  \frac1{N_n^2}\sum_{i=1}^{N_n}\sum_{j=1}^{N_n}
  \big\|\Cov\{Z_{i,n}(\theta),Z_{j,n}(\theta)\}\big\|_{\mathrm{op}}.
\]
The within-affinity part has at most $N_n b_n$ terms and each term is bounded by
the uniform second-moment bound; the outside-affinity part is exactly controlled
by $\rho_n$.  This gives \textup{(B.1$'$)}.

Second, Assumption~\ref{ass:as-lln-smooth}\,\textup{(S2)} implies that the moment vector is
Lipschitz in $\theta$.  For each $m$,
\begin{align*}
&|m_{i,m,\theta}(g_i(W;\theta))-m_{i,m,\theta'}(g_i(W;\theta'))|  \\
&\qquad\le
M_{i,m,n}\{1+G_{i,n}(W)\}\|\theta-\theta'\|,
\end{align*}
so
\[
  \|\Psi_i(\theta)-\Psi_i(\theta')\|
  \le
  L_{i,n}\|\theta-\theta'\|,
  \qquad
  L_{i,n}:=
  \sqrt q C_\phi\max_m M_{i,m,n}\{1+G_{i,n}(W)\}.
\]
By the envelope condition,
\[
  \frac{1}{N_n}\sum_i \E L_{i,n}=O(1),
  \qquad
  \frac{1}{N_n}\sum_i L_{i,n}=O_p(1).        \tag{B.2$'$}
\]

Now let $\Theta_\varepsilon$ be a finite $\varepsilon$-net of the compact set $\Theta$.  For
fixed $\varepsilon$, (B.1$'$), Chebyshev's inequality, and a union bound imply
\[
  \max_{\theta_r\in\Theta_\varepsilon}\|\bar Z_n(\theta_r)\|\pto0.
\]
For arbitrary $\theta\in\Theta$, let $\pi_\varepsilon(\theta)$ be a nearest net point.  By
(B.2$'$),
\begin{align*}
  \|\bar Z_n(\theta)-\bar Z_n(\pi_\varepsilon(\theta))\|
  &\le
  \varepsilon\left\{
  \frac{1}{N_n}\sum_i L_{i,n}
  +
  \frac{1}{N_n}\sum_i \E L_{i,n}
  \right\}
  =\varepsilon O_p(1).
\end{align*}
Combining the net-point bound with the oscillation bound and then sending
$\varepsilon\downarrow0$ proves the result.
\end{proof}

\redrev{For bounded assignments $|W_j|\le1$ (in particular binary treatment, as in our applications), the exposure-derivative bounds below hold as displayed.}
For normalized gravity weights
$a_{ij}(\theta)\propto A_j\exp(-\theta d_{ij})$, the exposure derivative satisfies
\[
  \left|
  \frac{\partial}{\partial\theta}\sum_{j\ne i}a_{ij}(\theta)W_j
  \right|
  \le
  2\sum_{j\ne i}a_{ij}(\theta)d_{ij}.
\]
For market-access weights $a_{ij}(\theta)\propto A_j(1+\alpha d_{ij})^{-\theta}$,
\[
  \left|
  \frac{\partial}{\partial\theta}\sum_{j\ne i}a_{ij}(\theta)W_j
  \right|
  \le
  2\sum_{j\ne i}a_{ij}(\theta)\log(1+\alpha d_{ij}).
\]
Thus the exposure-smoothness part of Assumption~\ref{ass:as-lln-smooth}\,\textup{(S2)} follows
from bounded weighted first moments of distance or log-distance, while the regularity of the
design projection is the separate primitive condition stated above.

\subsubsection*{B.2.2 Ring exposures}
\label{app:continuous-hard-ring-ulln}

The smooth proof above does not apply to a ring exposure because the sample
path $\theta\mapsto\Psi_i(\theta)$ can be stepwise.  Nevertheless, a continuum ULLN can be
proved under a direct shell-regularity condition.  The condition says that small
changes in the cutoff affect only a vanishing average mass of pair distances and
that the design projection changes only through this shell.

For this subsection the parameter $\theta\in\Theta=[\underline\theta,\overline\theta]$ is the ring radius, and set
\[
  Z_{i,n}(\theta):=\Psi_i(\theta)-\E\Psi_i(\theta),
  \qquad
  \bar Z_n(\theta):=\frac1{N_n}\sum_{i=1}^{N_n}Z_{i,n}(\theta).
\]

\begin{assumption}[Ring ULLN primitives]
\label{ass:hard-ring-ulln}
The following conditions hold.
\begin{enumerate}[label=\textup{(HR\arabic*)}, leftmargin=1.55cm]
\item \textbf{Bounded moment contributions.}
There is a finite constant $C$ such that
\[
  \sup_n\sup_{i\le N_n}\sup_{\theta\in\Theta}\|\Psi_i(\theta)\|\le C
  \quad\text{almost surely}.
\]

\item \textbf{Pointwise affinity variance control.}
There exist affinity sets $A_i=A_{i,n}$ with $i\in A_i$ and
$b_n:=\max_i|A_i|=o(N_n)$ such that
\[
  \rho_n^{\mathrm{ring}}
  :=
  \sup_{\theta\in\Theta}
  \frac{1}{N_n^2}
  \sum_{i=1}^{N_n}\sum_{j\notin A_i}
  \big\|\Cov\{Z_{i,n}(\theta),Z_{j,n}(\theta)\}\big\|_{\mathrm{op}}
  \to0.
\]

\item \textbf{Uniform shell oscillation.}
For every interval $I\subset\Theta$, there are deterministic envelopes
$H_{i,n}(I)\ge0$ such that, whenever $\theta,\theta'\in I$,
\[
  \|\Psi_i(\theta)-\Psi_i(\theta')\|\le H_{i,n}(I)
  \quad\text{almost surely}.
\]
Define
\[
  \omega_n(\delta)
  :=
  \sup_{I\subset\Theta:\ |I|\le\delta}
  \frac1{N_n}\sum_{i=1}^{N_n}H_{i,n}(I).
\]
Then
\[
  \lim_{\delta\downarrow0}\limsup_{n\to\infty}\omega_n(\delta)=0.
\]
\end{enumerate}
\end{assumption}

\begin{lemma}[Ring ULLN]
\label{lem:continuous-hard-ring-ulln}
Under Assumption~\ref{ass:hard-ring-ulln},
\[
  \sup_{\theta\in\Theta}
  \left\|
    \bar\Psi_n(\theta)-\mu_n(\theta)
  \right\|
  =
  \sup_{\theta\in\Theta}\|\bar Z_n(\theta)\|
  \pto0.
\]
\end{lemma}

\begin{proof}
First fix $\theta\in\Theta$.  The same trace and covariance decomposition used in
Lemma~\ref{lem:as-ulln-S} gives
\[
  \E\|\bar Z_n(\theta)\|^2
  \le
  C\left(\frac{b_n}{N_n}+\rho_n^{\mathrm{ring}}\right),
\]
where the constant can change from line to line.  This bound holds uniformly in
$\theta$ by Assumption~\ref{ass:hard-ring-ulln}.

Fix $\delta>0$ and cover $\Theta$ by intervals $I_1,\ldots,I_{M_\delta}$ of
length at most $\delta$, with $M_\delta\le C/\delta$.  Choose one grid point
$r_\ell\in I_\ell$ for each interval.  For fixed $\delta$, Chebyshev's inequality
and a union bound imply
\[
  \Pr\left(
    \max_{\ell\le M_\delta}\|\bar Z_n(r_\ell)\|>\eta
  \right)
  \le
  \frac{C M_\delta}{\eta^2}
  \left(\frac{b_n}{N_n}+\rho_n^{\mathrm{ring}}\right)
  \to0.
\]
For any $\theta\in I_\ell$,
\begin{align*}
  \|\bar Z_n(\theta)-\bar Z_n(r_\ell)\|
  &\le
  \frac1{N_n}\sum_{i=1}^{N_n}\|\Psi_i(\theta)-\Psi_i(r_\ell)\|
  +
  \frac1{N_n}\sum_{i=1}^{N_n}\E\|\Psi_i(\theta)-\Psi_i(r_\ell)\|  \\
  &\le 2\omega_n(\delta).
\end{align*}
Therefore
\[
  \sup_{\theta\in\Theta}\|\bar Z_n(\theta)\|
  \le
  \max_{\ell\le M_\delta}\|\bar Z_n(r_\ell)\|+2\omega_n(\delta).
\]
Taking $n\to\infty$ for fixed $\delta$ and then sending $\delta\downarrow0$
proves the result.
\end{proof}

\begin{lemma}[Shell regularity implies continuity of the mean map]
\label{lem:hard-ring-mu-continuity}
Suppose Assumption~\ref{ass:hard-ring-ulln}\,\textup{(HR3)} holds and
Assumption~\ref{ass:as-lln}\,\textup{\ref{ass:as-lln1}} holds for the continuous
ring class $\Theta=[\underline\theta,\overline\theta]$.  Then the limit mean map
$\theta\mapsto\mu(\theta)$ is uniformly continuous on $\Theta$.  Hence the continuity part
of Assumption~\ref{ass:as-lln}\,\textup{\ref{ass:as-lln3}} follows from the shell
primitive rather than being imposed separately.
\end{lemma}

\begin{proof}
For $\theta,\theta'\in I$ with $|I|\le\delta$, Assumption~\ref{ass:hard-ring-ulln}\,
\textup{(HR3)} gives
\[
  \|\mu_n(\theta)-\mu_n(\theta')\|
  \le
  \frac1{N_n}\sum_{i=1}^{N_n}\E\|\Psi_i(\theta)-\Psi_i(\theta')\|
  \le \omega_n(\delta).
\]
Therefore, using the uniform convergence in
Assumption~\ref{ass:as-lln}\,\textup{\ref{ass:as-lln1}},
\[
  \sup_{|\theta-\theta'|\le\delta}\|\mu(\theta)-\mu(\theta')\|
  \le
  2\sup_{\theta\in\Theta}\|\mu_n(\theta)-\mu(\theta)\|+\omega_n(\delta).
\]
Taking first $\limsup_{n\to\infty}$ and then $\delta\downarrow0$ yields uniform
continuity of $\mu$.
\end{proof}

A useful primitive sufficient condition for the envelope condition \textup{(HR3)} of Assumption~\ref{ass:hard-ring-ulln}
comes from weighted shell mass.  Suppose
\[
  g_i(W;\theta)=\sum_{j\ne i}a_{ij,n}W_j\1\{d_{ij}\le \theta\},
  \qquad
  \sup_{i,n}\sum_{j\ne i}|a_{ij,n}|<\infty,
\]
and define the weighted shell mass
\[
  S_{i,n}(I):=\sum_{j\ne i}|a_{ij,n}|\1\{d_{ij}\in I\}.
\]
If the design projection satisfies the shell bound
\[
  \left|
    \E[\psi_m(W)\mid g_i(W;\theta)]
    -
    \E[\psi_m(W)\mid g_i(W;\theta')]
  \right|
  \le L S_{i,n}(I)
  \quad\text{whenever } \theta,\theta'\in I,
\]
and $\phi_m(Y_i)$ is uniformly bounded, then one may take
$H_{i,n}(I)=C S_{i,n}(I)$.  The shell condition becomes
\[
  \lim_{\delta\downarrow0}\limsup_{n\to\infty}
  \sup_{I\subset\Theta:\ |I|\le\delta}
  \frac1{N_n}\sum_{i=1}^{N_n}\sum_{j\ne i}
  |a_{ij,n}|\1\{d_{ij}\in I\}=0.
\]
\redrev{This is an \emph{atomless-shell} (no-mass-point) condition for continuous
cutoffs: no single radius may carry a non-negligible share of the weighted
pairwise mass. It is the spatial analogue of the familiar requirement, in
threshold and regression-discontinuity models, that the running variable have no
mass point at the cutoff, and the modulus $\omega_n(\delta)$ above is the
asymptotic-equicontinuity device that turns it into a uniform law of large
numbers.}  Equivalently,
for the weighted shell measures
\[
  \nu_n(I):=\frac1{N_n}\sum_{i=1}^{N_n}\sum_{j\ne i}
  |a_{ij,n}|\1\{d_{ij}\in I\},
\]
\redrev{which record, on the radius axis, how much exposure weight sits at each
distance (a weighted analogue of the pair-distance distribution, or Ripley's $K$
function in spatial statistics),} small intervals must have uniformly small mass.  A sufficient way to obtain this
is that $\nu_n$ converge uniformly over intervals in the relevant cutoff region
to a limiting measure with no atoms and locally bounded interval mass.  The
condition can fail in regular lattice designs when many pairs lie exactly at the
candidate cutoff radius; in that case the continuous-cutoff target may be ill
posed at that radius.  This observation is separate from finite VC dimension:
VC controls the number of threshold sets, while the shell condition controls the
size of the jumps under the design.

\subsubsection*{B.2.3 Finite-grid and $K$-hop maps}
\label{app:finite-grid-ulln-proof}

For genuinely discrete maps, the main primitive result is finite-grid selection.
This is the natural formulation for network cutoffs $K\in\{1,\ldots,K_{\max,n}\}$
and for any exposure problem that the researcher intentionally defines as a
finite list of candidate specifications.

\begin{assumption}[Finite-grid ULLN primitives for non-smooth maps]
\label{ass:as-lln-vc}
Let $\Theta_n$ be a finite candidate set with cardinality $K_n=|\Theta_n|$.  Assume
\[
  \sup_n\sup_{i\le N_n}\sup_{\theta\in\Theta_n}
  \E\|\Psi_i(\theta)\|^2<\infty.
\]
There exist affinity sets $A_i=A_{i,n}$ with $i\in A_i$ and
$b_n:=\max_i|A_i|$ such that, with
\[
  \rho_n^{\mathrm{fg}}
  :=
  \sup_{\theta\in\Theta_n}
  \frac{1}{N_n^2}
  \sum_{i=1}^{N_n}\sum_{j\notin A_i}
  \big\|\Cov\{Z_{i,n}(\theta),Z_{j,n}(\theta)\}\big\|_{\mathrm{op}},
\]
we have
\[
  K_n\left(\frac{b_n}{N_n}+\rho_n^{\mathrm{fg}}\right)\to0.
\]
\end{assumption}

\begin{lemma}[Finite-grid ULLN]
\label{lem:finite-grid-ulln}
Under Assumption~\ref{ass:as-lln-vc},
\[
  \sup_{\theta\in\Theta_n}
  \left\|
  \frac{1}{N_n}\sum_{i=1}^{N_n}\{\Psi_i(\theta)-\E\Psi_i(\theta)\}
  \right\|\pto0.
\]
\end{lemma}

\begin{proof}
For each fixed $\theta\in\Theta_n$, the same trace and covariance decomposition
as in the smooth proof gives
\[
  \E\|\bar Z_n(\theta)\|^2
  \le
  C\left(\frac{b_n}{N_n}+\rho_n^{\mathrm{fg}}\right).
\]
Therefore, by Chebyshev's inequality and a union bound,
\[
  \Pp\left(\sup_{\theta\in\Theta_n}\|\bar Z_n(\theta)\|>\eta\right)
  \le
  \frac{C K_n}{\eta^2}
  \left(\frac{b_n}{N_n}+\rho_n^{\mathrm{fg}}\right)\to0.
\]
\end{proof}

\begin{lemma}[Finite-grid argmin consistency]
\label{lem:finite-grid-argmin-consistency}
Let $\Theta_n$ be finite and define
\[
  Q_n(\theta):=\bar\Psi_n(\theta)^\top\Lambda_n\bar\Psi_n(\theta),
  \qquad
  Q_n^0(\theta):=\mu_n(\theta)^\top\Lambda\mu_n(\theta).
\]
Suppose the finite-grid ULLN in Lemma~\ref{lem:finite-grid-ulln} holds,
$\Lambda_n\pto\Lambda$, and
$\sup_{\theta\in\Theta_n}\|\mu_n(\theta)\|=O(1)$.  Suppose also that there is a
unique grid target $\theta_{0,n}\in\Theta_n$ and a constant $c>0$ such that
\[
  \min_{\theta\in\Theta_n:\ \theta\ne\theta_{0,n}}
  \{Q_n^0(\theta)-Q_n^0(\theta_{0,n})\}
  \ge c
\]
for all sufficiently large $n$.  Then any finite-grid minimizer
\[
  \hat\theta_n\in\arg\min_{\theta\in\Theta_n}Q_n(\theta)
\]
satisfies
\[
  \Pr(\hat\theta_n=\theta_{0,n})\to1.
\]
\end{lemma}

\begin{proof}
The finite-grid ULLN implies
\[
  \sup_{\theta\in\Theta_n}\|\bar\Psi_n(\theta)-\mu_n(\theta)\|\pto0.
\]
Together with $\Lambda_n\pto\Lambda$ and
$\sup_{\theta\in\Theta_n}\|\mu_n(\theta)\|=O(1)$, this gives
\[
  \sup_{\theta\in\Theta_n}|Q_n(\theta)-Q_n^0(\theta)|\pto0.
\]
With probability tending to one, this supremum is smaller than $c/3$.  On that
event, for every $\theta\ne\theta_{0,n}$,
\[
  Q_n(\theta)-Q_n(\theta_{0,n})
  \ge
  \{Q_n^0(\theta)-Q_n^0(\theta_{0,n})\}-2c/3
  \ge c/3>0.
\]
Thus the unique minimizer of $Q_n$ over $\Theta_n$ is $\theta_{0,n}$ with
probability tending to one.
\end{proof}

\subsection{Asymptotic normality of the baseline regular case}
\label{app:DB-AN-affinity}

This subsection proves the regular asymptotic-normality result used in
Section~\ref{sec:DB-consistency-AN-affinity}.  The dependent triangular-array
central limit theorem is needed only at the true exposure parameter, and the
local behavior of the criterion is handled by the Z-estimator theorem of
\citet[Theorem~3.3.1]{vaart1996weak}.  For ring exposures the sample moments need
not be differentiable: the requirement is the projected estimating equation in
Assumption~\ref{ass:AN-AFF}\,\textup{\ref{ass:an-aff6}}, a condition on the
estimator rather than on the sample paths.  When the radius is selected on a
numerical grid, the grid mesh enters the same first-order-negligibility
requirement.

We use the notation of Section~\ref{sec:DB-consistency-AN-affinity}.  Define
\[
  \bar\Psi_n(\theta)
  :=
  \frac{1}{N_n}\sum_{i=1}^{N_n}\Psi_i(\theta),
  \qquad
  \mu_n(\theta)
  :=
  \frac{1}{N_n}\sum_{i=1}^{N_n}\E[\Psi_i(\theta)],
\]
and, at the true parameter,
\[
  Z_{i,n}:=\Psi_i(\theta_0)-\E[\Psi_i(\theta_0)],
  \qquad
  \bar Z_n:=\frac{1}{N_n}\sum_{i=1}^{N_n}Z_{i,n}.
\]
The within-affinity covariance matrix is
\[
  \Omega_n
  :=
  \sum_{i=1}^{N_n}\sum_{j\in A_i}
  \Cov\{Z_{i,n},Z_{j,n}\}.
\]
Assumption~\ref{ass:AN-AFF}\,\textup{\ref{ass:an-aff3}} states that
$\Omega_n/N_n\to\Omega$, where $\Omega$ is finite and positive definite.

\medskip
\noindent\bluerev{The asymptotic-normality conditions summarized in Section~\ref{sec:DB-consistency-AN-affinity} are the following.}

\begin{assumption}[Design-based AN conditions]
\label{ass:AN-AFF}
In addition to Assumption~\ref{ass:as-lln}, suppose:
\begin{enumerate}[label=\textup{AN--AFF\arabic*}, ref=\textup{AN--AFF\arabic*}, leftmargin=1.6cm]

\item \label{ass:an-aff1}\textbf{Interiority.}
$\theta_0$ is an interior point of $\Theta$.

\item \label{ass:an-aff2}\textbf{Mean differentiability and full rank.}
There exists a neighborhood $\Ncal(\theta_0)$ and a continuous map
$G:\Ncal(\theta_0)\to\R^{q\times p}$ such that
\[
  \mu(\theta)
  =
  \mu(\theta_0)
  + G(\theta_0)(\theta-\theta_0)
  + r(\theta),
  \qquad
  \|r(\theta)\| \le C\|\theta-\theta_0\|^2
  \quad\text{for }\theta\in\Ncal(\theta_0),
\]
for some finite $C$, and $G^\top\Lambda G$ is nonsingular at
$\theta_0$, where $G:=G(\theta_0)$.

\item \label{ass:an-aff3}\textbf{Affinity-set CLT and covariance stabilization.}
The array $\{Z_{i,n}\}$ has uniformly bounded fourth moments.
Let $Z_{i,n,\ell}$ denote the $\ell$th component of $Z_{i,n}$ for
$\ell=1,\ldots,q$.
For each scalar index $(i,\ell)$, define a scalar affinity set
\[
  A_n(i,\ell)
  \subseteq \Ucal_n \times \{1,\ldots,q\}
  \quad\text{by}\quad
  A_n(i,\ell)
  :=
  \big\{(j,\ell') : j\in A_i,\ \ell'=1,\ldots,q\big\},
\]
and define the ``outside-affinity'' scalar sum
\[
  Z_{-i,\ell}
  :=
  \sum_{(j,\ell') \notin A_n(i,\ell)} Z_{j,n,\ell'}.
\]
The scalar array $\{Z_{i,n,\ell}\}$ then satisfies the three
affinity-set covariance conditions of \citet{chandrasekhar2023general}:
\begin{align*}
  &\sum_{i=1}^{N_n}\sum_{\ell=1}^{q}
    \ \sum_{(j,\ell'),(k,\ell'')\in A_n(i,\ell)}
    \E\!\big[\ |Z_{i,n,\ell}|\ Z_{j,n,\ell'} Z_{k,n,\ell''}\ \big]
    \ =\ o\!\big(\|\Omega_n\|_F^{3/2}\big), \tag{A1} \\
  &\sum_{i=1}^{N_n}\sum_{\ell=1}^{q}
    \ \sum_{j=1}^{N_n}\sum_{\ell'=1}^{q}
    \ \sum_{(k,\ell'')\in A_n(i,\ell)} \ \sum_{(l,\tilde\ell)\in A_n(j,\ell')}
    \Cov\!\big(Z_{i,n,\ell}Z_{k,n,\ell''},\ Z_{j,n,\ell'}Z_{l,n,\tilde\ell}\big)
    \ =\ o\!\big(\|\Omega_n\|_F^{2}\big), \tag{A2} \\
  &\sum_{i=1}^{N_n}\sum_{\ell=1}^{q}
    \E\!\Big[\ \big|\E\big[\,Z_{i,n,\ell}\,Z_{-i,\ell}\ \big|\ Z_{-i,\ell}\big]\big|\ \Big]
    \ =\ o\!\big(\|\Omega_n\|_F\big). \tag{A3}
\end{align*}
Moreover, the within-affinity covariance has uniformly bounded row sums,
\[
  \sup_n\max_{1\le i\le N_n}\max_{1\le \ell\le q}
  \sum_{(j,\ell')\in A_n(i,\ell)}
  \big|\Cov(Z_{i,n,\ell},Z_{j,n,\ell'})\big|<\infty,
\]
and the aggregate covariance stabilizes:
\[
  \frac{\Omega_n}{N_n}\ \to\ \Omega,
  \qquad
  \Omega\ \text{finite and positive definite}.
\]
Under these conditions,
\[
  \sqrt{N_n}\,\bar Z_n \ \Rightarrow\ \Ncal(0,\Omega).
\]
In addition, the finite-population centering at the true parameter is negligible
at the $\sqrt{N_n}$ scale:
\[
  \sqrt{N_n}\,
  \big\|\mu_n(\theta_0)-\mu(\theta_0)\big\|
  \ \to\ 0.
\]
Consequently,
\[
  \sqrt{N_n}\,
  \big(\bar\Psi_n(\theta_0)-\mu(\theta_0)\big)
  \ \Rightarrow\ \Ncal(0,\Omega).
\]
In the exposure-mapping application, this condition holds exactly because the
residualized design-based orthogonality moments satisfy
$\mu_n(\theta_0)=\mu(\theta_0)=0$ for every $n$.

\item \label{ass:an-aff4}\textbf{Empirical process equicontinuity.}
Writing
\[
  \mathbb G_n(\theta)
  :=
  \sqrt{N_n}\big(\bar\Psi_n(\theta)-\mu(\theta)\big),
\]
the process $\{\mathbb G_n(\theta):\theta\in\Ncal(\theta_0)\}$ is stochastically
equicontinuous at $\theta_0$: for any sequence $h_n\downarrow0$,
\[
  \sup_{\|\theta-\theta_0\|\le h_n}
  \big\|\mathbb G_n(\theta)-\mathbb G_n(\theta_0)\big\|
  \ \pto\ 0.
\]
\item \label{ass:an-aff6}\textbf{Projected estimating equation.}
The estimator satisfies the projected GMM estimating equation up to
$o_p(N_n^{-1/2})$:
\[
  \big\|
    G^\top \Lambda_n \bar\Psi_n(\hat\theta_n)
  \big\|
  =
  o_p(N_n^{-1/2}).
\]
\end{enumerate}
\end{assumption}

\subsubsection{Pointwise affinity-set CLT}

\begin{lemma}[Affinity-set CLT for the moment vector]
\label{lem:affCLT}
Under Assumptions~\ref{ass:as-lln} and
\ref{ass:AN-AFF}\,\textup{\ref{ass:an-aff3}},
\[
  \Omega_n^{-1/2}\sum_{i=1}^{N_n} Z_{i,n}
  \Rightarrow
  \Ncal(0,I_q).
\]
Consequently,
\[
  \sqrt{N_n}\{\bar\Psi_n(\theta_0)-\mu_n(\theta_0)\}
  \Rightarrow
  \Ncal(0,\Omega).
\]
Since Assumption~\ref{ass:AN-AFF}\,\textup{\ref{ass:an-aff3}} also contains
the fixed-limit centering condition,
\[
  \sqrt{N_n}\{\bar\Psi_n(\theta_0)-\mu(\theta_0)\}
  \Rightarrow
  \Ncal(0,\Omega).
\]
\end{lemma}

\begin{proof}
The first display is the multivariate affinity-set CLT of
\citet{chandrasekhar2023general}.  Assumption~\ref{ass:AN-AFF}\,\textup{\ref{ass:an-aff3}}
verifies their three covariance conditions (their Assumptions~1--3) for the
\emph{stacked} scalar array $\{Z_{i,n,\ell}\}$ with the dimension-inclusive
affinity sets $A_n(i,\ell)$ defined there, including the cross-component terms;
their multivariate Theorem~1 then yields
$\Omega_n^{-1/2}\sum_{i=1}^{N_n}Z_{i,n}\Rightarrow\Ncal(0,I_q)$ directly.  Since
\[
  \sum_{i=1}^{N_n}Z_{i,n}
  =
  N_n\{\bar\Psi_n(\theta_0)-\mu_n(\theta_0)\},
\]
we can write
\[
  \sqrt{N_n}\{\bar\Psi_n(\theta_0)-\mu_n(\theta_0)\}
  =
  \Big(\frac{\Omega_n}{N_n}\Big)^{1/2}
  \Big(\Omega_n^{-1/2}\sum_{i=1}^{N_n}Z_{i,n}\Big).
\]
The first factor converges to $\Omega^{1/2}$ and the second factor converges to
$\Ncal(0,I_q)$, so Slutsky's theorem gives the stated $\Ncal(0,\Omega)$ limit.
The last display follows from the additional condition
$\sqrt{N_n}\|\mu_n(\theta_0)-\mu(\theta_0)\|\to0$.
\end{proof}

\subsubsection{The Z-estimator theorem used below}

\begin{lemma}[Finite-dimensional specialization of \citet{vaart1996weak}, Theorem~3.3.1]
\label{lem:vdvw-z-estimator}
Let $\Theta\subset\R^p$, let $M_n:\Theta\to\R^p$ be random maps, and let
$M:\Theta\to\R^p$ be deterministic.  Suppose that $\hat\theta_n\pto\theta_0$,
$M(\theta_0)=0$, and
\[
  \|M_n(\hat\theta_n)\|=o_p(N_n^{-1/2}).
\]
Suppose further that $M$ is differentiable at $\theta_0$ with nonsingular
derivative $A$, that
\[
  \sqrt{N_n}\{M_n(\theta_0)-M(\theta_0)\}\Rightarrow Z
\]
for a tight random vector $Z$, and that for every sequence
$\theta_n\pto\theta_0$,
\[
  \sqrt{N_n}
  \big[
    \{M_n(\theta_n)-M(\theta_n)\}
    -
    \{M_n(\theta_0)-M(\theta_0)\}
  \big]
  =
  o_p\big(1+\sqrt{N_n}\|\theta_n-\theta_0\|\big).
  \tag{Z}
\]
Then
\[
  \sqrt{N_n}(\hat\theta_n-\theta_0)
  =
  -A^{-1}\sqrt{N_n}\{M_n(\theta_0)-M(\theta_0)\}+o_p(1),
\]
and hence
\[
  \sqrt{N_n}(\hat\theta_n-\theta_0)\Rightarrow -A^{-1}Z.
\]
\end{lemma}

\begin{proof}
This is the finite-dimensional version of \citet[Theorem~3.3.1]{vaart1996weak}
with rate $\sqrt{N_n}$.  In their notation, $M_n$ and $M$ correspond to the
random and limiting maps $\Psi_n$ and $\Psi$, condition \textup{(Z)} is their
local stochastic expansion condition, the approximate-zero condition is
$M_n(\hat\theta_n)=o_p(N_n^{-1/2})$, and the nonsingular matrix $A$ is the
continuously invertible derivative of the limiting map.  Since the parameter is
finite-dimensional, continuous invertibility of the derivative is equivalent to
nonsingularity of $A$.
\end{proof}

\subsubsection{Proof of Theorem~\ref{thm:DB-AN-aff}}

\begin{proof}[Proof of Theorem~\ref{thm:DB-AN-aff}]
By Theorem~\ref{thm:as-consistency}, Assumption~\ref{ass:as-lln} implies
$\hat\theta_n\pto\theta_0$.  We apply Lemma~\ref{lem:vdvw-z-estimator} to the
projected GMM maps
\[
  M_n(\theta):=G^\top\Lambda_n\bar\Psi_n(\theta),
  \qquad
  M(\theta):=G^\top\Lambda\mu(\theta),
\]
where $G=G(\theta_0)$.

First, $M(\theta_0)=0$ because $\mu(\theta_0)=0$.  By
Assumption~\ref{ass:AN-AFF}\,\textup{\ref{ass:an-aff2}},
\[
  \mu(\theta)=\mu(\theta_0)+G(\theta-\theta_0)+r(\theta),
  \qquad
  \|r(\theta)\|\le C\|\theta-\theta_0\|^2,
\]
so $M$ is differentiable at $\theta_0$ with derivative
\[
  A:=G^\top\Lambda G.
\]
The matrix $A$ is nonsingular by
Assumption~\ref{ass:AN-AFF}\,\textup{\ref{ass:an-aff2}}.

Second, Lemma~\ref{lem:affCLT} and $\Lambda_n\pto\Lambda$ imply
\[
  \sqrt{N_n}\{M_n(\theta_0)-M(\theta_0)\}
  =
  G^\top\Lambda_n
  \sqrt{N_n}\{\bar\Psi_n(\theta_0)-\mu(\theta_0)\}
  \Rightarrow
  G^\top\Lambda Z,
\]
where $Z\sim\Ncal(0,\Omega)$.

Third, we verify the stochastic expansion condition \textup{(Z)}.  Let
$\theta_n\pto\theta_0$.  The random-sequence version of
Assumption~\ref{ass:AN-AFF}\,\textup{\ref{ass:an-aff4}} follows from its
deterministic shrinking-ball formulation: choose a deterministic sequence
$h_n\downarrow0$ such that
$\Pr(\|\theta_n-\theta_0\|>h_n)\to0$ and apply the supremum bound on the event
$\{\|\theta_n-\theta_0\|\le h_n\}$.  Therefore
\[
  \sqrt{N_n}
  \big[
    \{\bar\Psi_n(\theta_n)-\mu(\theta_n)\}
    -
    \{\bar\Psi_n(\theta_0)-\mu(\theta_0)\}
  \big]
  =o_p(1).
\]
Multiplying by $G^\top\Lambda_n$ preserves the $o_p(1)$ order.  The difference
between $\Lambda_n$ and $\Lambda$ contributes
\[
  G^\top(\Lambda_n-
  \Lambda)\sqrt{N_n}\{\mu(\theta_n)-\mu(\theta_0)\}
  =
  o_p(1)\,O\big(\sqrt{N_n}\|\theta_n-\theta_0\|\big),
\]
where the $O(\|\theta_n-\theta_0\|)$ bound follows from local mean
differentiability.  Hence the total remainder is
\[
  o_p\big(1+\sqrt{N_n}\|\theta_n-\theta_0\|\big),
\]
which is exactly condition \textup{(Z)}.

Fourth, Assumption~\ref{ass:AN-AFF}\,\textup{\ref{ass:an-aff6}} gives
\[
  \|M_n(\hat\theta_n)\|
  =
  \|G^\top\Lambda_n\bar\Psi_n(\hat\theta_n)\|
  =o_p(N_n^{-1/2}).
\]
This is a projected estimating-equation condition.  It is not a differentiability
condition and is therefore compatible with ring exposures
whose sample moments are step functions of the cutoff.  In smooth cases it can
be verified by an ordinary sample first-order condition and a Jacobian law of
large numbers.  In non-smooth continuous-cutoff cases it is maintained directly as part of the
regular Z/GMM condition.  If the continuous cutoff is computed by grid search,
the grid approximation must be first-order negligible, for
example a grid mesh $\Delta_n=o(N_n^{-1/2})$, so that the computed estimator
satisfies the projected estimating equation of
Assumption~\ref{ass:AN-AFF}\,\textup{\ref{ass:an-aff6}}.

All conditions of Lemma~\ref{lem:vdvw-z-estimator} are now verified.  Therefore
\[
  \sqrt{N_n}(\hat\theta_n-\theta_0)
  =
  -
  (G^\top\Lambda G)^{-1}G^\top\Lambda
  \sqrt{N_n}\{\bar\Psi_n(\theta_0)-\mu(\theta_0)\}
  +o_p(1).
\]
Combining this expansion with Lemma~\ref{lem:affCLT} gives
\[
  \sqrt{N_n}(\hat\theta_n-\theta_0)
  \Rightarrow
  \Ncal\!\left(
    0,
    (G^\top\Lambda G)^{-1}
    G^\top\Lambda\Omega\Lambda G
    (G^\top\Lambda G)^{-1}
  \right).
\]
If $\Lambda=\Omega^{-1}$, the variance reduces to
$(G^\top\Omega^{-1}G)^{-1}$.  If
$\Lambda_n=\hat\Omega_n^{-1}$ and $\hat\Omega_n\pto\Omega$, the efficient-weight
case follows by Slutsky's theorem.
\end{proof}

\begin{remark}[Continuous and finite-grid rings]
\label{rem:continuous-vs-finite-ring-AN}
A ring exposure can be regular even though the sample path
$\theta\mapsto\bar\Psi_n(\theta)$ is stepwise.  The regularity needed for Wald inference
is exactly the regular Z/GMM condition in Assumption~\ref{ass:AN-AFF}:
differentiability of the mean map $\theta\mapsto\mu(\theta)$ at $\theta_0$, a pointwise
affinity-set CLT, local stochastic equicontinuity, and the projected estimating
equation.  The continuum ULLN in Lemma~\ref{lem:continuous-hard-ring-ulln}
provides one primitive shell-regularity route for the corresponding uniform
convergence.  Thus the theorem does not require differentiability of the sample
moments, but it does require the shell and empirical-process regularity stated
in the assumptions.

A genuinely finite-grid radius or finite $K$-hop parameter is different.  Under a
fixed population separation gap, the selected discrete parameter is consistently
selected and has a degenerate first-order limit.  Wald inference for the
discrete cutoff is therefore not the target.  If a fine grid is used only as a
numerical implementation of an underlying continuous-cutoff estimator, the
continuous Wald limit applies provided the grid mesh $\Delta_n$ satisfies
$\Delta_n=o(N_n^{-1/2})$ and the projected estimating-equation remainder remains
$o_p(N_n^{-1/2})$.
\end{remark}

\subsection{Finite-grid exposure-map variants}
\label{app:finite-grid-nested-variants}

This subsection collects the finite-grid exception to the baseline regular
continuous theory.  It treats a finite discrete exposure-map component with an
optional continuous component.

\begin{corollary}[Finite discrete exposure-map components]
\label{cor:discrete-parameters}
Suppose the Stage~1 exposure-map parameter can be written as
$\theta=(\vartheta,\rho)$, where $\vartheta$ takes values in a finite set
$\Vcal$ and, conditional on $\vartheta$, $\rho$ lies in a compact set
$\mathcal B(\vartheta)\subset\R^{p_\rho}$.
For each $\vartheta\in\Vcal$, define the profile population criterion
\[
  Q^p(\vartheta)
  :=
  \inf_{\rho\in\mathcal B(\vartheta)} Q(\vartheta,\rho),
  \qquad
  Q(\vartheta,\rho)
  :=
  \mu(\vartheta,\rho)^\top\Lambda\mu(\vartheta,\rho).
\]
Assume that $Q^p$ has a unique minimizer $\vartheta_0$ with a strict gap:
\[
  Q^p(\vartheta)-Q^p(\vartheta_0)
  \ge c>0
  \qquad
  \text{for all }\vartheta\ne\vartheta_0,
\]
and let
\[
  \rho_0
  \in
  \arg\min_{\rho\in\mathcal B(\vartheta_0)} Q(\vartheta_0,\rho).
\]
For each $\vartheta$, let
\[
  \hat\rho_n(\vartheta)
  \in
  \arg\min_{\rho\in\mathcal B(\vartheta)} Q_n(\vartheta,\rho),
\]
and define the profiled finite-grid estimator
\[
  \hat\vartheta_n
  \in
  \arg\min_{\vartheta\in\Vcal}
  Q_n\big(\vartheta,\hat\rho_n(\vartheta)\big),
  \qquad
  \hat\rho_n:=\hat\rho_n(\hat\vartheta_n).
\]
Assume that the uniform LLN and deterministic-stabilization conditions hold uniformly over
$\{(\vartheta,\rho):\vartheta\in\Vcal,\rho\in\mathcal B(\vartheta)\}$.  Hence the
finite-population profile criteria
\[
  Q_n^{0,p}(\vartheta)
  :=\inf_{\rho\in\mathcal B(\vartheta)}
  \mu_n(\vartheta,\rho)^\top\Lambda\mu_n(\vartheta,\rho)
\]
satisfy $\sup_{\vartheta\in\Vcal}|Q_n^{0,p}(\vartheta)-Q^p(\vartheta)|\to0$,
so the fixed limit-profile gap transfers to the finite-population profile for
all sufficiently large $n$.  Assume also that the asymptotic-normality conditions
of Theorem~\ref{thm:DB-AN-aff} hold for the continuous component $\rho$
conditional on $\vartheta=\vartheta_0$.
Then
\[
  \Pr(\hat\vartheta_n=\vartheta_0)\to1.
\]
Consequently,
\[
  \sqrt{N_n}\,(\hat\rho_n-\rho_0)
  =
  \sqrt{N_n}\big(\hat\rho_n(\vartheta_0)-\rho_0\big)+o_p(1),
\]
so the continuous Stage~1 exposure-map parameter has the same first-order
asymptotic distribution as the oracle estimator that fixes the correct discrete
component $\vartheta_0$.
\end{corollary}

\begin{remark}[What the finite-grid results do and do not say]
\label{rem:finite-selection-interpretation}
Corollary~\ref{cor:discrete-parameters} is a Stage~1 result about exposure-map
parameters.  The continuous component $\rho$ may depend on the selected discrete
component $\hat\vartheta_n$ in finite samples.  The conclusion is not that this
dependence is absent, but that under a fixed strict profile gap the selected
class equals the population class with probability approaching one.  Hence the
selected estimator equals the oracle fixed-class estimator with probability
approaching one, which removes the discrete selection step from the first-order
asymptotic variance.

This fixed-gap condition rules out local ties among discrete classes.  If two
candidate exposure maps have population criteria that are equal, or differ only
at the $N_n^{-1/2}$ scale, the selected criterion can have a nonstandard
post-selection distribution, analogous to the maximum of two noisy sample means.
Such local-tie cases are not covered by the finite-grid corollary.
\end{remark}

\begin{remark}[Discrete rings and network radii]
\label{rem:discrete-rings-networks}
Corollary~\ref{cor:discrete-parameters} applies when a finite-grid component of
the exposure map is point-identified by a strict profile gap.  Examples include
non-nested finite grids of scalar ring radii or $K$-hop network radii when the
profile population criterion has a unique minimizer.

The ring case covered by Assumption~\ref{ass:AN-AFF} is a regular
continuous-parameter problem.  In the scalar ring case, $g(W;X_i,R)$ is an
average over units within radius $R$, and increasing $R$ does not usually refine
the smaller-radius exposure in the sense of sigma-fields; the regular
continuous-parameter theory applies provided the population moment map is mean
differentiable.
\end{remark}

\begin{remark}[Stage~1 versus Stage~2]
\label{rem:stage1-stage2-discrete}
The notation $\theta=(\vartheta,\rho)$ in this section refers only to the
Stage~1 exposure-map parameter.  Downstream causal or exposure--response
parameters, such as regression coefficients or average treatment effects, are
introduced in Section~\ref{sec:stage2}.  If the Stage~1 discrete component is
selected consistently, it contributes no first-order term to Stage~2 inference.
Any regular continuous Stage~1 component $\rho$, however, must be propagated into
Stage~2 by the stacked estimating-equation or delta-method argument described in
Section~\ref{sec:stage2} and Appendix~\ref{app:stage2}.
\end{remark}

\subsubsection{\texorpdfstring{Proof of Corollary~\ref{cor:discrete-parameters}}{Proof of the discrete-parameters corollary}}
\label{app:finite-discrete-selection}

Let $\theta=(\vartheta,\rho)$ with $\vartheta\in\Vcal$ finite and
$\rho\in\mathcal B(\vartheta)$ compact.  Define
\[
  Q^p(\vartheta)
  :=
  \inf_{\rho\in\mathcal B(\vartheta)} Q(\vartheta,\rho),
  \qquad
  Q_n^p(\vartheta)
  :=
  \inf_{\rho\in\mathcal B(\vartheta)} Q_n(\vartheta,\rho).
\]
The sample estimator in Corollary~\ref{cor:discrete-parameters} is equivalently
\[
  \hat\vartheta_n\in\arg\min_{\vartheta\in\Vcal} Q_n^p(\vartheta),
  \qquad
  \hat\rho_n=\hat\rho_n(\hat\vartheta_n).
\]

\begin{lemma}[Selection consistency for finite profiled criteria]
\label{lem:finite-profile-selection}
Suppose $\Vcal$ is finite and
\[
  \sup_{\vartheta\in\Vcal}\sup_{\rho\in\mathcal B(\vartheta)}
  |Q_n(\vartheta,\rho)-Q(\vartheta,\rho)|\pto0.
\]
If $Q^p$ has a unique minimizer $\vartheta_0$ and satisfies the strict profile
gap
\[
  Q^p(\vartheta)-Q^p(\vartheta_0)\ge c>0
  \qquad\text{for all }\vartheta\ne\vartheta_0,
\]
then $\Pr(\hat\vartheta_n=\vartheta_0)\to1$.
\end{lemma}

\begin{proof}
The uniform convergence of $Q_n$ implies uniform convergence of the profile
criteria:
\[
  \sup_{\vartheta\in\Vcal}|Q_n^p(\vartheta)-Q^p(\vartheta)|\pto0,
\]
because taking infima over $\rho$ is one-Lipschitz with respect to the sup norm.
For any $\vartheta\ne\vartheta_0$,
\[
  Q_n^p(\vartheta)-Q_n^p(\vartheta_0)
  =
  \{Q^p(\vartheta)-Q^p(\vartheta_0)\}+o_p(1)
  \ge c+o_p(1),
\]
uniformly over the finite set $\Vcal\setminus\{\vartheta_0\}$.  Hence, with
probability approaching one, every $\vartheta\ne\vartheta_0$ has a larger
profile sample criterion than $\vartheta_0$, so
$\Pr(\hat\vartheta_n=\vartheta_0)\to1$.
\end{proof}

\begin{proof}[Proof of Corollary~\ref{cor:discrete-parameters}]
Selection consistency follows from Lemma~\ref{lem:finite-profile-selection}.  On
the event $\{\hat\vartheta_n=\vartheta_0\}$,
\[
  \hat\rho_n=\hat\rho_n(\vartheta_0).
\]
Therefore, for any $\varepsilon>0$,
\begin{align*}
&\Pr\!\left(
  \left\|
    \sqrt{N_n}(\hat\rho_n-\rho_0)
    -
    \sqrt{N_n}\{\hat\rho_n(\vartheta_0)-\rho_0\}
  \right\|>\varepsilon
\right) \\
&\qquad\le
\Pr(\hat\vartheta_n\ne\vartheta_0)\to0.
\end{align*}
Thus
\[
  \sqrt{N_n}(\hat\rho_n-\rho_0)
  =
  \sqrt{N_n}\{\hat\rho_n(\vartheta_0)-\rho_0\}+o_p(1).
\]
The conditional asymptotic-normality assumption then gives the oracle first-order
distribution.
\end{proof}


\section{\texorpdfstring{Details and Proofs for Section~\ref{sec:db-optimal-instruments}}{Details and proofs for design-based optimal instruments}}
\label{app:db-efficiency}

This appendix formalizes the optimal-moment construction used in
Section~\ref{sec:db-optimal-instruments}; Theorems~\ref{thm:oracle_variance_bound}
and~\ref{thm:sieve_oracle} below together constitute the formal statement of
Theorem~\ref{thm:efficiency-informal}.  The main object is the class of
product residualized moments
\[
  \phi(Y_i)
  \left\{
    \psi(W)-\E[\psi(W)\mid g(W;X_i,\theta)]
  \right\}.
\]
The appendix proceeds in five steps.  First,
Appendix~\ref{app:moment_characterization} shows that, for each fixed unit,
the vanishing of all such product residualized moments is equivalent to the
unit-level exposure-sufficiency condition. Second, Appendix~\ref{app:avg_moment_space}
constructs the averaged moment space used by the GMM estimator.  Third,
Appendix~\ref{app:oracle_geometry} derives the oracle moment through the
Riesz representation theorem and the Cauchy--Schwarz inequality.  Fourth,
Appendix~\ref{app:sieve_gmm} shows how finite dictionaries and sieve GMM
approximate the oracle moment.  Finally, Appendix~\ref{app:vector_parameters}
extends the argument to vector-valued exposure parameters.

We write
\[
  G_{i,\theta}:=g(W;X_i,\theta),
  \qquad
  G_i:=G_{i,\theta_0}.
\]
The design-side residual operator is
\[
  R_{i,\theta}{\psi}(W)
  :=
  \psi(W)-\E[\psi(W)\mid G_{i,\theta}].
\]

\subsection{Moment characterization of unit-level exposure sufficiency}
\label{app:moment_characterization}
\label{app:conditional-poirier}

The first result is a population characterization of the unit-level
conditional-independence restriction.  It is a statement about the joint law
of \((Y_i,W,G_i)\) under the design \(\Dcal_n\) for a fixed \((i,n)\).  It does
not require independence across units.

\begin{assumption}[Measurability]
\label{ass:standard_borel}
\label{ass:standard-borel-unitwise}
For each fixed experiment \(n\) and unit \(i\), the random variables
\(Y_i\), \(W\), and \(G_i=g(W;X_i,\theta_0)\) take values in standard Borel
spaces.  The maps \(w\mapsto Y_i(w)\) and
\(w\mapsto g(w;X_i,\theta_0)\) are measurable.  Conditional expectations are
understood as regular conditional expectations under the design \(\Dcal_n\).
\end{assumption}

\begin{lemma}[Product-moment characterization]
\label{lem:product_moment_characterization}
\label{lem:conditional-poirier}
Fix an experiment \(n\) and a unit \(i\).  Under Assumption
\ref{ass:standard_borel}, the following two statements are equivalent:
\begin{enumerate}
  \item
  \[
    Y_i \indep W \mid G_i .
  \]
  \item For every bounded measurable \(\phi\) and every bounded measurable
  \(\psi\),
  \[
    \E
    \left[
      \phi(Y_i)
      \left\{
        \psi(W)-\E[\psi(W)\mid G_i]
      \right\}
    \right]
    =0 .
  \]
\end{enumerate}
\end{lemma}

\begin{proof}
First suppose that
\[
  Y_i\indep W\mid G_i .
\]
Then, for every bounded measurable \(\psi\),
\[
  \E[\psi(W)\mid Y_i,G_i]
  =
  \E[\psi(W)\mid G_i].
\]
Hence, for every bounded measurable \(\phi\),
\[
\begin{aligned}
&\E
\left[
  \phi(Y_i)
  \left\{
    \psi(W)-\E[\psi(W)\mid G_i]
  \right\}
\right]
\\
&\qquad
=
\E
\left[
  \phi(Y_i)
  \left\{
    \E[\psi(W)\mid Y_i,G_i]
    -
    \E[\psi(W)\mid G_i]
  \right\}
\right]
=0.
\end{aligned}
\]

Conversely, suppose that
\[
  \E
  \left[
    \phi(Y_i)
    \left\{
      \psi(W)-\E[\psi(W)\mid G_i]
    \right\}
  \right]
  =0
\]
for every bounded measurable pair \((\phi,\psi)\).  Fix a bounded
measurable \(\psi\) and define
\[
  U_\psi(W)
  :=
  \psi(W)-\E[\psi(W)\mid G_i].
\]
We will show that
\[
  \E[U_\psi(W)\mid Y_i,G_i]=0.
\]
Let \(a\) and \(b\) be bounded measurable functions.  Since \(G_i\) is a
measurable function of \(W\), the product
\[
  \widetilde\psi(W):=b(G_i)\psi(W)
\]
is a bounded measurable function of \(W\).  Applying the assumed moment
condition with \(\phi(Y_i)=a(Y_i)\) and design function
\(\widetilde\psi(W)\), we obtain
\[
0
=
\E
\left[
  a(Y_i)
  \left\{
    b(G_i)\psi(W)
    -
    \E[b(G_i)\psi(W)\mid G_i]
  \right\}
\right].
\]
But
\[
  \E[b(G_i)\psi(W)\mid G_i]
  =
  b(G_i)\E[\psi(W)\mid G_i],
\]
and therefore
\[
  0
  =
  \E
  \left[
    a(Y_i)b(G_i)
    \left\{
      \psi(W)-\E[\psi(W)\mid G_i]
    \right\}
  \right]
  =
  \E[a(Y_i)b(G_i)U_\psi(W)].
\]
The preceding identity holds for every product function \(a(Y_i)b(G_i)\).
Since products of functions of \(Y_i\) and functions of \(G_i\) generate the
bounded measurable functions of the pair \((Y_i,G_i)\), the same identity
extends, by a standard monotone-class argument, to every bounded measurable
function \(h(Y_i,G_i)\):
\[
  \E[h(Y_i,G_i)U_\psi(W)]=0.
\] Hence
\[
  \E[U_\psi(W)\mid Y_i,G_i]=0.
\]
Equivalently,
\[
  \E[\psi(W)\mid Y_i,G_i]
  =
  \E[\psi(W)\mid G_i]
\]
for every bounded measurable \(\psi\).  This is precisely
\[
  Y_i\indep W\mid G_i .
\]
\end{proof}

\begin{remark}[Finite support is only a special case]
\label{rem:finite_support}
Lemma~\ref{lem:product_moment_characterization} does not require \(W\) to
have finite support.  If \(W\) has finite support \(\Wcal_n\), then every
bounded measurable design function can be written as
\[
  \psi(W)=\sum_{w\in\Wcal_n}\psi(w)\mathbf 1\{W=w\}.
\]
In that special case, it is enough to verify the moment restrictions for
the atom indicators \(\mathbf 1\{W=w\}\).  With continuous-support or
large-support designs, the exact population characterization uses all
bounded measurable design functions, while feasible implementation uses a
countable determining dictionary.
\end{remark}

For feasible implementation it is useful to replace the full class of
bounded measurable functions by countable dense dictionaries.

\begin{assumption}[Countable determining dictionaries]
\label{ass:countable_dictionaries}
\label{ass:countable-dictionaries}
For each fixed \((i,n)\), \(L^2(Y_i)\) and \(L^2(W)\) are separable under
the design \(\Dcal_n\).  Let
\[
  \Phi_{i,n}=\{\phi_{\ell,i,n}:\ell\geq1\},
  \qquad
  \Psi_n=\{\psi_{m,n}:m\geq1\}
\]
be countable dense subsets of \(L^2(Y_i)\) and \(L^2(W)\), respectively.
\end{assumption}

\begin{lemma}[Countable moment characterization]
\label{lem:countable_characterization}
\label{lem:countable-characterization}
Suppose Assumptions~\ref{ass:standard_borel} and
\ref{ass:countable_dictionaries} hold.  Then
\[
  Y_i\indep W\mid G_i
\]
if and only if
\[
  \E
  \left[
    \phi_{\ell,i,n}(Y_i)
    \left\{
      \psi_{m,n}(W)-\E[\psi_{m,n}(W)\mid G_i]
    \right\}
  \right]
  =0
\]
for every \(\ell,m\geq1\).
\end{lemma}

\begin{proof}
The forward direction follows from Lemma~\ref{lem:product_moment_characterization}.
For the converse, define the bilinear form
\[
  B_{i,n}(\phi,\psi)
  :=
  \E
  \left[
    \phi(Y_i)
    \left\{
      \psi(W)-\E[\psi(W)\mid G_i]
    \right\}
  \right].
\]
By Cauchy--Schwarz and the fact that conditional expectation is an
\(L^2\)-projection,
\[
\begin{aligned}
  |B_{i,n}(\phi,\psi)|
  &\leq
  \|\phi(Y_i)\|_{L^2}
  \left\|
    \psi(W)-\E[\psi(W)\mid G_i]
  \right\|_{L^2}
  \\
  &\leq
  \|\phi(Y_i)\|_{L^2}
  \|\psi(W)\|_{L^2}.
\end{aligned}
\]
Thus \(B_{i,n}\) is continuous on \(L^2(Y_i)\times L^2(W)\).  If it
vanishes on the dense product dictionary \(\Phi_{i,n}\times\Psi_n\), it
vanishes on all of \(L^2(Y_i)\times L^2(W)\), and therefore on all bounded
measurable \((\phi,\psi)\).  Lemma~\ref{lem:product_moment_characterization}
then implies \(Y_i\indep W\mid G_i\).
\end{proof}

\begin{remark}[No cross-unit independence is used]
\label{rem:no_cross_unit_independence}
The equivalence in Lemmas~\ref{lem:product_moment_characterization} and
\ref{lem:countable_characterization} is unit-level.  It concerns the joint
law of \((Y_i,W,G_i)\) under the assignment design \(\Dcal_n\) for fixed
\((i,n)\).  It does not require \(Y_i\) and \(Y_j\), or the moment kernels
for different units, to be independent.  Cross-unit dependence affects the
variance and limiting distribution of averaged moments, not the validity of
the unit-level moment characterization.
\end{remark}

\subsection{Averaged product residualized moment space}
\label{app:avg_moment_space}

Lemma~\ref{lem:product_moment_characterization} gives a unit-level
equivalence.  The GMM estimator uses averaged versions of these unit-level
restrictions.  For an admissible pair \((\phi,\psi)\), define
\[
  \eta_{\phi,\psi,i,n}(\theta)
  :=
  \phi(Y_i)R_{i,\theta}{\psi}(W)
  =
  \phi(Y_i)
  \left\{
    \psi(W)-\E[\psi(W)\mid G_{i,\theta}]
  \right\}.
\]
The corresponding averaged product residualized moment is
\[
  \eta_{\phi,\psi,n}(\theta)
  :=
  \frac{1}{N_n}\sum_{i\in\Ucal_n}\eta_{\phi,\psi,i,n}(\theta).
\]
At \(\theta=\theta_0\), unit-level exposure sufficiency implies
\[
  \E[\eta_{\phi,\psi,i,n}(\theta_0)]=0
  \qquad
  \text{for every } i,
\]
and hence
\[
  \E[\eta_{\phi,\psi,n}(\theta_0)]=0.
\]


Let \(\Hcal_0\) denote the linear span of averaged product residualized
moment directions evaluated at \(\theta_0\).  A generic direction
\(\eta\in\Hcal_0\) is generated by finite linear combinations
\[
  \eta_{i,n}
  =
  \sum_{r=1}^R a_r\phi_r(Y_i)R_{i,\theta_0}{\psi_r}(W),
\]
with experiment-level average
\[
  \eta_n
  =
  \frac{1}{N_n}\sum_{i\in\Ucal_n}\eta_{i,n}.
\]
For two such directions \(\eta,\xi\), define the design-based asymptotic
covariance bilinear form
\[
  \Omega(\eta,\xi)
  :=
  \lim_{n\to\infty}
  N_n\Cov(\eta_n,\xi_n),
\]
whenever the limit exists.  Write
\[
  \|\eta\|_\Omega^2:=\Omega(\eta,\eta).
\]

\begin{assumption}[Asymptotic covariance Hilbert space]
\label{ass:hilbert_space}
The following conditions hold.
\begin{enumerate}
  \item For every \(\eta,\xi\in\Hcal_0\), the limit
  \(\Omega(\eta,\xi)\) exists and is finite.
  \item \(\Omega\) is symmetric and positive semidefinite on \(\Hcal_0\).
  \item After quotienting out zero-variance directions
  \[
    \Ncal_\Omega
    :=
    \{\eta\in\Hcal_0:\|\eta\|_\Omega=0\},
  \]
  the completion of \(\Hcal_0/\Ncal_\Omega\) under \(\|\cdot\|_\Omega\) is a
  Hilbert space, denoted
  \[
    (\Hcal,\langle\cdot,\cdot\rangle_\Omega).
  \]
\end{enumerate}
\end{assumption}

The covariance \(\Omega\) is the covariance of averaged design-based
moments.  It contains all cross-unit covariance terms:
\[
  N_n\Var(\eta_n)
  =
  \frac{1}{N_n}
  \sum_{i\in\Ucal_n}\sum_{j\in\Ucal_n}
  \Cov(\eta_{i,n},\eta_{j,n}).
\]
Thus cross-unit dependence is not ruled out.  It is summarized by the
asymptotic covariance metric and handled by the LLN, CLT, and graph-HAC
conditions in Section~\ref{sec:DB-consistency-AN-affinity}.

\subsection{Oracle moment for scalar exposure parameters}
\label{app:oracle_geometry}

This subsection treats the scalar case \(\theta_0\in\R\).  The vector case
is given in Appendix~\ref{app:vector_parameters}.

For a direction \(\eta\in\Hcal\), let \(\eta_n(\theta)\) denote the
corresponding averaged product residualized moment at candidate parameter
\(\theta\), and define the population mean map
\[
  \mu_\eta(\theta)
  :=
  \lim_{n\to\infty}
  \E[\eta_n(\theta)],
\]
whenever the limit exists.  At the true value,
\[
  \mu_\eta(\theta_0)=0.
\]
Define the local derivative functional
\[
  G(\eta)
  :=
  \left.
  \frac{\partial}{\partial\theta}\mu_\eta(\theta)
  \right|_{\theta=\theta_0}.
\]
For $\eta$ in the linear span of the dictionary directions,
$\eta_n(\theta)$ is the corresponding finite linear combination of sample
moments and $\mu_\eta(\theta)$ is defined directly. A general $\eta\in\Hcal$
is an $\Omega$-norm limit of span elements and need not come with a canonical
finite-sample moment path away from $\theta_0$. Accordingly, off-$\theta_0$
objects such as $\mu_\eta(\theta)$ and the just-identified estimator in
part~3 below are used only for span directions; $G$ is then extended from the
span to all of $\Hcal$ by continuity, using the boundedness in part~2, and
statements for general $\eta\in\Hcal$ are understood in this extension sense.

\begin{assumption}[Local regularity]
\label{ass:local_regularity}
The following conditions hold.
\begin{enumerate}
  \item For every \(\eta\in\Hcal\), the map \(\theta\mapsto\mu_\eta(\theta)\)
  is differentiable in a neighborhood of \(\theta_0\), with derivative
  \(G(\eta)\) at \(\theta_0\).
  \item The map \(G:\Hcal\to\R\) is a bounded linear functional under the
  \(\Omega\)-norm.  That is, there exists \(C_G<\infty\) such that
  \[
    |G(\eta)|\leq C_G\|\eta\|_\Omega
    \qquad
    \text{for all }\eta\in\Hcal.
  \]
  \item For every \(\eta\in\Hcal\) with \(G(\eta)\neq0\), the just-identified
  estimator based on
  \[
    \eta_n(\theta)=0
  \]
  admits the local expansion
  \[
    \sqrt{N_n}(\hat\theta_{\eta,n}-\theta_0)
    =
    -
    \frac{\sqrt{N_n}\eta_n(\theta_0)}{G(\eta)}
    +o_p(1),
  \]
  and
  \[
    \sqrt{N_n}\eta_n(\theta_0)
    \Rightarrow
    \Ncal(0,\|\eta\|_\Omega^2).
  \]
  \item \textbf{Nonzero information.} The derivative functional does
  not vanish identically on \(\Hcal\): there exists \(\eta\in\Hcal\) with
  \(G(\eta)\neq0\). Equivalently, via the Riesz representation below,
  \(\|\eta^\star\|_\Omega>0\), so that \(V^\star=1/\|\eta^\star\|_\Omega^2\) is
  finite and the infimum in Theorem~\ref{thm:oracle_variance_bound} runs over a
  nonempty set. If instead \(G\equiv0\), the exposure parameter is unidentified
  at first order by this moment class and no \(\sqrt{N_n}\)-consistent estimator
  built from it exists; empirically this corresponds to a flat GMM criterion
  (weak identification).
\end{enumerate}
\end{assumption}

\begin{lemma}[Variance of a scalar moment direction]
\label{lem:scalar_direction_variance}
Under Assumptions~\ref{ass:hilbert_space} and
\ref{ass:local_regularity}, the asymptotic variance of the just-identified
estimator based on a scalar direction \(\eta\in\Hcal\) with \(G(\eta)\neq0\) is
\[
  V(\eta)
  =
  \frac{\|\eta\|_\Omega^2}{G(\eta)^2}.
\]
\end{lemma}

\begin{proof}
The expansion in Assumption~\ref{ass:local_regularity} gives
\[
  \sqrt{N_n}(\hat\theta_{\eta,n}-\theta_0)
  =
  -
  \frac{\sqrt{N_n}\eta_n(\theta_0)}{G(\eta)}
  +o_p(1).
\]
Since
\[
  \sqrt{N_n}\eta_n(\theta_0)
  \Rightarrow
  \Ncal(0,\|\eta\|_\Omega^2),
\]
the continuous mapping theorem implies
\[
  \sqrt{N_n}(\hat\theta_{\eta,n}-\theta_0)
  \Rightarrow
  \Ncal\left(0,\frac{\|\eta\|_\Omega^2}{G(\eta)^2}\right).
\]
\end{proof}

By Assumption~\ref{ass:local_regularity}, \(G\) is a bounded linear
functional on \(\Hcal\).  The Riesz representation theorem therefore implies
that there exists a unique element \(\eta^\star\in\Hcal\) such that
\[
  G(\eta)
  =
  \langle \eta^\star,\eta\rangle_\Omega
  \qquad
  \text{for all }\eta\in\Hcal.
\]
We call \(\eta^\star\) the oracle product residualized moment.

\begin{lemma}[Cauchy--Schwarz lower bound]
\label{lem:cauchy_schwarz_bound}
Under Assumptions~\ref{ass:hilbert_space} and
\ref{ass:local_regularity}, for any \(\eta\in\Hcal\) with \(G(\eta)\neq0\),
\[
  V(\eta)
  =
  \frac{\|\eta\|_\Omega^2}{G(\eta)^2}
  \geq
  \frac{1}{\|\eta^\star\|_\Omega^2}.
\]
Equality holds if and only if \(\eta\) is proportional to \(\eta^\star\).
\end{lemma}

\begin{proof}
By the Riesz representation,
\[
  G(\eta)^2
  =
  \langle \eta^\star,\eta\rangle_\Omega^2.
\]
By Cauchy--Schwarz,
\[
  \langle \eta^\star,\eta\rangle_\Omega^2
  \leq
  \|\eta^\star\|_\Omega^2\|\eta\|_\Omega^2.
\]
Rearranging gives
\[
  \frac{\|\eta\|_\Omega^2}{G(\eta)^2}
  \geq
  \frac{1}{\|\eta^\star\|_\Omega^2}.
\]
The equality condition is the equality condition in Cauchy--Schwarz:
\(\eta\) must be proportional to \(\eta^\star\) in \(\Hcal\).
\end{proof}

\begin{theorem}[Oracle variance bound]
\label{thm:oracle_variance_bound}
Suppose Assumptions~\ref{ass:standard_borel}, \ref{ass:hilbert_space}, and
\ref{ass:local_regularity} hold.  Among scalar directions in the averaged
product residualized moment space \(\Hcal\), the oracle variance is
\[
  V^\star
  :=
  \frac{1}{\|\eta^\star\|_\Omega^2}.
\]
Equivalently,
\[
  \inf_{\eta\in\Hcal:G(\eta)\neq0}
  \frac{\|\eta\|_\Omega^2}{G(\eta)^2}
  =
  V^\star,
\]
and equality is attained if and only if the moment direction is
proportional to \(\eta^\star\).
\end{theorem}

\begin{proof}
The moment class is generated by the product residualized moments
characterized in Lemma~\ref{lem:product_moment_characterization}.  Within
the Hilbert space \(\Hcal\) induced by the asymptotic covariance metric,
Lemma~\ref{lem:cauchy_schwarz_bound} gives the sharp lower bound and the
equality condition.
\end{proof}

\begin{remark}[Within-class optimality]
\label{rem:within_class_optimality}
Theorem~\ref{thm:oracle_variance_bound} is a within-class optimality
result.  It does not characterize the globally efficient estimator under
all implications of exposure sufficiency.  In particular, exposure
sufficiency also implies multi-unit restrictions involving \((Y_i,Y_j)\),
joint exposure vectors, and more complicated kernels.  The oracle variance
\(V^\star\) is the efficiency bound for the averaged product residualized
moment space \(\Hcal\), which is the space used by the proposed GMM procedure.
\end{remark}

\begin{remark}[Instrument-only moments as a subspace]
\label{rem:instrument_subspace}
The instrument-only moments correspond to the special case
\[
  \phi(Y_i)=Y_i,
  \qquad
  \psi(W)=h(W).
\]
Their span is generally a strict subspace of \(\Hcal\).  Theorem
\ref{thm:oracle_variance_bound} optimizes over the larger product
residualized moment space generated by arbitrary admissible \((\phi,\psi)\).
\end{remark}

\subsection{Finite dictionaries and sieve GMM}
\label{app:sieve_gmm}

The oracle moment \(\eta^\star\) is generally infinite-dimensional.  The
feasible estimator approximates it using a finite dictionary of product
residualized moments.

Let
\[
  \{(\phi_m,\psi_m):m\geq1\}
\]
be a countable dictionary of outcome transformations and design-side
functions.  For each \(m\), define
\[
  \eta_{m,n}(\theta)
  :=
  \frac{1}{N_n}\sum_{i\in\Ucal_n}
  \phi_m(Y_i)
  \left\{
    \psi_m(W)-\E[\psi_m(W)\mid G_{i,\theta}]
  \right\}.
\]
For a finite dictionary of size \(M\), stack
\[
  \eta_n^{(M)}(\theta)
  :=
  \left(
    \eta_{1,n}(\theta),
    \ldots,
    \eta_{M,n}(\theta)
  \right)^\top.
\]
Let
\[
  G^{(M)}
  :=
  \left(G(\eta_1),\ldots,G(\eta_M)\right)^\top,
  \qquad
  \Omega^{(M)}
  :=
  \left(\Omega(\eta_k,\eta_\ell)\right)_{k,\ell\leq M}.
\]

\begin{lemma}[Finite-dimensional oracle]
\label{lem:finite_dimensional_oracle}
Suppose \(\Omega^{(M)}\) is nonsingular \redrev{and \(G^{(M)}\neq0\)}.  For scalar \(\theta_0\), the
optimal linear combination
\[
  a^\top\eta_n^{(M)}(\theta)
\]
solves
\[
  \min_{a:a^\top G^{(M)}\neq0}
  \frac{a^\top\Omega^{(M)}a}{(a^\top G^{(M)})^2}.
\]
The solution is proportional to
\[
  a_M^\star
  =
  \Omega^{(M)-1}G^{(M)},
\]
and the finite-dictionary oracle variance is
\[
  V_M^\star
  =
  \frac{1}{G^{(M)\top}\Omega^{(M)-1}G^{(M)}}.
\]
\end{lemma}

\begin{proof}
Normalize \(a^\top G^{(M)}=1\).  The problem becomes
\[
  \min_a a^\top\Omega^{(M)}a
  \quad
  \text{subject to}
  \quad
  a^\top G^{(M)}=1.
\]
The Lagrangian first-order condition gives
\[
  2\Omega^{(M)}a-\lambda G^{(M)}=0,
\]
so
\[
  a\propto \Omega^{(M)-1}G^{(M)}.
\]
Substituting the normalized solution yields
\[
  V_M^\star
  =
  \left\{
    G^{(M)\top}\Omega^{(M)-1}G^{(M)}
  \right\}^{-1}.
\]
\end{proof}

Equivalently, define the finite-dimensional Riesz representer
\[
  \eta_M^\star
  :=
  \sum_{m=1}^M\gamma_{M,m}^\star\eta_m,
  \qquad
  \gamma_M^\star
  :=
  \Omega^{(M)-1}G^{(M)}.
\]
Then \(\eta_M^\star\) is the Riesz representer of \(G\) restricted to
\[
  \Hcal_M:=\operatorname{span}\{\eta_1,\ldots,\eta_M\}.
\]

The feasible finite-dictionary GMM estimator is
\[
  \hat\theta_{M,n}
  \in
  \argmin_{\theta\in\Theta}
  \eta_n^{(M)}(\theta)^\top
  \Lambda_{M,n}
  \eta_n^{(M)}(\theta),
\]
where the efficient choice is
\[
  \Lambda_{M,n}=\widehat\Omega_{M,n}^{-1}.
\]
This yields the familiar two-step procedure: first obtain a preliminary
consistent estimate of \(\theta_0\), then estimate the asymptotic covariance
matrix of the stacked product residualized moments, and finally re-estimate
using the inverse covariance weight.

We next combine population sieve approximation with the paper's
fixed-dimensional feasible GMM theory.  The latter is imposed only pointwise
in the dictionary size; no growing-dimensional linearization or central limit
theorem is assumed.

\begin{assumption}[Dense sieve and pointwise feasible inference]
\label{ass:dense_sieve}
The following conditions hold.
\begin{enumerate}
  \item[(S1)] The linear span of \(\{\eta_m:m\geq1\}\) is dense in \(\Hcal\)
  under the \(\Omega\)-norm:
  \[
    \overline{
      \bigcup_{M\geq1}
      \operatorname{span}\{\eta_1,\ldots,\eta_M\}
    }^{\|\cdot\|_\Omega}
    =\Hcal.
  \]
  In continuous-support cases, this dictionary may be constructed from
  countable dense classes of outcome transformations and design-side
  functions as in Assumption~\ref{ass:countable_dictionaries}.

  \item[(S2)] For every fixed \(M\), the consistency and asymptotic-normality
  conditions of Theorem~\ref{thm:DB-AN-aff} hold for
  \(\eta_n^{(M)}(\theta)\), and the feasible two-step weight converges to
  \(\{\Omega^{(M)}\}^{-1}\), as in
  Assumption~\ref{ass:graph-hac}.  Thus,
  pointwise in \(M\),
  \[
    \sqrt{N_n}(\hat\theta_{M,n}-\theta_0)
    \Rightarrow
    \Ncal(0,V_M^\star).
  \]

  \item[(S3)] With \(b_n:=\max_{i\in\Ucal_n}|A_i|\),
  \[
    \frac{b_n}{N_n}\to0.
  \]
\end{enumerate}
\end{assumption}

\begin{theorem}[Oracle approximation and slow-sieve attainment]
\label{thm:sieve_oracle}
Suppose Assumptions~\ref{ass:hilbert_space}, \ref{ass:local_regularity},
and \ref{ass:dense_sieve} hold.  Then:
\begin{enumerate}[label=\textup{(\roman*)}]
  \item The finite-dimensional Riesz representers converge to the oracle
  representer:
  \[
    \|\eta_M^\star-\eta^\star\|_\Omega\to0,
    \qquad
    V_M^\star\to V^\star.
  \]

  \item There exists a deterministic, nondecreasing sequence
  \(M_n\to\infty\) such that
  \[
    \frac{M_n^2b_n}{N_n}\to0
  \]
  and the corresponding feasible sieve estimator satisfies
  \[
    \sqrt{N_n}(\hat\theta_{M_n,n}-\theta_0)
    \Rightarrow
    \Ncal(0,V^\star).
  \]
\end{enumerate}
\end{theorem}

\begin{proof}
For part~\textup{(i)}, Assumption~\ref{ass:dense_sieve}\textup{(S1)}
implies that the orthogonal projection of \(\eta^\star\) onto \(\Hcal_M\)
converges to \(\eta^\star\) in \(\|\cdot\|_\Omega\).  The
finite-dimensional Riesz representer on \(\Hcal_M\) is precisely
\(\eta_M^\star\), so
\[
  \|\eta_M^\star-\eta^\star\|_\Omega\to0.
\]
Since
\[
  V_M^\star=\frac{1}{\|\eta_M^\star\|_\Omega^2},
  \qquad
  V^\star=\frac{1}{\|\eta^\star\|_\Omega^2},
\]
the variance convergence follows by continuity and the nonzero-information
condition in Assumption~\ref{ass:local_regularity}.

For part~\textup{(ii)}, let \(d_{\mathrm{BL}}\) be the bounded-Lipschitz
metric on probability laws and write
\[
  X_{M,n}:=\sqrt{N_n}(\hat\theta_{M,n}-\theta_0).
\]
For every fixed \(M\), Assumption~\ref{ass:dense_sieve}\textup{(S2)} gives
\[
  d_{\mathrm{BL}}
  \left(
    \mathcal L(X_{M,n}),
    \Ncal(0,V_M^\star)
  \right)
  \to0.
\]
By this convergence and Assumption~\ref{ass:dense_sieve}\textup{(S3)}, one
can choose a strictly increasing sequence of integers \(n_M\), with
\(n_M\geq M\), such that, for every \(n\geq n_M\),
\[
  d_{\mathrm{BL}}
  \left(
    \mathcal L(X_{M,n}),
    \Ncal(0,V_M^\star)
  \right)
  \leq\frac1M,
  \qquad
  M^2\frac{b_n}{N_n}\leq\frac1M.
\]
Define
\[
  M_n:=\max\{M:n_M\leq n\},
\]
with \(M_n=1\) before \(n_1\).  Then \(M_n\) is deterministic,
nondecreasing, and tends to infinity.  Moreover,
\[
  \frac{M_n^2b_n}{N_n}\leq\frac1{M_n}\to0.
\]
Finally, the triangle inequality gives
\[
\begin{aligned}
  &d_{\mathrm{BL}}
  \left(
    \mathcal L(X_{M_n,n}),
    \Ncal(0,V^\star)
  \right)
  \\
  &\qquad\leq
  \frac1{M_n}
  +
  d_{\mathrm{BL}}
  \left(
    \Ncal(0,V_{M_n}^\star),
    \Ncal(0,V^\star)
  \right)
  \to0,
\end{aligned}
\]
where the last step uses part~\textup{(i)}.  Because
\(d_{\mathrm{BL}}\) metrizes weak convergence, this proves the claim.
\end{proof}

\begin{remark}[Scope of the slow-sieve result]
\label{rem:primitive_sieve_conditions}
Theorem~\ref{thm:sieve_oracle}\textup{(ii)} is an existence result.  It
shows that the fixed-dimensional feasible theory, population sieve
approximation, and the dimension restriction
\[
  M_n^2b_n/N_n\to0
\]
are jointly attainable along some sufficiently slow deterministic sequence.
It does not provide an explicit growth rule or claim that the displayed rate,
by itself, validates every prespecified growing sequence.  Such a result would
additionally require uniform growing-dimensional linearization, a
triangular-array central limit theorem for the efficient direction, and
uniform control of covariance estimation and estimated design projections.
\end{remark}

\subsection{Vector parameters and optimal linear aggregation}
\label{app:vector_parameters}

The previous subsections stated the oracle geometry for scalar
\(\theta_0\).  Suppose now that
\[
  \theta_0\in\R^p.
\]
For each coordinate \(j=1,\ldots,p\), define the coordinate derivative
functional
\[
  G_j(\eta)
  :=
  \left.
  \frac{\partial}{\partial\theta_j}\mu_\eta(\theta)
  \right|_{\theta=\theta_0}.
\]
Assume each \(G_j\) is a bounded linear functional on \(\Hcal\).  By Riesz,
there exist unique elements
\[
  s_1^\star,\ldots,s_p^\star\in\Hcal
\]
such that
\[
  G_j(\eta)
  =
  \langle s_j^\star,\eta\rangle_\Omega
  \qquad
  \text{for all }\eta\in\Hcal.
\]
Define the oracle information matrix
\[
  \Ical^\star
  :=
  \left(
    \langle s_j^\star,s_k^\star\rangle_\Omega
  \right)_{j,k=1}^p.
\]
If \(\Ical^\star\) is nonsingular, the vector-parameter oracle variance is
\[
  V^\star
  =
  (\Ical^\star)^{-1}.
\]

For a finite dictionary of size \(M\), let
\[
  G_M
  :=
  \left.
  \frac{\partial}{\partial\theta}
  \E[\eta_n^{(M)}(\theta)]
  \right|_{\theta=\theta_0}
  \in\R^{M\times p},
\]
and let \(\Omega_M\in\R^{M\times M}\) denote the asymptotic covariance matrix
of the stacked moments.  The finite-dimensional optimal GMM variance is
\[
  V_M^\star
  =
  (G_M^\top\Omega_M^{-1}G_M)^{-1}.
\]
Equivalently, the \(M\)-dimensional moment vector can be compressed without
first-order loss to the \(p\)-dimensional optimal moment
\[
  G_M^\top\Omega_M^{-1}\eta_n^{(M)}(\theta).
\]
The feasible two-step GMM estimator uses \(\widehat\Omega_M^{-1}\) in place
of \(\Omega_M^{-1}\).  Under the density condition in Assumption
\ref{ass:dense_sieve}\textup{(S1)},
\[
  V_M^\star\to V^\star,
\]
where \(V^\star=(\Ical^\star)^{-1}\).  If the fixed-dimensional feasible
GMM conditions hold pointwise for every \(M\), the same diagonal argument as
in Theorem~\ref{thm:sieve_oracle}\textup{(ii)} yields a sufficiently slow
deterministic sequence \(M_n\to\infty\) along which the feasible estimator
attains this vector oracle variance.

\begin{remark}[Discrete or non-smooth exposure parameters]
\label{rem:discrete_nonregular}
The Riesz derivative argument is a local regularity result.  It applies
most directly to smooth exposure parameters, such as decay rates in smooth
spatial kernels.  For rings, finite grids of radii, or network
\(K\)-hop exposure parameters, the relevant object is generally not the
local derivative
\[
  \partial_\theta\mu_\eta(\theta_0).
\]
Those cases should be treated as finite-grid or nonregular
selection problems.  The product residualized moment class remains the
source of valid design-based moments, but the oracle derivative geometry in
this appendix should be interpreted as applying to the regular smooth case.
\end{remark}

\section{\texorpdfstring{Details and Proofs for Section~\ref{sec:stage2}}{Details and proofs for the Stage-2 results}}
\label{app:stage2}

This appendix formalizes the two-step inference problem in Section~\ref{sec:stage2}.
The main text states the regular continuous case.  Here we also allow a finite-grid exposure component and show that, after it has been selected consistently, the remaining two-step estimator is a finite-dimensional Z-estimator based on a stacked map.
The first block is the \emph{projected} Stage~1 GMM equation.

Throughout, write the exposure-map parameter as
\[
  \theta=(\vartheta,\rho),
\]
where $\vartheta$ is a finite-grid component, possibly absent, and $\rho\in\R^{p_\rho}$ is a regular continuous component, possibly absent.
Under Corollary~\ref{cor:discrete-parameters}, assume that the finite-grid component is selected consistently:
\[
  \Pr(\hat\vartheta_n=\vartheta_0)\to1.
\]
All statements below are conditional on the selected class $\vartheta_0$.
If there is no finite-grid component, simply drop $\vartheta$ and read $\rho=\theta$.
If there is no regular continuous first-stage component, the first block below is absent.

\subsection{Two-step estimator as a stacked Z-estimator}
\label{subsec:two-step-def-app}

\paragraph{Stage 1.}
Let $\Psi_{1i}(\rho)\in\R^{q_1}$ denote the Stage~1 design-based moment vector conditional on $\vartheta_0$, with sample average
\[
  \bar\Psi_{1,n}(\rho)
  :=
  \frac{1}{N_n}\sum_{i\in\Ucal_n}\Psi_{1i}(\rho).
\]
Let
\[
  \mu_1(\rho)
  :=
  \lim_{n\to\infty}
  \frac{1}{N_n}\sum_{i\in\Ucal_n}\E[\Psi_{1i}(\rho)]
\]
denote the deterministic limit of the Stage~1 moment map, and define
\[
  G_1
  :=
  \frac{\partial \mu_1(\rho_0)}{\partial\rho^\top},
\]
with the derivative understood in the mean-differentiability sense of Assumption~\ref{ass:AN-AFF}.
The Stage~1 GMM estimator for the regular component satisfies the projected estimating equation
\begin{equation}
\label{eq:stage1-projected-ee-app}
  G_1^\top\Lambda_{1,n}\bar\Psi_{1,n}(\hat\rho_n)
  =
  o_p(N_n^{-1/2}),
\end{equation}
where $\Lambda_{1,n}$ is the Stage~1 weight matrix.
Equation \eqref{eq:stage1-projected-ee-app} is the same projected GMM/Z condition used in Theorem~\ref{thm:DB-AN-aff}.
When the Stage~1 moments are exactly identified, it reduces to the usual moment equation up to a nonsingular transformation.

\paragraph{Stage 2.}
For the linear exposure--response projection, define
\[
  g_i(\rho)
  :=
  g(W;X_i,\vartheta_0,\rho),
  \qquad
  Z_i(\rho)
  :=
  \begin{pmatrix}
    1\\
    g_i(\rho)
  \end{pmatrix}.
\]
The Stage~2 target $(\alpha_0,\beta_0)$ is defined by the population moment condition
\begin{equation}
\label{eq:linear-app}
  \lim_{n\to\infty}
  \frac{1}{N_n}\sum_{i\in\Ucal_n}
  \E\!\left[
    Z_i(\rho_0)
    \{Y_i-\alpha_0-\beta_0^\top g_i(\rho_0)\}
  \right]
  =0.
\end{equation}
The linear conditional mean restriction
\[
  \E[Y_i\mid g_i(\rho_0)]
  =
  \alpha_0+\beta_0^\top g_i(\rho_0)
\]
is a sufficient condition for \eqref{eq:linear-app}, but is not required for the projection interpretation.
Define the Stage~2 moment vector
\begin{equation}
\label{eq:m2-def-app}
  \Psi_{2i}(\rho,\alpha,\beta)
  :=
  Z_i(\rho)
  \{Y_i-\alpha-\beta^\top g_i(\rho)\}
  \in\R^{1+k},
\end{equation}
and its sample average
\[
  \bar\Psi_{2,n}(\rho,\alpha,\beta)
  :=
  \frac{1}{N_n}\sum_{i\in\Ucal_n}\Psi_{2i}(\rho,\alpha,\beta).
\]
The Stage~2 estimator solves
\begin{equation}
\label{eq:stage2-ee-app}
  \bar\Psi_{2,n}(\hat\rho_n,\hat\alpha_n,\hat\beta_n)
  =
  o_p(N_n^{-1/2}).
\end{equation}

\paragraph{Stacked map.}
Collect the regular parameters into
\[
  \zeta
  :=
  (\rho^\top,\alpha,\beta^\top)^\top,
  \qquad
  \zeta_0
  :=
  (\rho_0^\top,\alpha_0,\beta_0^\top)^\top .
\]
Define the stacked estimating map
\begin{equation}
\label{eq:Phi-bar-app}
  \bar\Phi_n(\zeta)
  :=
  \begin{pmatrix}
    G_1^\top\Lambda_{1,n}\bar\Psi_{1,n}(\rho)\\
    \bar\Psi_{2,n}(\rho,\alpha,\beta)
  \end{pmatrix}.
\end{equation}
Equations \eqref{eq:stage1-projected-ee-app} and \eqref{eq:stage2-ee-app} imply
\[
  \bar\Phi_n(\hat\zeta_n)=o_p(N_n^{-1/2}),
  \qquad
  \hat\zeta_n:=(\hat\rho_n^\top,\hat\alpha_n,\hat\beta_n^\top)^\top .
\]
Thus the two-step estimator is a Z-estimator for the regular parameter vector $\zeta$.

\subsection{Joint asymptotic normality}
\label{subsec:delta-linearization-app}

Let $\Phi(\zeta)$ denote the deterministic limit of \eqref{eq:Phi-bar-app}:
\[
  \Phi(\zeta)
  :=
  \begin{pmatrix}
    G_1^\top\Lambda_1\mu_1(\rho)\\
    \mu_2(\rho,\alpha,\beta)
  \end{pmatrix},
\]
where $\Lambda_{1,n}\pto\Lambda_1$ and
\[
  \mu_2(\rho,\alpha,\beta)
  :=
  \lim_{n\to\infty}
  \frac{1}{N_n}\sum_{i\in\Ucal_n}
  \E[\Psi_{2i}(\rho,\alpha,\beta)].
\]
By construction, $\Phi(\zeta_0)=0$.

\begin{assumption}[Two-step high-level regularity]
\label{ass:stage2-highlevel}
Conditional on the selected finite-grid exposure component $\vartheta_0$, the following conditions hold.
\begin{enumerate}[label=\textup{(S\arabic*)}, leftmargin=1.4cm]
  \item \label{ass:stage2-consistency} \textbf{Consistency and approximate zero.}
  The estimator satisfies $\hat\zeta_n\pto\zeta_0$ and
  \[
    \bar\Phi_n(\hat\zeta_n)=o_p(N_n^{-1/2}).
  \]

  \item \label{ass:stage2-mean-diff} \textbf{Mean differentiability and rank.}
  The map $\Phi$ is mean-differentiable at $\zeta_0$:
  \[
    \Phi(\zeta)
    =
    A(\zeta-\zeta_0)+r(\zeta),
    \qquad
    \|r(\zeta)\|=o(\|\zeta-\zeta_0\|),
  \]
  for a nonsingular matrix $A$.

  \item \label{ass:stage2-clt} \textbf{Joint affinity-set CLT.}
  The stacked centered map satisfies
  \[
    \sqrt{N_n}\{\bar\Phi_n(\zeta_0)-\Phi(\zeta_0)\}
    \Rightarrow
    \Ncal(0,\Omega_\Phi),
  \]
  for a finite positive definite matrix $\Omega_\Phi$.
  This condition should be verified by applying the affinity-set CLT to the stacked array underlying \eqref{eq:Phi-bar-app}; equivalently, the CLT must hold for every fixed linear combination of the Stage~1 projected moments and Stage~2 moments.

  \item \label{ass:stage2-equicont} \textbf{\redrev{Local stochastic expansion.}}
  \redrev{For every sequence $\zeta_n\pto\zeta_0$,
  \[
    \sqrt{N_n}
    \Big[
      \{\bar\Phi_n(\zeta_n)-\Phi(\zeta_n)\}
      -
      \{\bar\Phi_n(\zeta_0)-\Phi(\zeta_0)\}
    \Big]
    =
    o_p\big(1+\sqrt{N_n}\,\|\zeta_n-\zeta_0\|\big).
  \]}
\end{enumerate}
\end{assumption}

Assumption~\ref{ass:stage2-highlevel} is the exact analogue of the high-level Z-estimator assumptions used in Theorem~\ref{thm:DB-AN-aff}, applied to the stacked Stage~1/Stage~2 map.
Smooth exposure-response systems can verify it using derivative-level LLNs.
Non-smooth exposure maps, such as rings, can verify it through mean differentiability of $\Phi$ and stochastic equicontinuity of the centered stacked process; no sample-path differentiability of the ring exposure is required.

\begin{theorem}[Joint asymptotic normality for the two-step estimator]
\label{thm:joint-clt-app}
Suppose Assumption~\ref{ass:stage2-highlevel} holds. Then, under the designs $\Dcal_n$,
\[
  \sqrt{N_n}(\hat\zeta_n-\zeta_0)
  \Rightarrow
  \Ncal\!\left(0,V_\zeta\right),
  \qquad
  V_\zeta:=A^{-1}\Omega_\Phi A^{-\top}.
\]
\end{theorem}

\begin{proof}
\redrev{We apply Lemma~\ref{lem:vdvw-z-estimator} to the stacked Stage~1/Stage~2 map, taking $M_n:=\bar\Phi_n$, $M:=\Phi$, and parameter $\zeta$. Its hypotheses are supplied by Assumption~\ref{ass:stage2-highlevel}: consistency $\hat\zeta_n\pto\zeta_0$ and the approximate-zero condition $\bar\Phi_n(\hat\zeta_n)=o_p(N_n^{-1/2})$ from \textup{\ref{ass:stage2-consistency}}; $\Phi(\zeta_0)=0$ (the population moment vanishes at the truth) with nonsingular derivative $A$ from \textup{\ref{ass:stage2-mean-diff}}; the joint CLT $\sqrt{N_n}\{\bar\Phi_n(\zeta_0)-\Phi(\zeta_0)\}\Rightarrow\Ncal(0,\Omega_\Phi)$ from \textup{\ref{ass:stage2-clt}}; and the local stochastic expansion condition~\textup{(Z)} from \textup{\ref{ass:stage2-equicont}}. Because condition~\textup{(Z)} normalizes the remainder by $1+\sqrt{N_n}\|\zeta_n-\zeta_0\|$, the lemma delivers the $\sqrt{N_n}$ rate rather than presupposing it, and yields
\[
  \sqrt{N_n}(\hat\zeta_n-\zeta_0)
  =
  -A^{-1}\sqrt{N_n}\{\bar\Phi_n(\zeta_0)-\Phi(\zeta_0)\}+o_p(1).
\]
The joint CLT in \textup{\ref{ass:stage2-clt}} and Slutsky's theorem then give
$\sqrt{N_n}(\hat\zeta_n-\zeta_0)\Rightarrow\Ncal(0,A^{-1}\Omega_\Phi A^{-\top})=\Ncal(0,V_\zeta)$.}
\end{proof}

\begin{remark}[Discrete exposure components]
\label{rem:discrete-theta-stage2}
The theorem is stated after conditioning on a consistently selected finite-grid exposure component.
If $\theta=(\vartheta,\rho)$ and $\Pr(\hat\vartheta_n=\vartheta_0)\to1$, then
\[
  \hat\zeta_n(\hat\vartheta_n)
  =
  \hat\zeta_n(\vartheta_0)
  \qquad\text{w.p.a.1}.
\]
Consequently, finite-grid selection has no first-order contribution to the distribution in Theorem~\ref{thm:joint-clt-app}.
This statement relies on fixed-gap selection consistency; it is not a claim about nonregular post-selection inference under ties or local-to-zero gaps.
\end{remark}

\subsection{\texorpdfstring{Inference for the exposure--response coefficient $\beta_0$}{Inference for the exposure-response coefficient (beta-0)}}
\label{subsec:beta-inference-app}

Let $K_\beta$ denote the selection matrix that extracts the $\beta$ coordinates from $\zeta$:
\[
  \beta=K_\beta\zeta,
  \qquad
  \beta_0=K_\beta\zeta_0,
  \qquad
  \hat\beta_n=K_\beta\hat\zeta_n.
\]
Define
\[
  V_\beta:=K_\beta V_\zeta K_\beta^\top.
\]

\begin{corollary}[Asymptotic normality of $\hat\beta_n$]
\label{cor:beta-clt-app}
Under the assumptions of Theorem~\ref{thm:joint-clt-app},
\[
  \sqrt{N_n}(\hat\beta_n-\beta_0)
  \Rightarrow
  \Ncal(0,V_\beta).
\]
\end{corollary}

\begin{proof}
This follows immediately from Theorem~\ref{thm:joint-clt-app} and the continuous mapping theorem applied to the linear map $\zeta\mapsto K_\beta\zeta$.
\end{proof}

In practice, estimate $A$ and $\Omega_\Phi$ by plug-in analogues evaluated at $\hat\zeta_n$.
The covariance $\Omega_\Phi$ should be estimated using the same affinity-set or graph-HAC logic as in Section~\ref{sec:DB-consistency-AN-affinity}, applied to the stacked moment contributions underlying \eqref{eq:Phi-bar-app}: the projected Stage~1 contribution $G_1^\top\Lambda_1\Psi_{1i}(\hat\rho_n)$ and the Stage~2 contribution $\Psi_{2i}(\hat\rho_n,\hat\alpha_n,\hat\beta_n)$.
The resulting sandwich estimate
\[
  \hat V_\zeta=
  \hat A^{-1}\hat\Omega_\Phi\hat A^{-\top}
\]
can be projected to obtain
\[
  \hat V_\beta=K_\beta\hat V_\zeta K_\beta^\top.
\]
This delivers design-based inference for the Stage~2 exposure--response coefficients that accounts for regular continuous first-stage exposure uncertainty and for the same spatial/network dependence structure used in Stage~1.
\begin{remark}[Conservativeness of the Stage-2 variance estimator]
\label{rem:stage2-conservative-app}
Write $u_i:=Y_i-\alpha_0-\beta_0^\top g_i(\rho_0)$ and
\[
  \mu_{2,i}:=\E[\Psi_{2i}(\zeta_0)]=\E[Z_i(\rho_0)\,u_i].
\]
The projection target \eqref{eq:linear-app} imposes only
$\tfrac1{N_n}\sum_{i}\mu_{2,i}\to0$. The individual $\mu_{2,i}$ need not vanish,
and they are not identified from the realized assignment: each is a design
expectation involving $\widetilde{Y}_i$ at exposure values other than the
realized one. The covariance $\Omega_\Phi$ in
Assumption~\ref{ass:stage2-highlevel} is built from contributions centered at
these unit-level means, whereas the feasible $\hat\Omega_\Phi$ described
above centers at the common sample mean, the only
centering available. Evaluated at $\zeta_0$, with $A_i$ the affinity set of
unit $i$ (Section~\ref{sec:affset}; not the Jacobian $A$),
\[
  \E[\hat\Omega_\Phi]
  =\Omega_\Phi
  +\underbrace{\frac{1}{N_n}\sum_{i}\sum_{j\in A_i}
    \begin{pmatrix}0&0\\[2pt]0&\mu_{2,i}\mu_{2,j}^\top\end{pmatrix}}_{=:\,B_n}
  +o(1);
\]
the Stage-1 block and the cross block are unbiased because the Stage-1 unit
means vanish identically (Corollary~\ref{cor:design-moment-unit}), and
evaluating at $\hat\zeta_n$ instead of $\zeta_0$ adds an $o_p(1)$ term under
Assumption~\ref{ass:stage2-highlevel} without affecting the decomposition.

The sign of $B_n$ determines sharpness. The own-index terms
$\tfrac1{N_n}\sum_i\mu_{2,i}\mu_{2,i}^\top$ are positive semidefinite, so if the
affinity sets are singletons or block-diagonal, $B_n\succeq0$; since
$V_\zeta=A^{-1}\Omega_\Phi A^{-\top}$ and
$V_\beta=K_\beta V_\zeta K_\beta^\top$ are congruence transformations, which
preserve the positive semidefinite order, $\hat V_\beta$ then weakly overstates
$V_\beta$ and confidence intervals for $\beta_0$ are asymptotically valid but
conservative, with no assumption of correct specification. If the
exposure--response is correctly specified unit by unit, so
$\mu_{2,i}=0$ for every $i$, then $B_n=0$ and inference is exact. Under general
overlapping dependence the cross terms $\mu_{2,i}\mu_{2,j}^\top$ are not
sign-definite; a sufficient condition restoring $B_n\succeq0$ is local
alignment of the unit-level means,
$\sum_i\sum_{j\in A_i}(c^\top\mu_{2,i})(c^\top\mu_{2,j})\ge0$ for every
direction $c$, as when misspecification varies smoothly relative to the
affinity radius. Absent such structure we confine the conservativeness
statement to the block-diagonal case.

\end{remark}

\section{Additional Application Details}

The appendix follows the order of the main applications. We first report
additional implementation details for the
\citet{muralidharan2023generaleq} application, where the design-based procedure
largely validates the original exposure scale. We then report additional results
for the \citet{egger2022general} application, where the exposure choice is more
consequential for the downstream IV estimates and multiplier accounting.

\subsection{\texorpdfstring{Additional Stage-1 implementation details for \citet{muralidharan2023generaleq}}{Additional Stage-1 implementation details for Muralidharan et al. (2023)}}
\label{app:mns-stage1}

This appendix provides further details on the Stage~1 exposure-mapping
exercise for \citet{muralidharan2023generaleq} used in the main text.

Throughout, we work with the same GP-level outcomes as in the main text:
NREGS earnings, total income, and
wage-labor income. We use the population-weighted
village-level ring exposure measure and village-consistent annulus
average instruments (``bands'' specification), the strict placebo design, and
the two-step fixed-weight GMM criterion searched over the 0.1 km grid
$\mathcal R=\{0.1,0.2,\ldots,30.0\}$ km. The baseline in the main text uses
3 km-wide bands up to $r_{\max}=20$ km ($M=7$ moments).

\paragraph{Two-step criterion and Wald interval.}
Let $\hat\mu(\theta)$ denote the stacked design-based moment vector for outcome
$Y$ at radius $\theta$, computed using the mandal-level randomization. The
Stage~1 estimator minimizes a two-step fixed-weight criterion. In the first
step we form the identity-weighted pilot criterion
\[
  \widehat Q_1(\theta)
  \;:=\;
  \hat\mu(\theta)^\top \hat\mu(\theta),
  \qquad
  \hat\theta_1
  \;:=\;
  \arg\min_{\theta\in\mathcal R} \widehat Q_1(\theta).
\]
In the second step we estimate the weighting matrix at the pilot radius and
hold it fixed: $\widehat W_0 := \widehat S(\hat\theta_1)^{-1}$, where
$\widehat S(\theta)$ is the design-based covariance estimator for the moment
vector at radius $\theta$. The radius estimator is then
\[
  \hat\theta
  \;:=\;
  \arg\min_{\theta\in\mathcal R}
  \widehat Q_2(\theta),
  \qquad
  \widehat Q_2(\theta)
  \;:=\;
  \hat\mu(\theta)^\top \widehat W_0\hat\mu(\theta).
\]
The plotted objective in all figures is $\widehat Q_2(\theta)/n$ for scale
comparability; this rescaling does not affect the minimizing radius.

Regarding inference: because physical interaction distances vary
continuously and the 0.1 km grid step is fine, the asymptotic normality
result of Section~\ref{sec:DB-consistency-AN-affinity} applies, and we report
a Wald-style standard error using a finite-difference Jacobian evaluated at
the neighboring grid points,
\[
  \widehat{\mathrm{se}}(\hat\theta)
  \;=\;
  \left[
    n \Bigl( \hat G^\top \widehat W_0 \hat G \Bigr)
  \right]^{-1/2},
  \qquad
  \hat G
  \;:=\;
  \frac{
    \hat\mu(\hat\theta + h) - \hat\mu(\hat\theta - h)
  }{2h},
\]
where $h = 0.1$ km is the grid step. The 95\% Wald interval is
$\hat\theta \pm 1.96 \cdot \widehat{\mathrm{se}}(\hat\theta)$.
\rev{One regularity point is worth making explicit here. The \emph{sample}
criterion in the ring radius is piecewise constant, with jumps at observed
inter-unit distances, so the Wald interval is justified by smoothness of the
\emph{population} (design-expectation) moment map $\theta\mapsto\mu(\theta)$ near
$\hat\theta$, not by the fineness of the grid; a fine grid only ensures that the
finite-difference Jacobian approximates the population derivative well. Where
that population smoothness is in doubt (e.g.,\ the flatter total- and wage-income
objectives), the exercise is better read as finite-grid model selection over the
candidate radii, and we treat the reported Wald intervals as approximate in those
cases.}

The fixed-radius overidentification diagnostic at the selected radius is
\[
  J(\hat\theta)
  \;:=\;
  n\,\hat\mu(\hat\theta)^\top \widehat S(\hat\theta)^{-1}\hat\mu(\hat\theta),
\]
which, \rev{because it is evaluated at the radius $\hat\theta$ chosen by the
criterion, is under the null asymptotically $\chi^2$ with
$\mathrm{df} = \mathrm{rank}\!\bigl(\widehat S(\hat\theta)\bigr) - \dim\theta$
degrees of freedom ($\dim\theta=1$ for the scalar radius), i.e.,\ the $M-1$
convention used in the table notes}. A large $p$-value suggests that the selected radius is
consistent with the moment restrictions; it does not imply that a wide range
of radii are equally plausible.

\subsubsection{Randomization design and placebo assignments}

Our Stage~1 procedure relies only on the mandal-level randomization of the
Smartcard rollout. \citet{muralidharan2023generaleq} randomized 296 mandals
in eight districts into multiple implementation waves, stratifying by district
and by a principal component of mandal socio-economic characteristics. The
randomization was conducted in two batches (six ``v5'' districts and two
``v6'' districts) using R scripts that (i) construct the first principal
component of stratification variables within each district, (ii) sort mandals
by district, revenue division, and this principal component, (iii) assign
mandals cyclically to a set of intermediate waves (seven waves in the first
batch and ten waves in the second), (iv) permute wave labels within
revenue-division $\times$ wave-cell groups, and (v) collapse these
intermediate waves into treatment, buffer, and control groups while enforcing
district-specific targets for the number of treatment and control mandals.

To align our placebo assignments with this design as closely as possible, we
obtained the original randomization scripts from the authors and implemented
an analogous ``strict'' re-randomization procedure in our code.\footnote{We
thank the authors for sharing the original randomization code. Our
implementation follows their logic but operates on the replication files
rather than the original merged census data.}

A subtle but important limitation is that the original experiment included a
third, ``buffer'' wave of mandals that were scheduled to receive Smartcards
after the treatment group, but no household surveys were conducted in these
buffer mandals. As \citet{muralidharan2023generaleq} emphasize, the survey
data are only available for treatment and control mandals, and the timing and
intensity of Smartcard rollout in buffer mandals are not fully observed. In
Appendix~B.1, they therefore treat buffer mandals as effectively untreated
when constructing neighborhood-treatment measures, and show that re-weighting
buffer mandals as partially treated (e.g.,\ 10--50 percent of full intensity)
yields qualitatively similar results and typically larger estimated total
effects, suggesting that the baseline specification is conservative. Following
the same spirit, our re-randomization scheme is defined over the study mandals
(those with survey data and $T_m\in\{0,1\}$) and treats the buffer group as a
latent complement that is not part of the analysis sample. Thus our placebo
assignments approximate the conditional randomization distribution rather than
the full distribution over all 296 mandals including the buffer wave.

In addition to this strict placebo design, which mirrors the wave structure
and district-level targets in the original randomization, we also consider a
simpler robustness specification in which we re-draw treatment indicators
within strata defined by district and discrete bins of the principal
component, holding fixed the number of treated mandals in each stratum. This
``simplified'' design preserves the key stratification features of the
experiment but abstracts from its detailed wave structure. Our main results
are based on the strict design; the simplified design yields qualitatively
similar Stage~1 radius estimates and objective values, reinforcing that our
conclusions are driven by the core features of the mandal-level randomization
rather than by fine details of the implementation.

\subsection{Additional results for Egger et al.\ (2022)}
\label{app:egger-extra}

This appendix gives the outcome-level results behind the Egger et al. application
in Section~\ref{subsec:egger-main}. The main text reports the three aggregates
most relevant for the multiplier discussion. Here we report the same
selected-support IV comparison for all outcomes and display the full set of
radius-path diagnostics.

\subsubsection*{A. Linear exposure index and IV implementation}

For each candidate support $R_m=2m$ km, we use Egger et al.'s annular transfer
variables to form the linear spillover exposure index
\[
g_v^{(m)}(W;\theta_m)
=
\sum_{j=1}^{m}\theta_{m,j}T^{\neg v}_{v,2(j-1)-2j}(W),
\qquad \mathbf 1'\theta_m=1.
\]
Stage~1 estimates the index weights $\theta_m$ from the design-based
exposure-sufficiency moments. Stage~2 then keeps the Egger et al. IV outcome
equation and replaces the original scalar spillover variable with
$g_v^{(m)}(W;\widehat\theta_m)$. Thus the appendix tables vary the spillover
exposure map, but otherwise keep the downstream IV exercise as close as possible
to the original specification.

\subsubsection*{B. Full selected-support IV results}

Table~\ref{tab:appendix-IV-linear-index-selected} extends the main-text
comparison to all outcomes. \rev{It reports the original Egger et al.\ IV
estimates, the selected-support linear-index estimates, the change in point
estimate and standard error, the selected support $\hat R$, and the $p$-value of
the support test (the likelihood-ratio test of the full $20$ km model against
the restricted model that zeroes the annuli beyond $\hat R$) that drives the
selection.}

\begin{table}[htbp]
  \centering
  \caption{All outcome IV components: Egger et al. vs. selected linear exposure-index estimates}
  \label{tab:appendix-IV-linear-index-selected}
  \resizebox{0.98\linewidth}{!}{%
  \begin{tabular}{llccccc}
    \toprule
    Outcome & Group & Egger IV & Linear-index IV & $\Delta\beta,\ \Delta\mathrm{SE}$ & $R$ & Support p-value \\
    \midrule
    Household expenditure, annualized & Recipients & 338.6 (100.6) & 320.0 (69.4) & $-18.6,\ -31.2$ & 6.0 & 0.249 \\
     & Non-recipients & 334.7 (131.1) & 224.9 (85.2) & $-109.8,\ -45.9$ & 6.0 & 0.249 \\[0.35em]
    Non-durable expenditure, annualized & Recipients & 227.2 (91.4) & 224.5 (68.2) & $-2.7,\ -23.1$ & 6.0 & 0.302 \\
     & Non-recipients & 317.5 (126.9) & 222.1 (77.5) & $-95.4,\ -49.4$ & 6.0 & 0.302 \\[0.35em]
    Food expenditure, annualized & Recipients & 133.8 (60.6) & 85.8 (35.1) & $-48.0,\ -25.5$ & 4.0 & 0.100 \\
     & Non-recipients & 133.3 (63.7) & 54.6 (27.3) & $-78.7,\ -36.3$ & 4.0 & 0.100 \\[0.35em]
    Temptation goods expenditure, annualized & Recipients & 5.9 (8.4) & 7.4 (5.8) & $1.5,\ -2.6$ & 10.0 & 0.214 \\
     & Non-recipients & -0.7 (6.8) & 1.1 (4.6) & $1.8,\ -2.2$ & 10.0 & 0.214 \\[0.35em]
    Durable expenditure, annualized & Recipients & 109.0 (26.5) & 108.9 (14.6) & $-0.2,\ -11.9$ & 4.0 & 0.236 \\
     & Non-recipients & 8.4 (11.8) & 2.0 (4.4) & $-6.4,\ -7.4$ & 4.0 & 0.236 \\[0.35em]
    Assets (non-land, non-house), net borrowing & Recipients & 183.4 (55.2) & 177.3 (31.2) & $-6.1,\ -24.0$ & 4.0 & 0.250 \\
     & Non-recipients & 133.1 (101.3) & 41.6 (44.8) & $-91.5,\ -56.5$ & 4.0 & 0.250 \\[0.35em]
    Housing value & Recipients & 477.4 (56.1) & 425.3 (29.3) & $-52.1,\ -26.7$ & 4.0 & 0.298 \\
     & Non-recipients & 80.6 (231.6) & -32.6 (96.1) & $-113.2,\ -135.5$ & 4.0 & 0.298 \\[0.35em]
    Land value & Recipients & 158.5 (304.3) & 73.2 (156.1) & $-85.3,\ -148.2$ & 4.0 & 0.062 \\
     & Non-recipients & 544.8 (416.4) & 180.0 (175.3) & $-364.7,\ -241.2$ & 4.0 & 0.062 \\[0.35em]
    Household income, annualized & Recipients & 135.7 (82.1) & 97.2 (36.9) & $-38.5,\ -45.1$ & 4.0 & 0.152 \\
     & Non-recipients & 225.0 (100.1) & 81.1 (42.6) & $-143.8,\ -57.4$ & 4.0 & 0.152 \\[0.35em]
    Net value of household transfers received, annualized & Recipients & -7.4 (14.5) & -42.8 (45.3) & $-35.3,\ 30.8$ & 16.0 & 0.309 \\
     & Non-recipients & 8.8 (16.8) & -93.2 (77.5) & $-102.0,\ 60.7$ & 16.0 & 0.309 \\[0.35em]
    Profits (ag \& non-ag), annualized & Recipients & 35.9 (37.7) & 27.4 (21.4) & $-8.5,\ -16.3$ & 4.0 & 0.157 \\
     & Non-recipients & 36.4 (46.5) & 10.2 (18.6) & $-26.3,\ -27.9$ & 4.0 & 0.157 \\[0.35em]
    Wage earnings, annualized & Recipients & 73.7 (62.4) & 54.6 (30.0) & $-19.0,\ -32.4$ & 6.0 & 0.910 \\
     & Non-recipients & 182.6 (90.2) & 103.3 (36.8) & $-79.3,\ -53.4$ & 6.0 & 0.910 \\[0.35em]
    Tax paid, annualized & Recipients & -0.1 (2.2) & 1.6 (1.6) & $1.7,\ -0.6$ & 20.0 & 1.000 \\
     & Non-recipients & 1.7 (2.5) & 0.2 (1.7) & $-1.4,\ -0.8$ & 20.0 & 1.000 \\
    \bottomrule
  \end{tabular}%
  }
  \begin{flushleft}
  \footnotesize\emph{Notes:} Entries are IV point estimates with standard errors in parentheses. Egger IV values re-estimate the original 2~km Table~1 specification on our replication sample (cluster standard errors); they closely match the published Table~I values. Recipient linear-index entries report total effects; non-recipient entries report pooled spillover effects, both with joint delta-method standard errors at each outcome's selected support $\hat R$. The Support p-value column reports the likelihood-ratio support-test p-value at $\hat R$, the diagnostic shown in the paper radius-path figures.
  \end{flushleft}
\end{table}

\rev{The broader pattern matches the main text: recipient expenditure and asset
components are stable, while non-recipient and income-side outcomes are more
sensitive and noisier. The design-based step does not move every coefficient,
but it moves enough of the downstream accounting to matter.}

\subsubsection*{C. Full radius-path diagnostics}

Figures~\ref{fig:egger-radius-path-app-exp}--\ref{fig:egger-radius-path-app-tax}
report the radius-path diagnostics for the full outcome list. The shaded region
marks the selected support. The lower panel in each figure reports the support
$p$-value path, the likelihood-ratio test of the full $20$ km model against the
restricted model at each candidate support. \bluerev{Diamonds report the
sublocation-level split-sample medians with their median intervals (dotted
whiskers), as in the main text.}

\begin{figure}[htbp]
\centering
\includegraphics[width=0.46\textwidth]{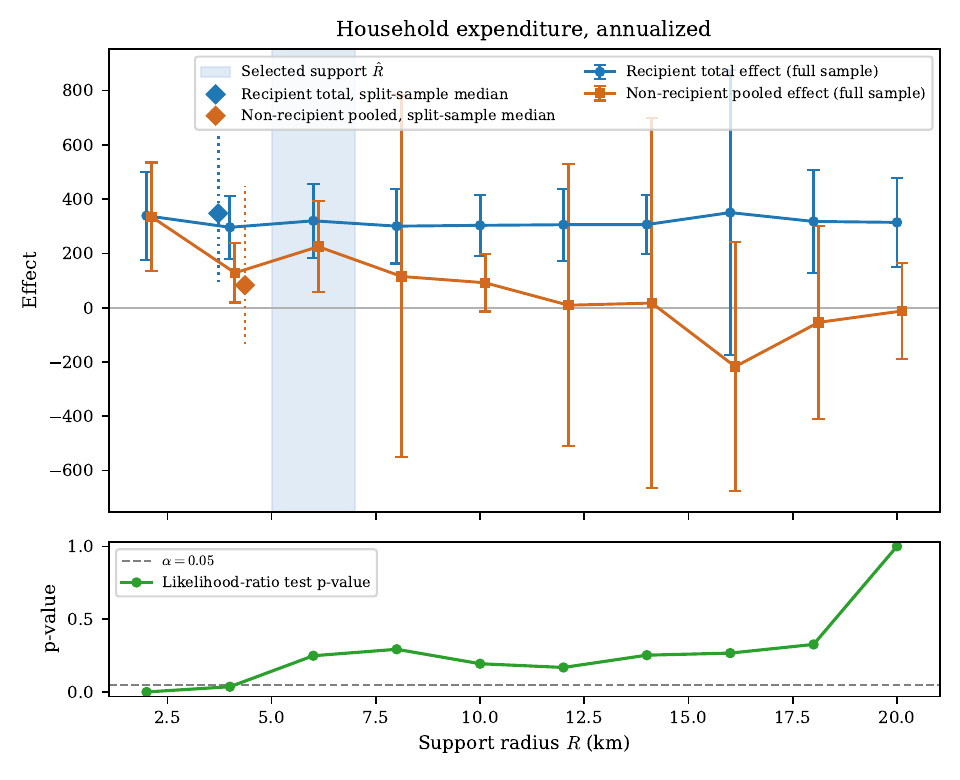}
\includegraphics[width=0.46\textwidth]{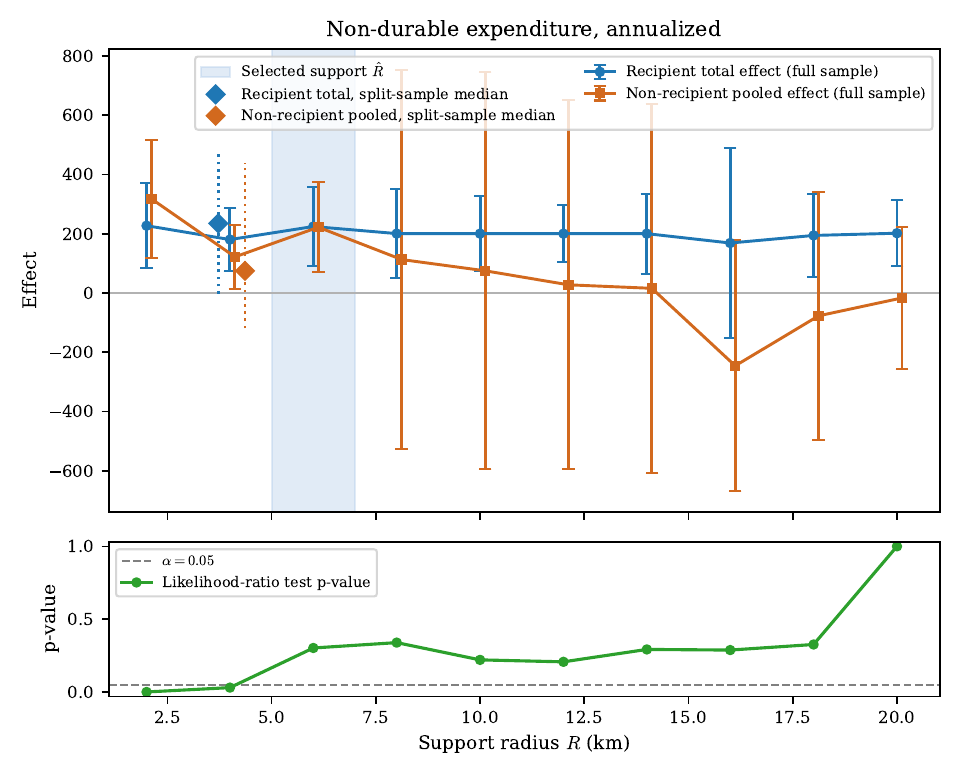}

\vspace{0.3em}

\includegraphics[width=0.46\textwidth]{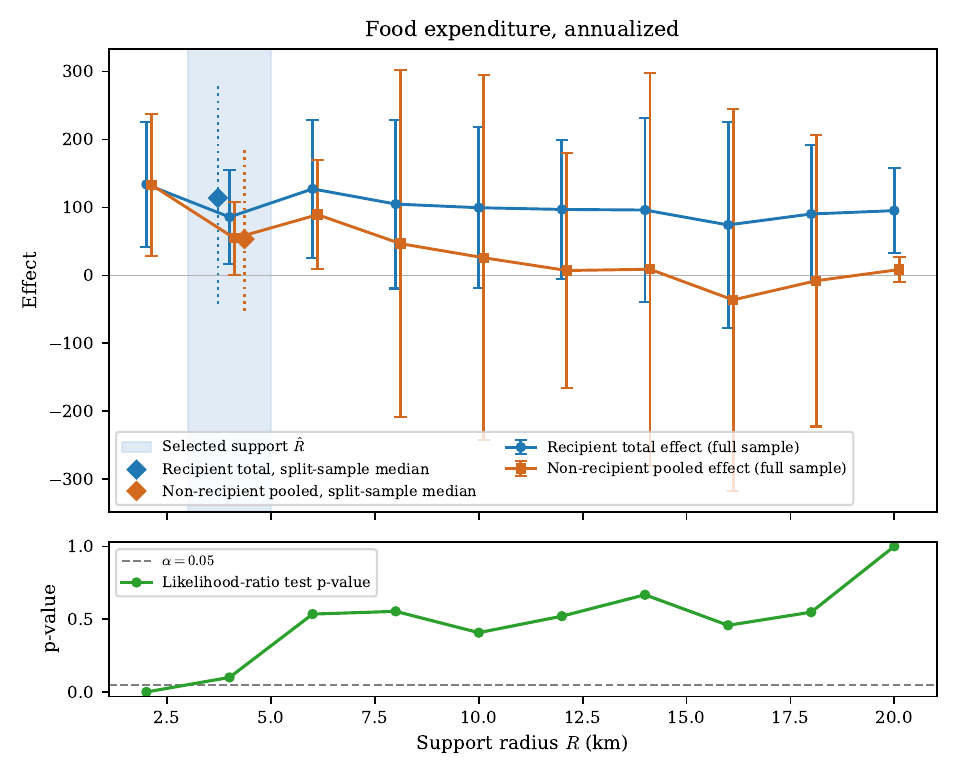}
\includegraphics[width=0.46\textwidth]{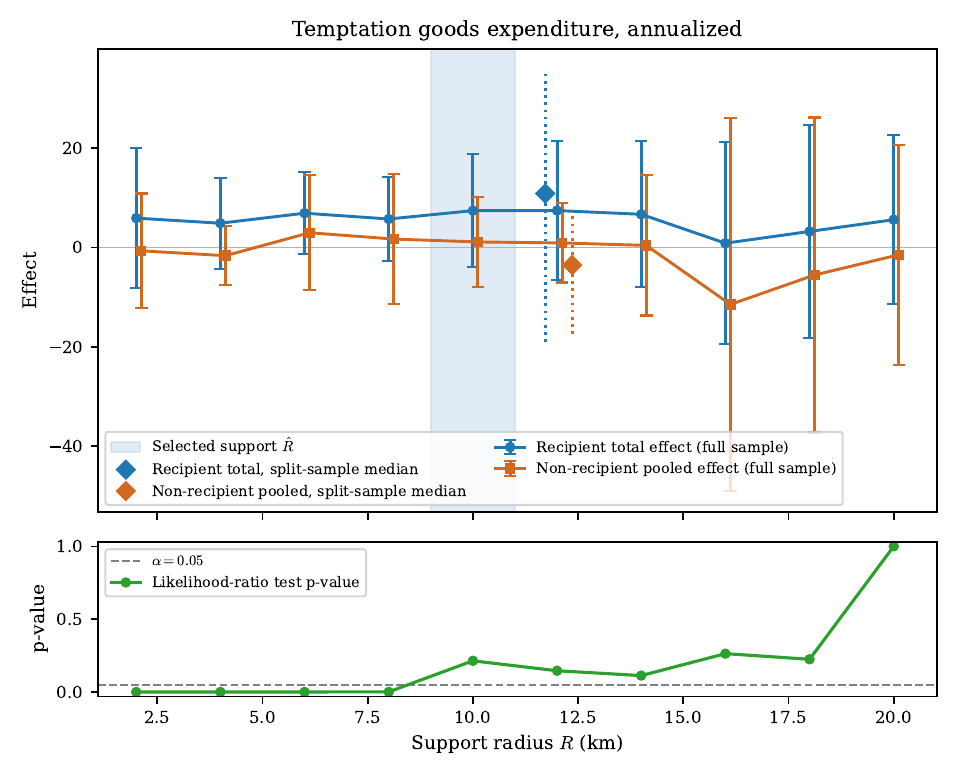}

\vspace{0.3em}

\includegraphics[width=0.46\textwidth]{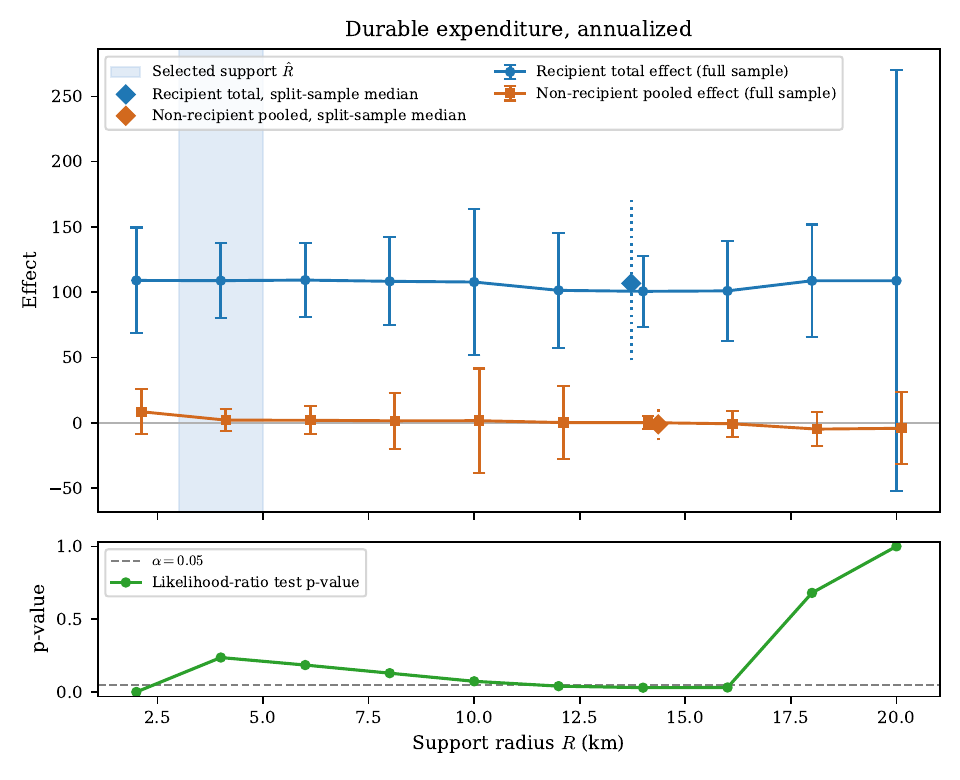}
\includegraphics[width=0.46\textwidth]{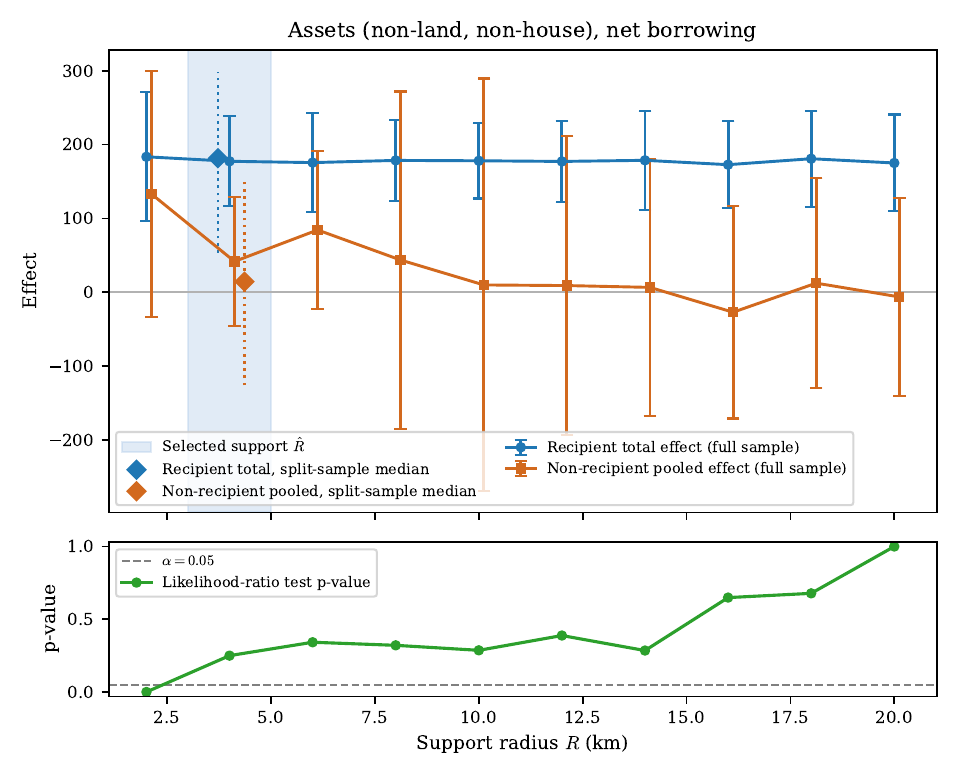}
\caption{Radius-path diagnostics for expenditure and asset outcomes in the Egger application}
\label{fig:egger-radius-path-app-exp}
\end{figure}

\begin{figure}[htbp]
\centering
\includegraphics[width=0.46\textwidth]{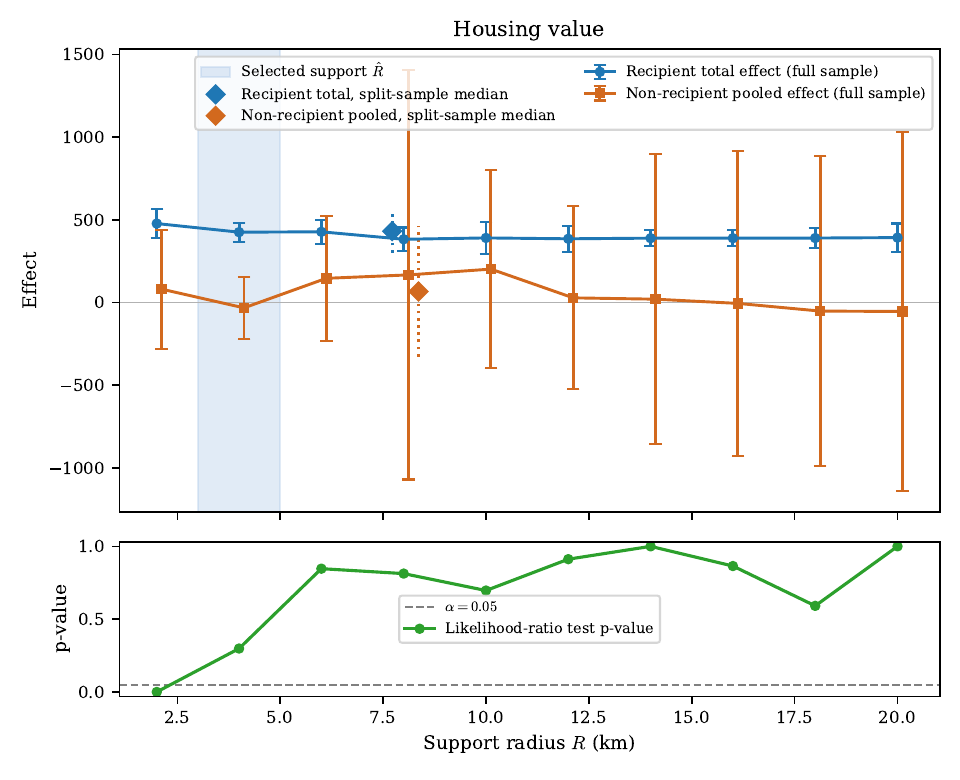}
\includegraphics[width=0.46\textwidth]{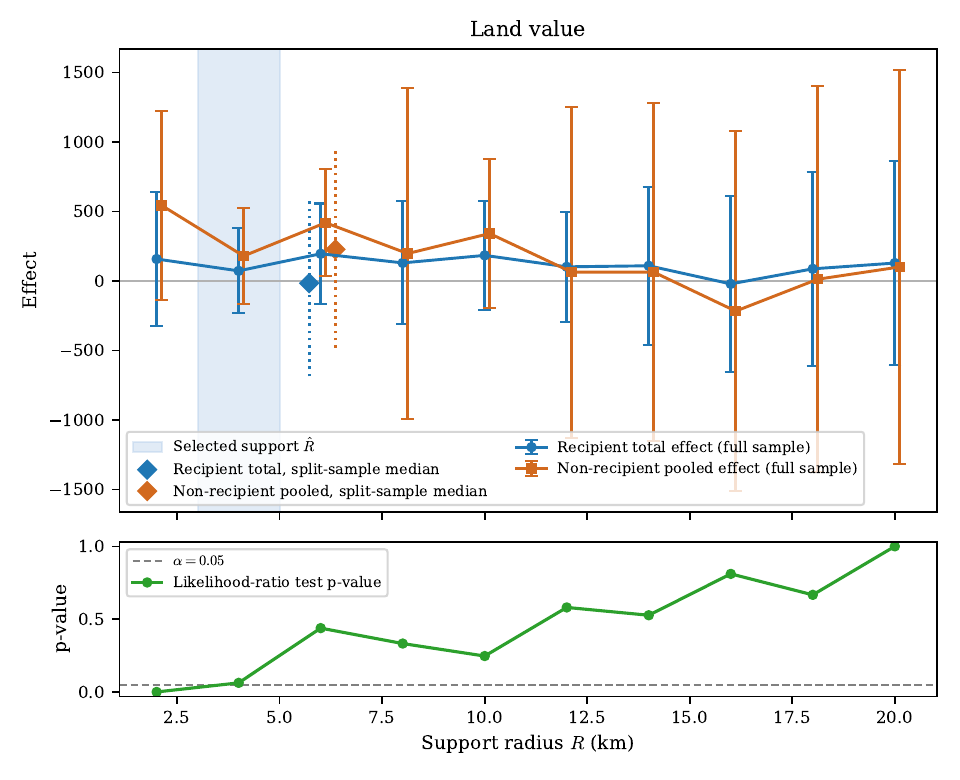}

\vspace{0.3em}

\includegraphics[width=0.46\textwidth]{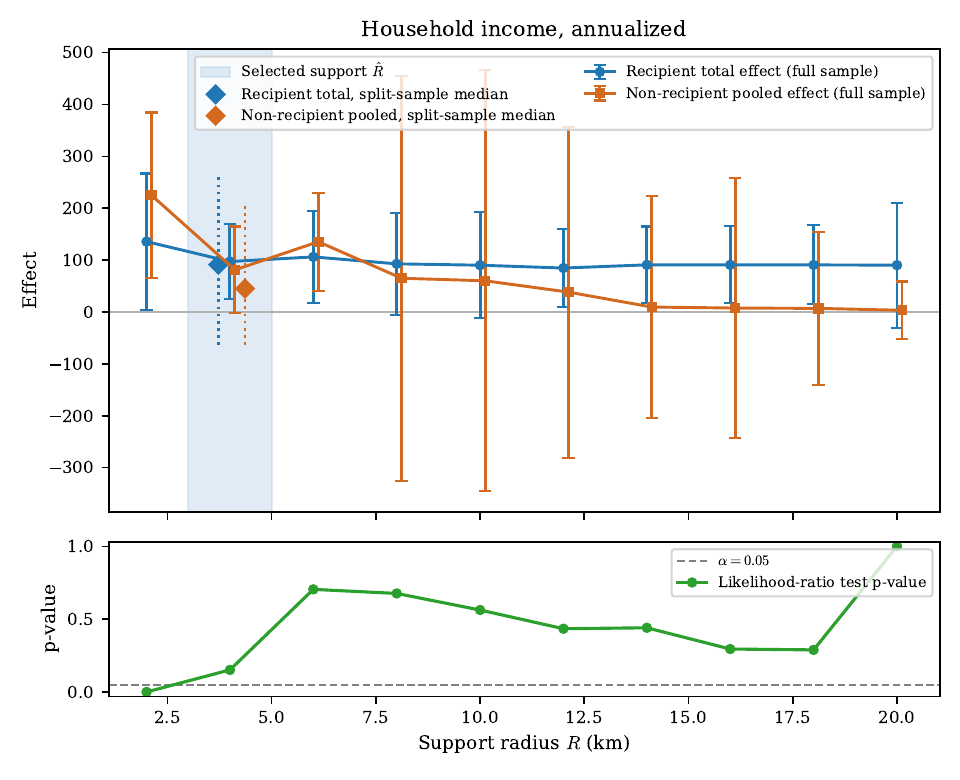}
\includegraphics[width=0.46\textwidth]{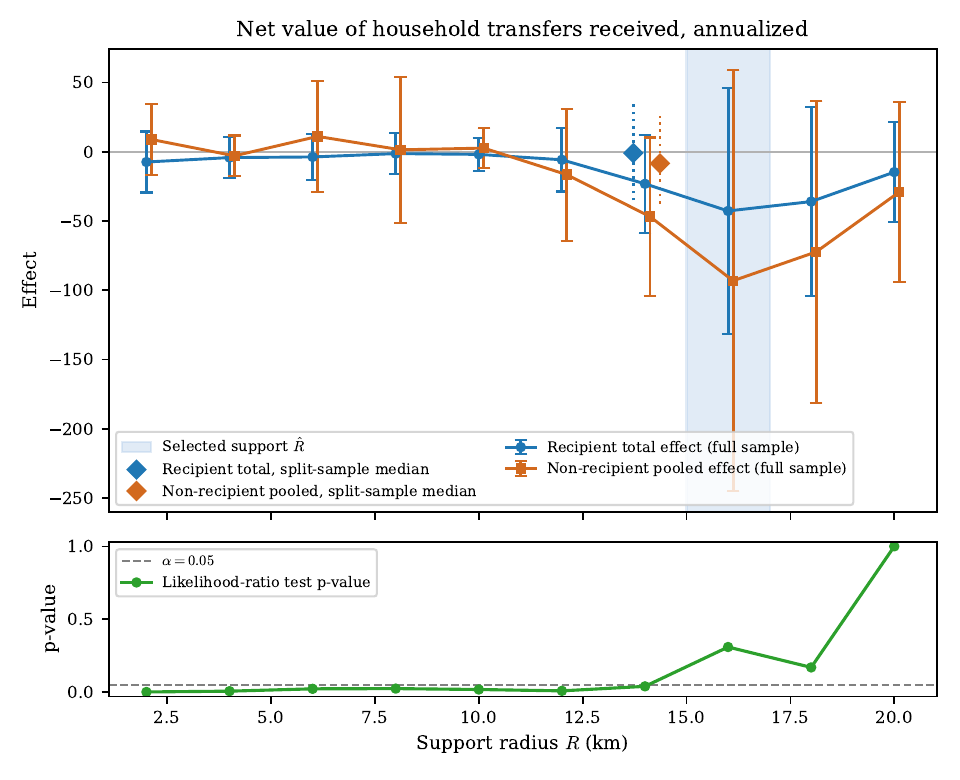}

\vspace{0.3em}

\includegraphics[width=0.46\textwidth]{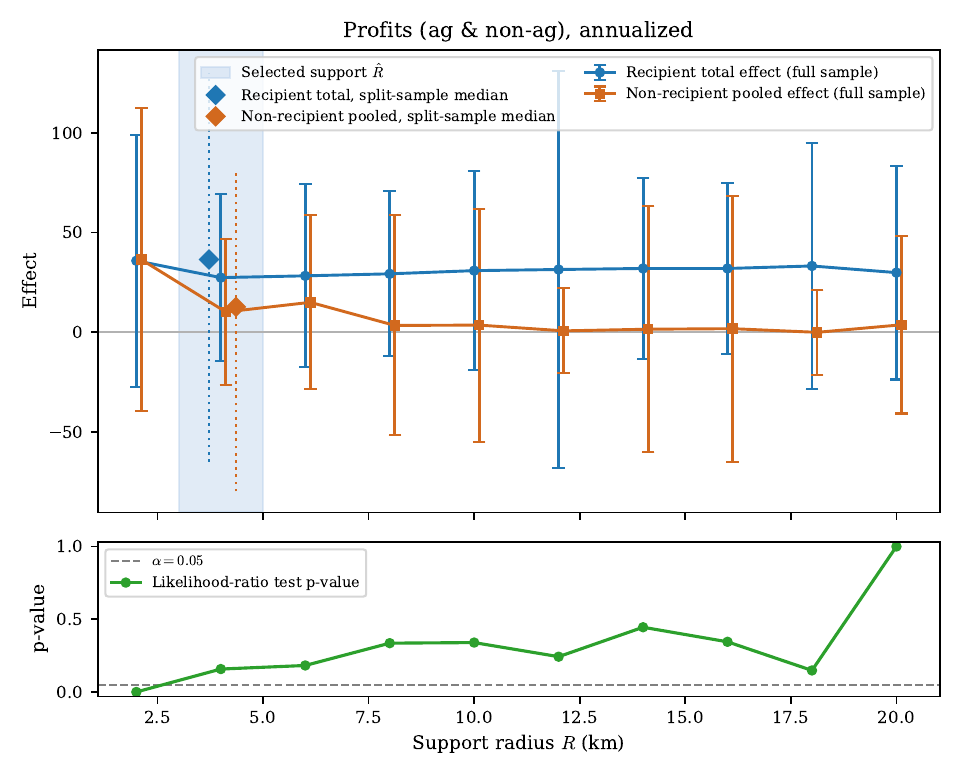}
\includegraphics[width=0.46\textwidth]{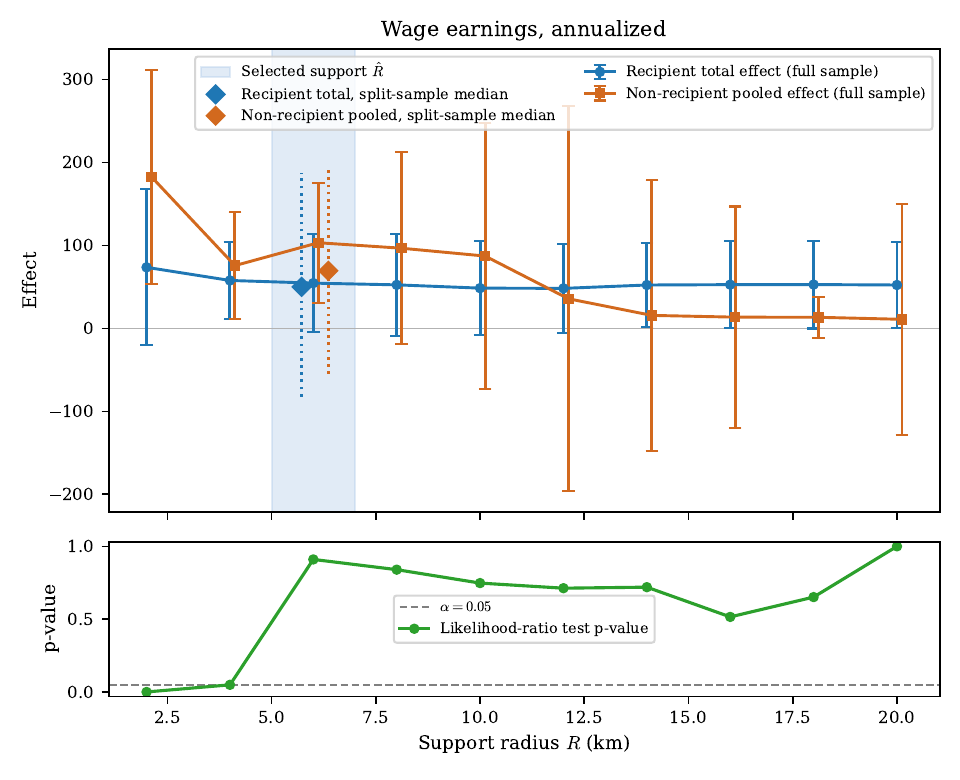}
\caption{Radius-path diagnostics for wealth, income, and transfer outcomes in the Egger application}
\label{fig:egger-radius-path-app-assets}
\end{figure}

\begin{figure}[htbp]
\centering
\includegraphics[width=0.55\textwidth]{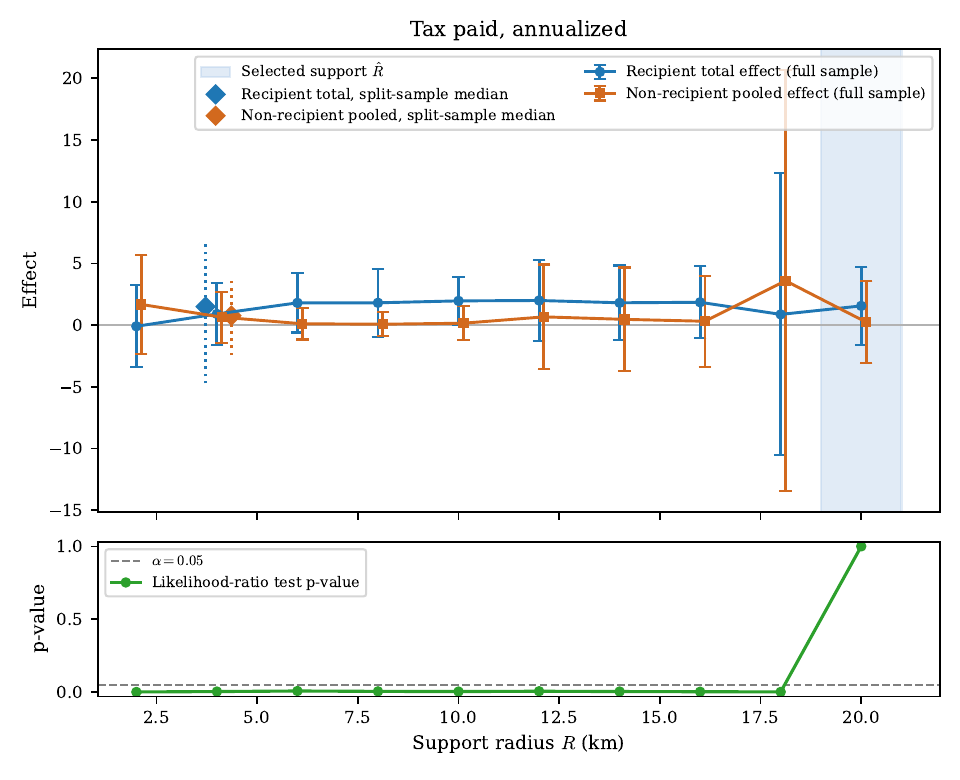}
\caption{Radius-path diagnostics for taxes paid in the Egger application}
\label{fig:egger-radius-path-app-tax}
\end{figure}

The full figure set reinforces the main-text message. \bluerev{The original 2 km support
is rejected for every outcome; the selected support is most often modestly larger,
commonly $4$--$6$ km, with a few outcomes selecting longer supports.} At the same
time, the block split checks are noisy, so we interpret the selected supports as
diagnostics for the exposure map.

\subsubsection*{D. The radius path and the resolution of the annular instruments}

The radius-path standard errors in
Figures~\ref{fig:egger-radius-path-app-exp}--\ref{fig:egger-radius-path-app-tax}
are not always increasing in the support; for a few outcomes, consumption among
them, the intervals widen sharply near $16$~km and then tighten again. This
unevenness reflects the construction of the index rather than the strength of
spillovers at those distances: the annular transfer variables at neighboring
radii are highly collinear when the rings are narrow, so the estimated weights
are only weakly determined at particular supports and the first-step uncertainty
carried into the second stage is uneven rather than smoothly increasing. The
reported estimates do not depend on this.
Table~\ref{tab:egger-radius-grid-invariance} reports the recipient total at a
common $4$~km support for annular instrument bands of $1$, $1.5$, and $2$~km;
the point estimates move by far less than one standard error.

\begin{table}[htbp]
\centering
\caption{Recipient total at a common $4$~km support, by instrument band width}
\label{tab:egger-radius-grid-invariance}
\begin{tabular}{lccc}
\toprule
& \multicolumn{3}{c}{Annular instrument band width} \\
\cmidrule(lr){2-4}
Outcome & $1$~km & $1.5$~km & $2$~km \\
\midrule
Consumption  & $296\ (59)$ & $296\ (54)$ & $296\ (62)$ \\
Assets       & $177\ (31)$ & $177\ (29)$ & $177\ (33)$ \\
Total income & $\phantom{0}97\ (37)$ & $\phantom{0}91\ (36)$ & $\phantom{0}97\ (40)$ \\
\bottomrule
\end{tabular}
\begin{flushleft}
\footnotesize Notes: Recipient total effect at a fixed $4$~km support, with the
joint delta-method standard error in parentheses, using annular instruments of
width $1$, $1.5$, and $2$~km. The index weights and the downstream IV are
otherwise as in the main specification. \ridgerev{The design-side conditional
expectations in this exhibit are estimated by unpenalized least squares on the
quadratic basis rather than the main specification's cross-validated ridge fit.}
\end{flushleft}
\end{table}


\subsubsection*{\rev{E. Multiplier accounting: reparameterization and delta-method inference}}
\label{app:egger-multiplier}

\rev{This appendix details the multiplier exercise of
Section~\ref{subsec:egger-main}. \citet{egger2022general}'s Table~V maps the
coefficients of several weighted IV regressions into component multipliers. For
component $c$, role $g$ (recipient/non-recipient), and ring $k$, we impose the
design reparameterization
\[
  \beta_{cgk}=\gamma_{cg}\,\theta_{cgk},
  \qquad \textstyle\sum_{k}\theta_{cgk}=1,
\]
so that the normalized shape $\theta_{cg}$ is exactly the Stage~1 index and
$\gamma_{cg}$ is a single slope on the constructed index $S_{icg}(\widehat\theta_{cg})$, the
$\widehat\theta_{cg}$-weighted combination of unit $i$'s ring exposures,
which enters each Table~V equation as an endogenous regressor instrumented by the
corresponding index instrument. Everything else (deflators, transfer
normalizations, weights, supports) is held fixed.}

\rev{We propagate uncertainty to the multipliers by a cluster-robust delta-method
sandwich, clustering at the sublocation level (the unit of randomization). Stacking
the free shape parameters $\phi$ (an $m_{cg}-1$ vector per outcome/role, with $m_{cg}$ the number of annuli in the selected support) and the full
Table~V coefficient vector $b$ into $p=(\phi',b')'$, write, for each sublocation
cluster $s$, the stacked influence-function contribution $\psi_s$ that combines the
Stage~1 design-GMM influence, which maps the design moments to $\widehat\phi$, with the
Table~V two-stage least squares influence; the latter includes the generated-regressor
adjustment by which the Stage~1 uncertainty in $\widehat\phi$ enters $\widehat b$. For
the vector of reported multipliers $f(p)$ with $D=\partial f/\partial p'$, the
covariance is
\[
  \widehat V_f \;=\; D\,\Big(\tfrac{C}{C-1}\sum_{s=1}^{C}\psi_s\psi_s'\Big)\,D',
  \qquad C=84.
\]
This construction never inverts the moment covariance $\widehat\Omega=\sum_s U_s U_s'$
(with $U_s$ the cluster-$s$ stacked Stage-1 and Table~V moment contribution),
and that is what keeps it well posed. The unit of independence is the sublocation
cluster, so $\operatorname{rank}(\widehat\Omega)\le C=84$, while the Table~V coefficient
vector is high-dimensional ($\approx 400$); the efficient-GMM information matrix
$(G'\widehat\Omega^{-1}G)^{-1}$, with $G=\partial\bar g/\partial p'$ the moment Jacobian,
would invert through this rank-$84$ ceiling and is ill-posed. The sandwich avoids the
inversion entirely: $\widehat\Omega$ enters only as the cluster meat, and the reported
multipliers are a low-dimensional smooth function whose variance is well defined
whatever the rank of $\widehat\Omega$. The standard errors must be read against the
cluster count rather than the household count; with $C=84$ they are comparable in
magnitude to \citet{egger2022general}'s own (Table~\ref{tab:egger-multiplier-results}),
and applying their wild cluster bootstrap (Rademacher signs by sublocation) to our
index, holding the shape weights fixed, reproduces the Table~V component to within a
few percent. The resulting design-based multipliers and their benchmarks are reported
in the main text (Table~\ref{tab:egger-multiplier-results}).}

\end{document}